\documentclass[acmsmall,10pt,screen]{acmart}
\pdfoutput=1 

\usepackage{mathpartir}
\usepackage{proof}

\usepackage{amsmath}
\usepackage[inline,shortlabels]{enumitem}
\usepackage{thmtools}
\usepackage{thm-restate}
\usepackage{bm}
\usepackage{booktabs}
\usepackage{framed}
\usepackage[utf8]{inputenc}
\usepackage{mathtools}
\usepackage{mdframed}
\usepackage{graphicx}
\usepackage{multirow}
\usepackage{stmaryrd}
\usepackage{subcaption}
\usepackage{tabularx}
\usepackage{verbatim}
\usepackage{xspace}

\usepackage{cmap}
\usepackage[T1]{fontenc}

\usepackage{wrapfig}
\usepackage[inference]{semantic}

\usepackage[textsize=tiny]{todonotes}

\usepackage[scaled=0.75]{beramono}
\usepackage{nanoml}

\newcommand{\tname}[1]{\textsc{#1}\xspace}
\newcommand{\ectalibrary}{\tname{ecta}}
\newcommand{\hplus}{\tname{Hoogle+}}
\newcommand{\tool}{\tname{Hectare}}
\newcommand{\hpp}{\tool}
\newcommand{\stackoverflow}{\tname{StackOverflow}}

\newcommand{\vs}{\emph{vs.}\xspace}
\newcommand{\ie}{\emph{i.e.}\xspace}

\newcommand{\eg}{\emph{e.g.}\xspace}

\newcommand{\circled}[1]{\raisebox{.5pt}{\textcircled{\raisebox{-.9pt} {#1}}}}
\newcommand{\staten}[1]{{\sffamily\itshape  #1}}

\newcommand{\emphbf}[1]{\emph{\textbf{#1}\xspace}}
\newcommand{\mypara}[1]{\smallskip\noindent\emphbf{#1.}\xspace}

\newcommand{\secref}[1]{\S \ref{#1}}
\newcommand{\appref}[1]{App. \ref{#1}}
\newcommand{\figref}[1]{Fig. \ref{#1}}

\newtheorem{remark}{Remark}

\newenvironment{trivproof}{\comment}{\endcomment}

\newcommand{\nat}{\mathbb{N}}
\newcommand{\many}[1]{\overline{#1}}
\newcommand{\bigmid}{\mathrel{\Big|}}
\newcommand{\restrict}[2]{#1 \big|_{#2}}
\newcommand{\s}[1]{\mathsf{#1}}
\newcommand{\twodots}{\mathinner {\ldotp \ldotp}}

\newcommand{\set}[1]{\ensuremath{\textsf{#1}}}

\newcommand{\partialfn}{\rightharpoonup}
\newcommand{\denotation}[1]{\llbracket #1 \rrbracket}
\newcommand{\dennode}[1]{\denotation{#1}^N}
\newcommand{\denedge}[1]{\denotation{#1}^E}
\newcommand{\denstate}[1]{\denotation{#1}^S}
\newcommand{\denpt}[1]{\denotation{#1}^{PT}}

\newcommand{\hastype}{\mathrel{:}}
\newcommand{\subst}[3]{[#1 / #2]#3}
\newcommand{\bigsubst}[3]{\left[#1 \middle/ #2\right]#3}
\renewcommand{\P}[1]{\mathbb{P}(#1)}

\newcommand{\closedtermsof}[1]{\mathcal{T}(#1)}
\newcommand{\termsofunode}[1]{\mathbb{T}(#1)}

\newcommand{\at}[2]{#1 |_{#2}}
\newcommand{\arity}{\ensuremath{\mathsf{arity}}}
\newcommand{\dom}{\ensuremath{\mathsf{dom}}}
\newcommand{\closure}{\ensuremath{\mathsf{cl}}}

\newcommand{\PEC}{\set{PEC}}

\newcommand{\pecsat}[2]{#2 \models #1}

\newcommand{\nodesymbol}{\boldsymbol{\mathtt{U}}}
\newcommand{\edgesymbol}{\boldsymbol{\Pi}}
\newcommand{\node}[1]{\nodesymbol(#1)}
\newcommand{\edge}[3]{\edgesymbol(#1, #2, #3)}
\newcommand{\edgeuc}[2]{\edgesymbol(#1, #2)}
\newcommand{\edgeucnull}[1]{\edgesymbol(#1)}
\newcommand{\edgebot}{\ensuremath{\edgesymbol_\bot}}
\newcommand{\nodebot}{\ensuremath{\nodesymbol_\bot}}
\newcommand{\recnode}[2]{\mu #1 . #2}
\newcommand{\unfold}[1]{\mathsf{unfold}(#1)}
\newcommand{\skeleton}[1]{\ensuremath{\mathsf{sk}(#1)}}
\newcommand{\nodesat}{\mathsf{nodes}}
\newcommand{\reduce}{\mathsf{reduce}}
\newcommand{\union}{\sqcup}
\newcommand{\bigunion}{\bigsqcup}
\newcommand{\intersect}{\sqcap}
\newcommand{\bigintersect}{\bigsqcap}
\newcommand{\recintersect}[3]{#1 \mathrel{\underset{#3}{\intersect}} #2}
\newcommand{\recintersectinner}[3]{#1 \mathrel{\underset{#3}{\tilde{\intersect}}} #2}
\newcommand{\ctx}{\mathcal{C}}
\newcommand{\fragment}[2]{\langle #1 = #2 \rangle}

\newcommand{\proj}{\mathsf{project}}
\newcommand{\finalstate}{\ensuremath{\sigma_{\bullet}}\xspace}
\newcommand{\varOf}{\mathsf{var}}

\newcommand{\unorderedpairs}[1]{\binom{#1}{2}}

\newcommand{\intersectatpath}[3]{#1 |_{#2}^{\intersect #3}}

\newcommand{\unenum}[2]{\ensuremath{\square(#1,#2)}}
\newcommand{\unenumuc}[1]{\ensuremath{\square(#1)}}
\newcommand{\stepsymb}{\longrightarrow}
\newcommand{\steps}[2]{\ensuremath{#1 \stepsymb #2}}
\newcommand{\reduces}[2]{\ensuremath{#1 \stepsymb^* #2}}

\newcommand{\gramdef}{\ensuremath{\mathrel{\mathord{:}\mathord{:=}}}}

\newcommand{\hplusAll}{\tname{HplusAll}}
\newcommand{\hppNone}{\tool-\tname{Na\"ive}}
\newcommand{\hppStatic}{\tool-\tname{StaticOnly}}
\newcommand{\hppDyn}{\tool-\tname{DynamicOnly}}

\newcommand{\nHplus}{45\xspace}

\newcommand{\nStackoverflow}{19\xspace}

\newcommand{\nTimeout}{300 seconds\xspace}
\newcommand{\nTimeoutL}{600 seconds\xspace}
\newcommand{\nComps}{291\xspace}
\newcommand{\nRepeat}{three\xspace}

\newcommand{\successRate}{88\%\xspace}
\newcommand{\successRateHplus}{64\%\xspace}
\newcommand{\nHplusHO}{44\%\xspace}
\newcommand{\polypercent}{40\%\xspace}

\newcommand{\nEctaSolved}{43\xspace}
\newcommand{\nHplusSolved}{39\xspace}

\newcommand{\nEctaSolvedSo}{13\xspace}
\newcommand{\nHplusSolvedSo}{3\xspace}
\newcommand{\speedup}{$8\times$\xspace}
\newcommand{\hplusSpeedup}{$7\times$\xspace}
\newcommand{\soSpeedup}{$40\times$\xspace}

\newcommand{\nStaticSolved}{34\xspace}
\newcommand{\nDynSolved}{36\xspace}

\acmJournal{PACMPL}
\acmNumber{ICFP}
\acmYear{2022}

\setcopyright{rightsretained}

\begin{document}

\title{Searching Entangled Program Spaces}

\author{James Koppel}
\affiliation{
  \institution{MIT}
  \city{Cambridge}
  \state{MA}
 \country{USA}
}
\email{jkoppel@mit.edu}

\author{Zheng Guo}
\affiliation{
  \institution{UC San Diego}
  \city{La Jolla}
  \state{CA}
 \country{USA}
}
\email{zhg069@eng.ucsd.edu}

\author{Edsko de Vries}
\affiliation{
  \institution{Well-Typed LLP}
 \country{United Kingdom}
}
\email{edsko@well-typed.com}

\author{Armando Solar-Lezama}
\affiliation{
  \institution{MIT}
  \city{Cambridge}
  \state{MA}
 \country{USA}
}
\email{asolar@csail.mit.edu}

\author{Nadia Polikarpova}
\affiliation{
  \institution{UC San Diego}
  \city{La Jolla}
  \state{CA}
 \country{USA}
}
\email{npolikarpova@eng.ucsd.edu}

\begin{abstract}
Many problem domains, including program synthesis and rewrite-based optimization,
require searching astronomically large spaces of programs.
Existing approaches often rely on building specialized data structures---%
version-space algebras, finite tree automata, or e-graphs---%
to compactly represent such spaces.
%
At their core, all these data structures exploit independence of subterms;
as a result, they cannot efficiently represent more complex program spaces,
where the choices of subterms are entangled.

We introduce \emph{equality-constrained tree automata} (ECTAs),
a new data structure, designed to compactly represent large spaces of programs with entangled subterms. 
We present efficient algorithms for extracting programs from ECTAs, 
implemented in a performant Haskell library, \ectalibrary.
Using the \ectalibrary library, we construct \tool, 
a type-driven program synthesizer for Haskell.
\tool significantly outperforms a state-of-the-art synthesizer \hplus---%
providing an average speedup of \speedup---%
despite its implementation being an order of magnitude smaller.

\end{abstract}




\maketitle

\section{Introduction}
\label{sec:intro}

From program synthesis to theorem proving and compiler optimizations,
a range of problem domains make use of data structures 
that compactly represent large spaces of terms.
In program synthesis, the most well-known example is \emph{version space algebras} (VSAs)~\cite{lau2003programming,polozov2015flashmeta},
the data structure behind the successful spreadsheet-by-example tool \tname{FlashFill}~\cite{gulwani2011automating}.
%
Although there may be over $10^{100}$ programs matching an input/output example, 
\tname{FlashFill} is able to represent all of them as a tiny VSA, 
efficiently run functions over every program in the space, 
and then extract the best concrete solution.

To illustrate the idea behind VSAs, 
consider the space of nine terms $\mathcal{T}=\{\s{f}(t_1) + \s{f}(t_2)\}$ where $t_1,t_2 \in \{\s{a},\s{b},\s{c}\}$.
\autoref{fig:ex1-vsa} shows a VSA that represents this space.
In a VSA a \emph{union node}, marked with $\cup$, represents a union of all its children,
while a \emph{join node}, marked with $\bowtie$, 
applies a function symbol to every combination of terms represented by its children.
You can see how, by exploiting the shared top-level structure of the terms in $\mathcal{T}$,
this VSA is able to compactly represent nine terms, each of size five, using only six nodes.

Another data structure that exploits sharing in a similar way is \emph{e-graphs},
which enjoy a wide range of applications, including theorem proving~\cite{demoura2007ematching}, 
rewrite-based optimization~\cite{tate2009equality},
domain-specific synthesis~\cite{nandi2020synthesizing,ruler},
and semantic code search~\cite{premtoon2020semantic}.
Both VSAs and e-graphs are now known~\cite{pollock_e-graphs_nodate,koppel2021version} to be equivalent 
to special cases of \emph{finite tree automata} (FTAs), 
which have independently experienced a surge of interest in recent years~\cite{adams2017restricting,wang2017synthesis,DBLP:journals/pacmpl/WangDS18}.
\autoref{fig:ex1-fta} shows an FTA that represents the same term space as the VSA in \autoref{fig:ex1-vsa}.
An FTA consists of \emph{states} (circles) and \emph{transitions} (rectangles),
with each transition connecting zero or more states to a single state.
Intuitively, FTA transitions correspond to VSA's join nodes,
and FTA states correspond to VSA's union nodes 
(although in a VSA, union nodes with a single child are omitted).
Importantly, all three data structures%
\footnote{We omit e-graphs from \autoref{fig:ex1} for space reasons, 
but also because e-graphs are typically used to represent congruence relations rather than arbitrary sets of terms, 
which makes them less relevant to our setting, as we discuss in \secref{sec:related-work}.}
thrive on spaces where terms share some top-level structure,
while their divergent sub-terms can be chosen \emph{independently} of each other.
%

\begin{figure}
  \centering
  \begin{subfigure}[t]{0.25\textwidth}
  \centering
  \includegraphics[scale=0.45]{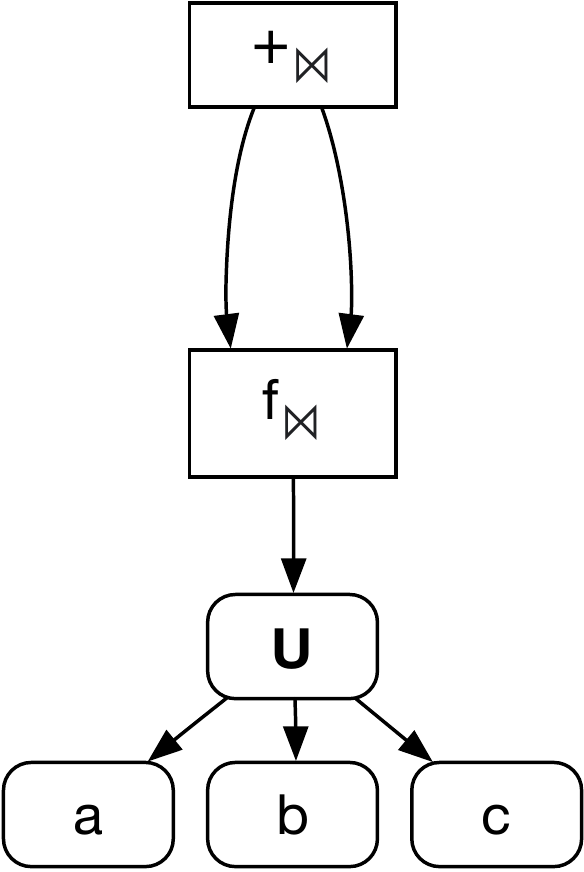}
  \subcaption{VSA for $\mathcal{T}$}
  \label{fig:ex1-vsa}
  \end{subfigure}
  ~ 
  \begin{subfigure}[t]{0.25\textwidth}
  \centering
  \includegraphics[scale=0.45]{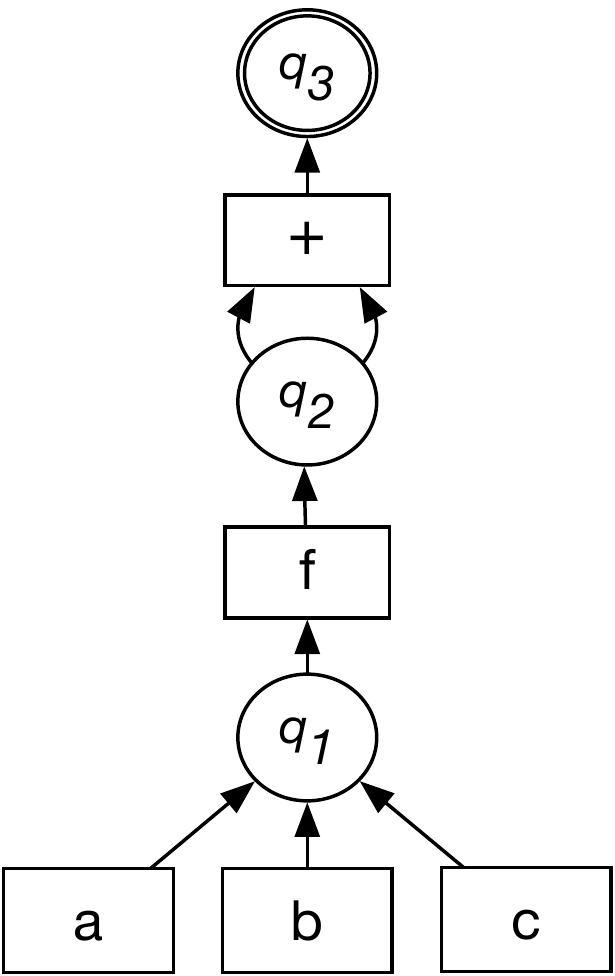}
  \subcaption{FTA for $\mathcal{T}$}
  \label{fig:ex1-fta}
  \end{subfigure}
  ~
  \begin{subfigure}[t]{0.25\textwidth}
  \centering
  \includegraphics[scale=0.45]{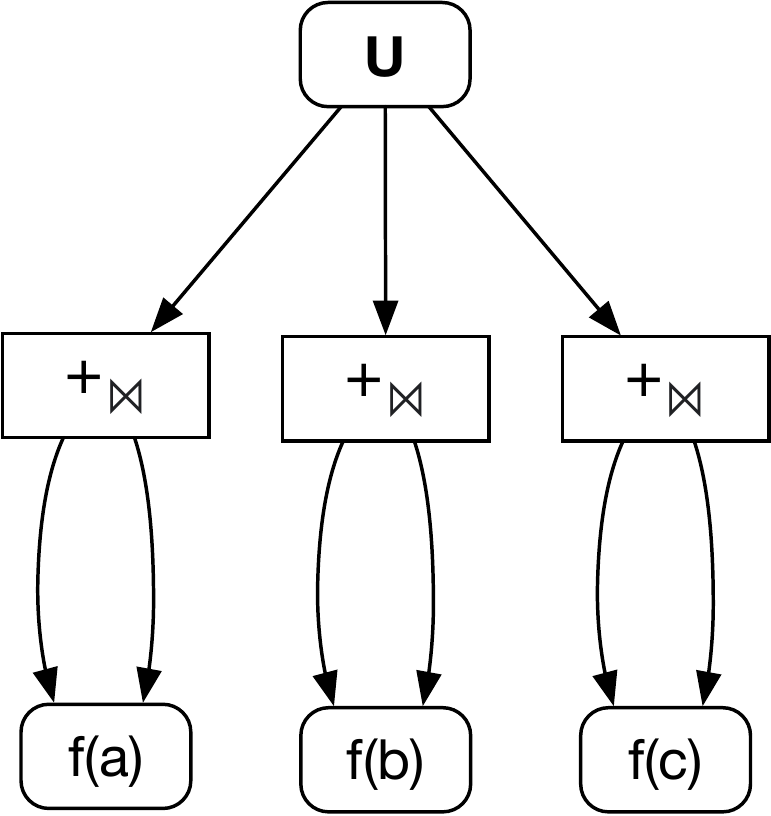}
  \subcaption{VSA for $\mathcal{U}$}
  \label{fig:ex2-vsa}
  \end{subfigure}
  ~
  \begin{subfigure}[t]{0.25\textwidth}
  \centering
  \includegraphics[scale=0.45]{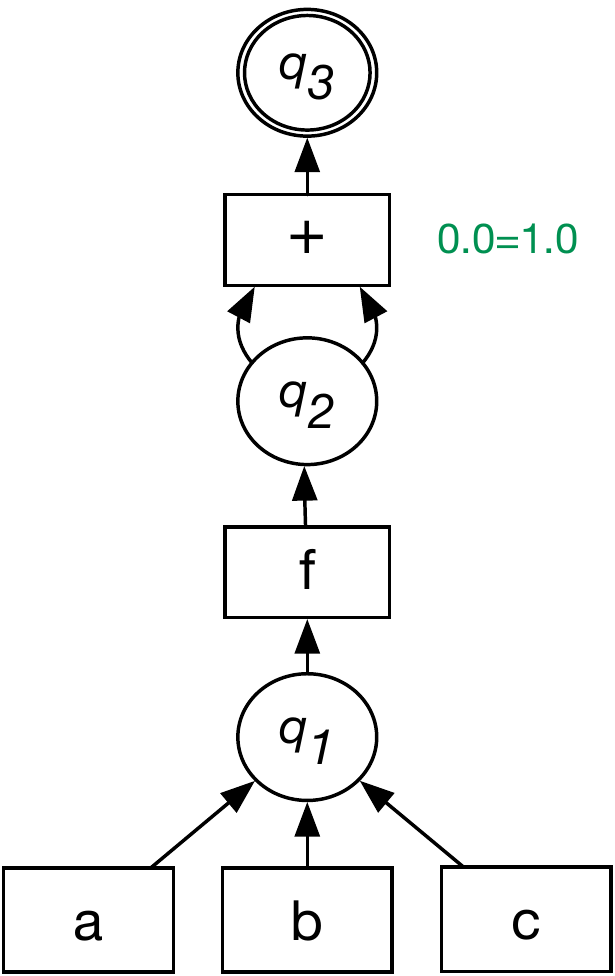}
  \subcaption{ECTA for $\mathcal{U}$}
  \label{fig:ex2-ecta}
  \end{subfigure}
  \caption{Representations of $\mathcal{T}=\{\s{f}(t_1) + \s{f}(t_2)\}$ and $\mathcal{U} = \{\s{f}(t) + \s{f}(t)\}$,
  where $t, t_1,t_2 \in \{\s{a},\s{b},\s{c}\}$.}
  \label{fig:ex1}
\end{figure}

\mypara{Challenge: Dependent Joins}
Consider now the term space $\mathcal{U} = \{\s{f}(t) + \s{f}(t)\}$, where $t \in \{\s{a},\s{b},\s{c}\}$,
that is, a sub-space of $\mathcal{T}$ where both arguments to $\s{f}$ \emph{must be the same term}.
Such ``entangled'' term spaces arise naturally in many domains. 
For example, in term rewriting or logic programming, 
we might want to represent the subset of $\mathcal{T}$ that matches the non-linear pattern $X + X$.
Similarly, in type-driven API search~\cite{Hoogle,Gissurarson2018}, 
we might want to represent the space of all \emph{types} of library functions
that unify with a given query type, such as $\T{List}\ \alpha \to \T{List}\ \alpha$.

Existing data structures are incapable of fully exploiting shared structure in such entangled spaces.
\autoref{fig:ex2-vsa} shows a VSA representing $\mathcal{U}$:
here, the node $+_{\bowtie}$ cannot be reused 
because VSA joins are \emph{independent},
whereas our example requires a dependency between the two children of $+$.
This limitation is well-known:
for example, the seminal work on VSAs~\cite{lau2003programming} notes that
``efficient representation of non-independent joins remains an item for future work.''

\mypara{Solution: ECTA}
To address this limitation,
we propose a new data structure we dub \emph{equality-constrained tree automata} (ECTAs).
ECTAs are tree automata whose transitions can be annotated with \emph{equality constraints}.%
\footnote{This might remind some readers of Dauchet's \emph{reduction automata};
we postpone a detailed comparison to related work (\secref{sec:related-work}).}
For example, \autoref{fig:ex2-ecta} shows an ECTA that represents the term space $\mathcal{U}$.
It is identical to the FTA in \autoref{fig:ex1-fta} save for the constraint $0.0 = 1.0$ on its $+$ transition.
This constraint restricts the set of terms accepted by the automaton
to those where the sub-term at path $0.0$ (the first child of the first child of $+$)
equals the sub-term at path $1.0$ (the first child of the second child of $+$).
The constraint enables this ECTA to represent a dependent join
while still fully exploiting shared structure, unlike the VSA in \autoref{fig:ex2-vsa}.

\mypara{Challenge: Enumeration}
Being able to represent a term space is not particularly useful
unless we also can efficiently \emph{extract} a concrete inhabitant of this space---%
or, more generally, \emph{enumerate} some number of its inhabitants.
Unsurprisingly, equality constraints make enumeration harder,
since the terms must now comply with those constraints
(in fact, as we demonstrate in \secref{sec:applications:sat},
extracting a term for an ECTA is at least as hard as SAT solving).
A na\"ive fix is to filter out spurious (constraint-violating) terms after the fact,
but such ``rejection sampling'' can be extremely inefficient.

\mypara{Solution: Dynamic and Static Reduction}
Our first insight for how to speed up enumeration is inspired by constraint-based type inference.
Instead of making an \emph{eager} choice at a constrained state, such as $q_1$ in \autoref{fig:ex2-ecta},
our enumeration technique \emph{postpones} this choice, 
instead introducing a ``unification variable'' $V_1$ to stand for the chosen term.
This variable gets reused the second time $q_1$ is visited.
At the end, $V_1$ is reified into a concrete term,
thereby making a simultaneous choice at the two constrained states,
which is guaranteed by construction to satisfy all equality constraints.
We dub this mechanism \emph{dynamic reduction},
where ``dynamic'' refers to operating during the enumeration process.
As we illustrate in \secref{sec:overview},
dynamic reduction becomes more involved 
when equality constraints relate different states:
in that case the term space associated with a unification variable gets refined during enumeration.

Our second insight is that enumeration can often be made even more efficient
by transforming the ECTA \emph{statically}---that is, before the enumeration starts---%
so that some of its constraints are ``folded'' into the structure of the underlying FTA. 
We will present examples in \secref{sec:overview} of using static reduction to ``prune'' away states 
that cannot be part of any term that satisfies the constraints.

%
%
%

\mypara{Contributions}
In summary, this paper makes the following contributions:
\begin{enumerate}
\item We introduce the \emph{ECTA data structure} (\secref{sec:acyclic}), 
which supports compact representation of program spaces with dependent joins,
as well as efficient enumeration (\secref{sec:enumeration}) via \emph{static} and \emph{dynamic reduction}.
We first formalize the simpler acyclic ECTAs, 
and then show how to add cycles in order to support infinite term spaces (\secref{sec:cyclic}).
\item We develop \emph{ECTA encodings} for two diverse domains:
Boolean satisfiability and type-driven program synthesis (\secref{sec:applications}).
These encodings illustrate that ECTAs are expressive and versatile,
and that ECTA enumeration can effectively be used as a general-purpose constraint solver.
\item We implement the data structure and its operations in a performant Haskell library, \ectalibrary.
\end{enumerate}

We evaluate the \ectalibrary library
on the domain of type-driven program synthesis (\secref{sec:eval}).
The experiments show that our ECTA-based synthesizer \tool significantly outperforms
its state-of-the-art competitor \hplus~\cite{guo2020tygar},
despite our implementation being \emph{only a tenth of the size}.
Specifically, \tool is able to solve \successRate of synthesis problems 
in the combined benchmark suite compared to only \successRateHplus by \hplus,
and on commonly solved benchmarks \tool is \emph{\speedup faster} on average. 
Further, our evaluation demonstrates that static and dynamic reduction are critical for performance:
ablating either of those mechanisms reduces the number of benchmarks solved,
while a na\"ive baseline that uses ``rejection sampling'' enumeration
is unable to solve \emph{any} benchmarks.

\section{ECTA by Example}\label{sec:overview}

In this section we illustrate the ECTA data structure and its two major features---%
static and dynamic reduction---%
using the problem of type-driven program synthesis as a motivating example. 
We give a simple encoding of the space of well-typed small programs into ECTAs, 
and then show how the \emph{general-purpose ECTA operations} are used to efficiently enumerate the well-typed terms.
We will present the full encoding, which also handles arbitrary polymorphism and higher-order functions, in \secref{sec:applications},
along an with ECTA encoding of another problem domain.

\subsection{Representing Spaces of Well-typed Terms}

\begin{figure}[t]
  \centering
  \begin{subfigure}[t]{0.5\textwidth}
  \centering
  \includegraphics[scale=0.45]{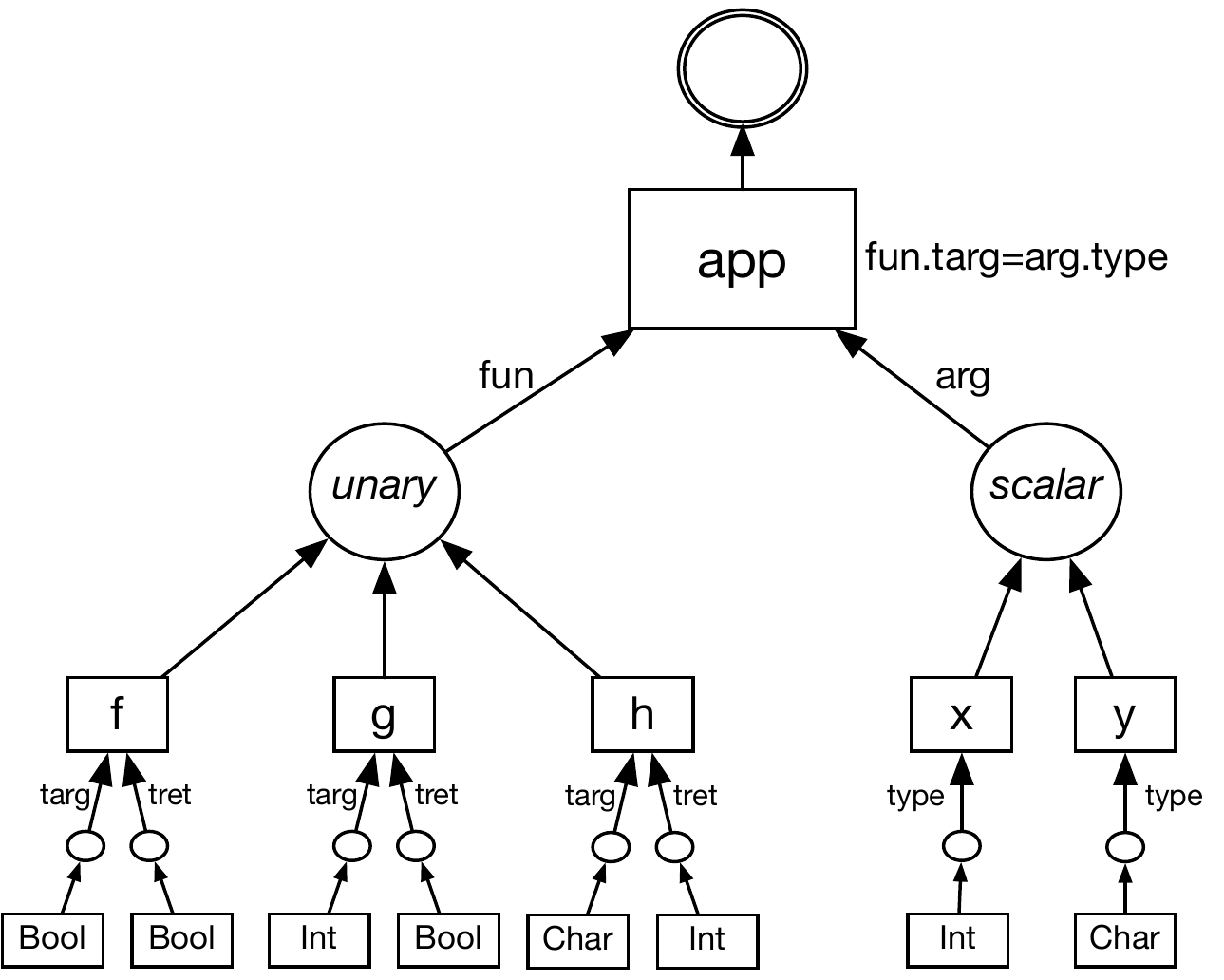}
  \subcaption{Initial ECTA.}
  \label{fig:overview-ex1-orig}
\end{subfigure}
~ 
\begin{subfigure}[t]{0.5\textwidth}
  \centering
  \includegraphics[scale=0.45]{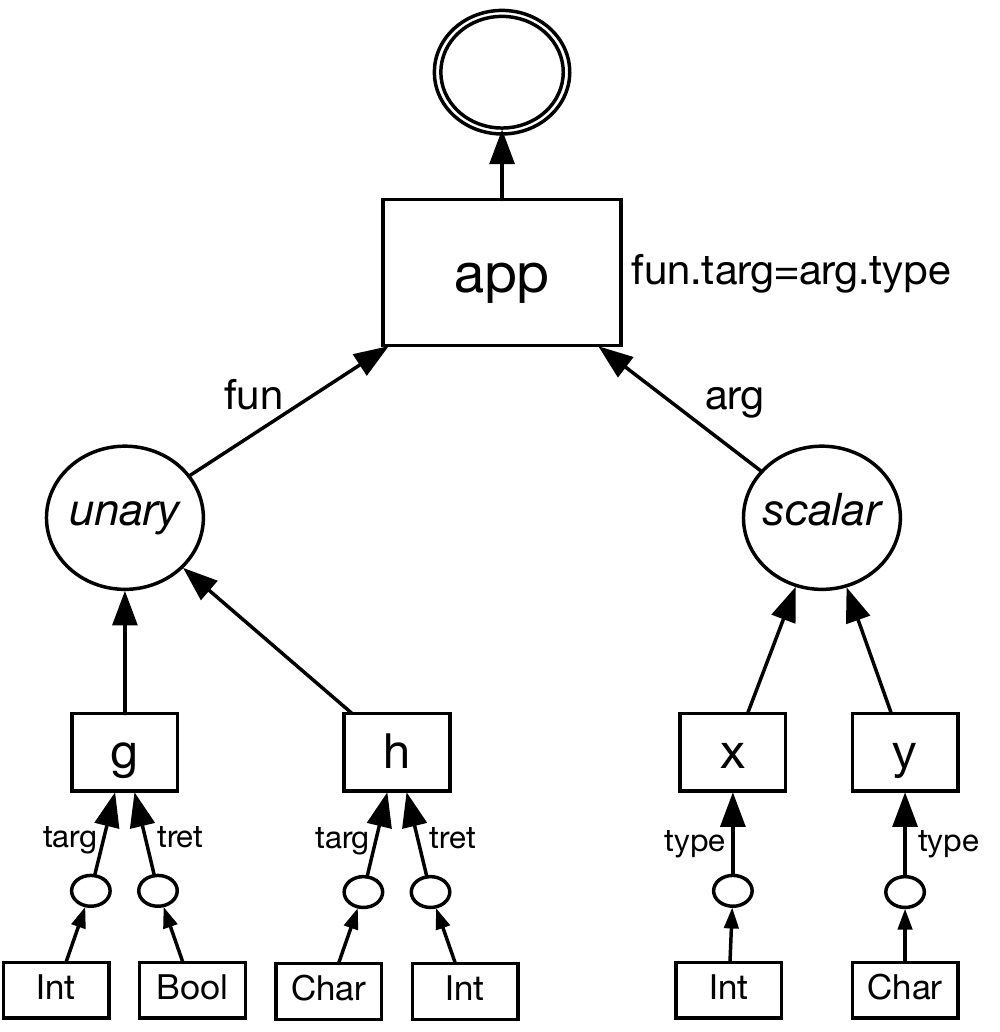}
  \caption{Reduced ECTA.}
  \label{fig:overview-ex1-reduced}
  \end{subfigure}
\caption{ECTAs representing all well-typed size-two terms in the environment $\Gamma_1$.}  
\label{fig:overview-ex1}
\end{figure}

Consider a typing environment 
$\Gamma_1 = \{x: \T{Int}, y: \T{Char}, f: \T{Bool}\to\T{Bool},  g: \T{Int}\to\T{Bool},  h: \T{Char}\to\T{Int}\}$.
Suppose we are interested in enumerating all application terms that are well-typed in $\Gamma_1$;
for now let us restrict our attention to terms of size two---%
that is, applications of variables to variables.
The space of all such terms can be compactly represented with an ECTA, 
as shown in \autoref{fig:overview-ex1-orig}.

This ECTA has a \emph{transition} for each variable in $\Gamma_1$;
scalar variables ($x$ and $y$) are annotated with their \T{type},
while functions ($f$, $g$, and $h$) are annotated with an argument type \T{targ} and a return type \T{tret}.
The \emph{node} (state) \staten{unary} represents the space of all unary variables,
while the node \staten{scalar} represents the space of all scalars.
The accepting node has a single incoming transition \T{app},
which represents an application of a unary \T{fun} to a scalar \T{arg}, fulfilling the restriction to size-two terms.%
\footnote{In our full encoding in \secref{sec:applications} we remove the distinction between the terms of different arity
in order to support higher-order and partial applications.}

While the underlying tree automaton of this ECTA (its \emph{skeleton})
accepts all terms of the form $A\ B$ where $A \in \{f,g,h\}$ and $B \in \{x,y\}$;
most of these terms, such as $f\ x$ are ill-typed.
In order to restrict the set of represented terms to only well-typed ones,
there is an \emph{equality constraint} \T{fun.targ = arg.type} attached to the \T{app} transition,
which demands that the types of the formal and the actual arguments coincide.
Thanks to this constraint, the full ECTA accepts only the two well-typed terms, $g\ x$ and $h\ y$.
(Note that in this presentation, we give names to the incoming edges of each transition
to make the constraints more readable;
in the formalization, we instead use indices to refer to the edges.)

\subsection{Static Reduction}

How would one go about enumerating the terms represented by the ECTA in \autoref{fig:overview-ex1-orig}?
%
A na\"ive approach is to 
\begin{enumerate*}
\item enumerate all terms represented by its skeleton and
\item filter out those terms that violate the constraint.
\end{enumerate*}
Step 1 is easily accomplished via depth-first search,
starting from the root (the accepting node) and picking a single incoming transition for every node.
This approach is, however, inefficient:
it ends up constructing six terms, only to filter out four of them.
In ECTA terminology, the skeleton admits six \emph{runs},
four of which are \emph{spurious} (violate the constraints).

Our \textbf{first key insight}
is that enumeration can often be made more efficient
by transforming the ECTA's skeleton so as to reduce the number of spurious runs.
We refer to this transformation as \emph{static reduction}
(because it happens once, \emph{before} the enumeration starts).
The reduced ECTA for our example is given in \autoref{fig:overview-ex1-reduced}.
Intuitively, we were able to eliminate the $f$ transition entirely
because there are no scalar variables that match its formal argument type \T{Bool};
as a result the reduced ECTA contains only two spurious runs instead of four.

More formally, static reduction works via automata \emph{intersection}.
For the ECTA in \autoref{fig:overview-ex1-orig}, 
reducing the constraint \T{fun.targ = arg.type} involves constructing an automaton 
that accepts all terms reachable via the path \T{arg.type}---%
namely \T{Int} and \T{Char}---%
and intersecting it with each node at the path \T{fun.targ}. 
Since the child node of \T{f} labeled \T{targ} represents only \T{Bool}, 
the intersection for that child is empty, 
meaning the \T{f} transition can never be used to satisfy the constraint, 
and hence can be eliminated.
The reduction algorithm performs a similar intersection for \T{g} and \T{h}, 
as well as (in the other direction) \T{x} and \T{y}, 
but finds that each of these other choices could be part of a satisfying run, 
and eliminates no further transitions.

\subsection{Type-Driven Program Synthesis with ECTAs}

\begin{figure}[t]
  \centering
  \begin{minipage}{0.5\textwidth}
  \includegraphics[scale=0.4]{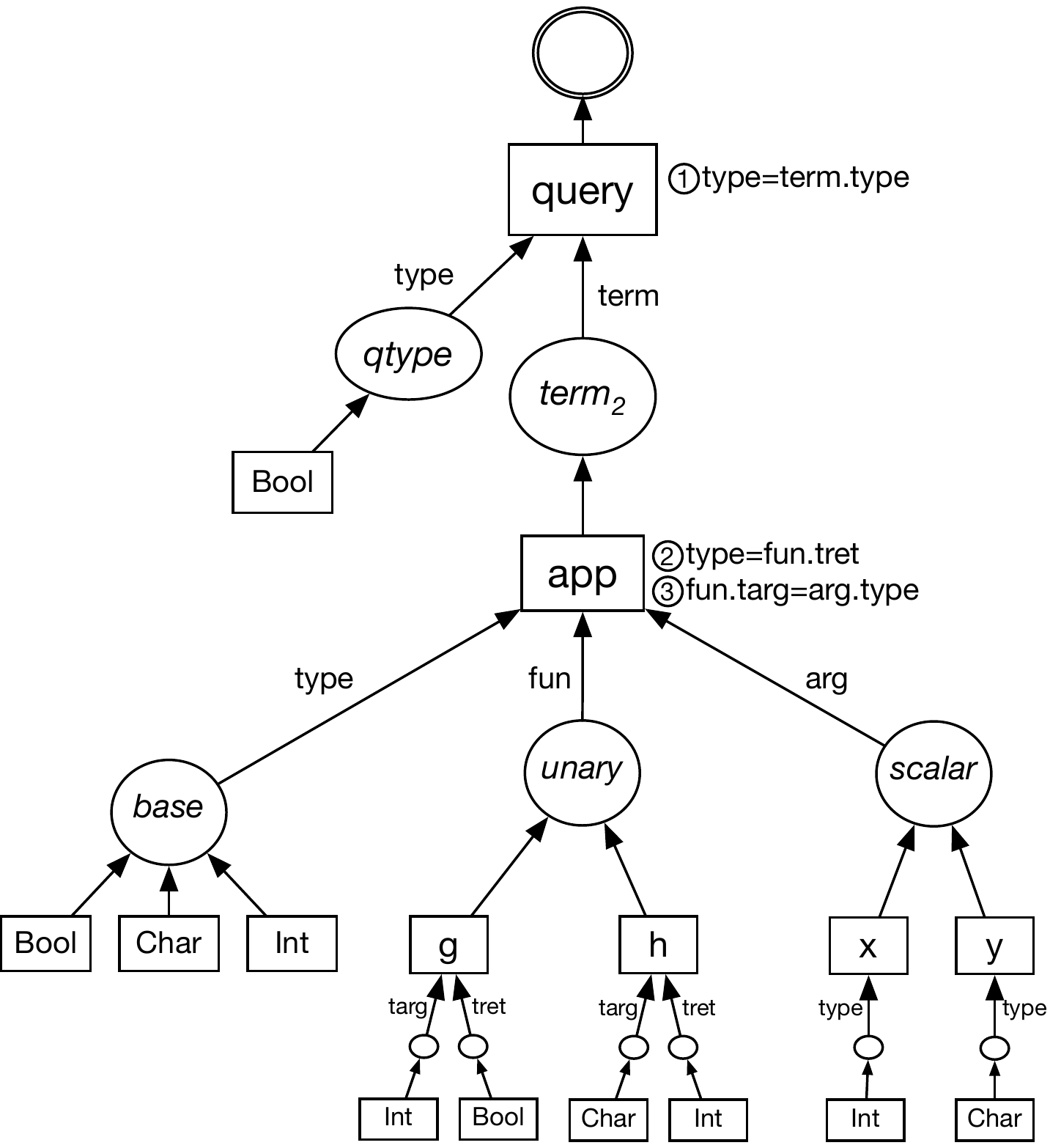}
  \end{minipage}%
  \begin{minipage}{0.5\textwidth}
  \includegraphics[scale=0.4]{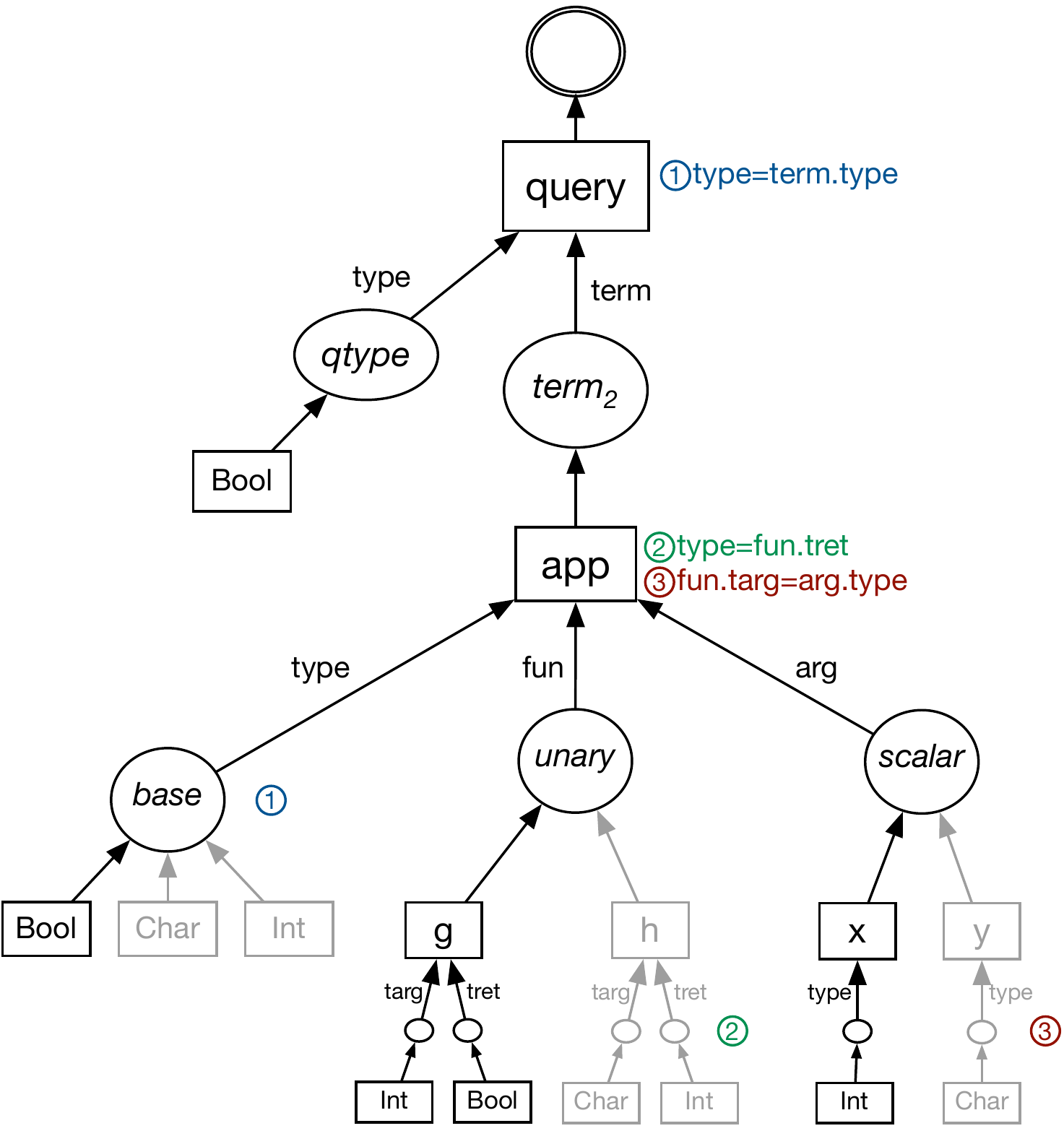}
  \end{minipage}  
\caption{ECTAs representing size-two terms of type \T{Bool}. 
The diagram on the right shows how a sequence of static reductions on the constraints \circled{1}, \circled{2}, and \circled{3}
eliminates the grayed-out transitions and nodes.}  
\label{fig:overview-ex2}
\end{figure}

In type-driven program synthesis, we are typically not interested in all well-typed terms,
but rather terms of a given \emph{query} type.
The ECTA in \autoref{fig:overview-ex2} (left) represents a type-driven synthesis problem
with the same environment $\Gamma_1$ as before and query type \T{Bool}.
The main difference between this automaton and the one in \autoref{fig:overview-ex1-reduced}
is the new transition \T{query}, whose \T{type} edge encodes the given query type
and whose \T{term} edge connects to the node representing all well-typed terms in the search space.
To filter out the terms of undesired types,
constraint \circled{1} prescribes that the \T{term}'s \T{type} be equal to the query type.
In order for this constraint to make sense,
we also add a \T{type} annotation to the \T{app} transition;
the type of an application is initially undetermined (can be any base type),
but is restricted by a new constraint \circled{2} to coincide with the return type of the function.

\autoref{fig:overview-ex2} (right) demonstrates a sequence of static reductions
that happens to eliminate \emph{all} spurious run of this ECTA,
until its skeleton represents the sole solution to the synthesis problem: the term $g\ x$.
First, reducing constraint \circled{1} eliminates all possible types of the application except \T{Bool};
next, reducing  \circled{2} eliminates the function $h$ as it has a wrong return type;
finally, reducing \circled{3} eliminates the argument $y$, 
since it is incompatible with the only remaining function $g$.

\subsection{Dynamic Reduction}\label{sec:overview:dynamic}

In the previous example, static reduction was able to eliminate all spurious runs of the ECTA before enumeration.
so that no spurious runs remained.
This is not always possible.
Consider a slightly more involved version of type-driven synthesis where functions can be polymorphic.
Specifically let $\Gamma_2 = \{x: \T{Int}, y: \T{Char}, g: \forall\alpha . \alpha \to \alpha, h: \T{Char}\to\T{Bool}\}$,
and suppose we are interested in all size-two terms of types $\T{Int}$ or $\T{Bool}$.
This problem can be represented by the ECTA in \autoref{fig:overview-ex3},
which is similar to the one in \autoref{fig:overview-ex2}.
The only interesting difference is how the polymorphic type of $g$ is represented:
the type variable $\alpha$ is encoded as a \emph{union} of all types it can unify with%
---here \T{Bool}, \T{Char}, and \T{Int}.%
\footnote{Here we consider a limited form of polymorphism, where type variables can be instantiated only with base types;
this restriction is relaxed in \secref{sec:applications}.}
Crucially, the constraint \circled{4} on $g$ guarantees that \emph{the same} type is used to instantiate both occurrences of $\alpha$.

\begin{figure}[t!]
  \begin{minipage}{0.5\textwidth}
    \includegraphics[scale=0.4]{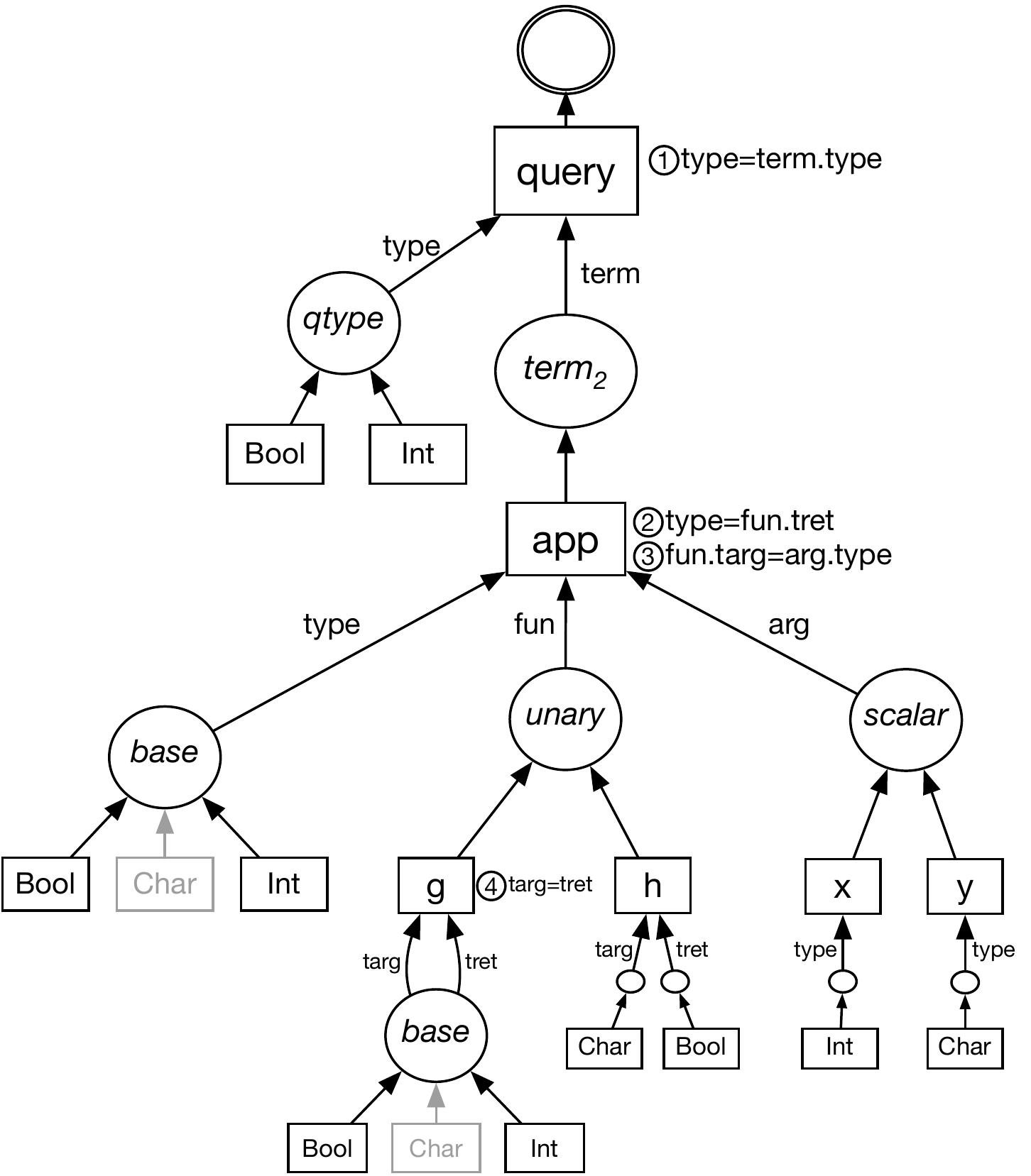}
    \end{minipage}%
    \begin{minipage}{0.5\textwidth}
    \includegraphics[scale=0.4]{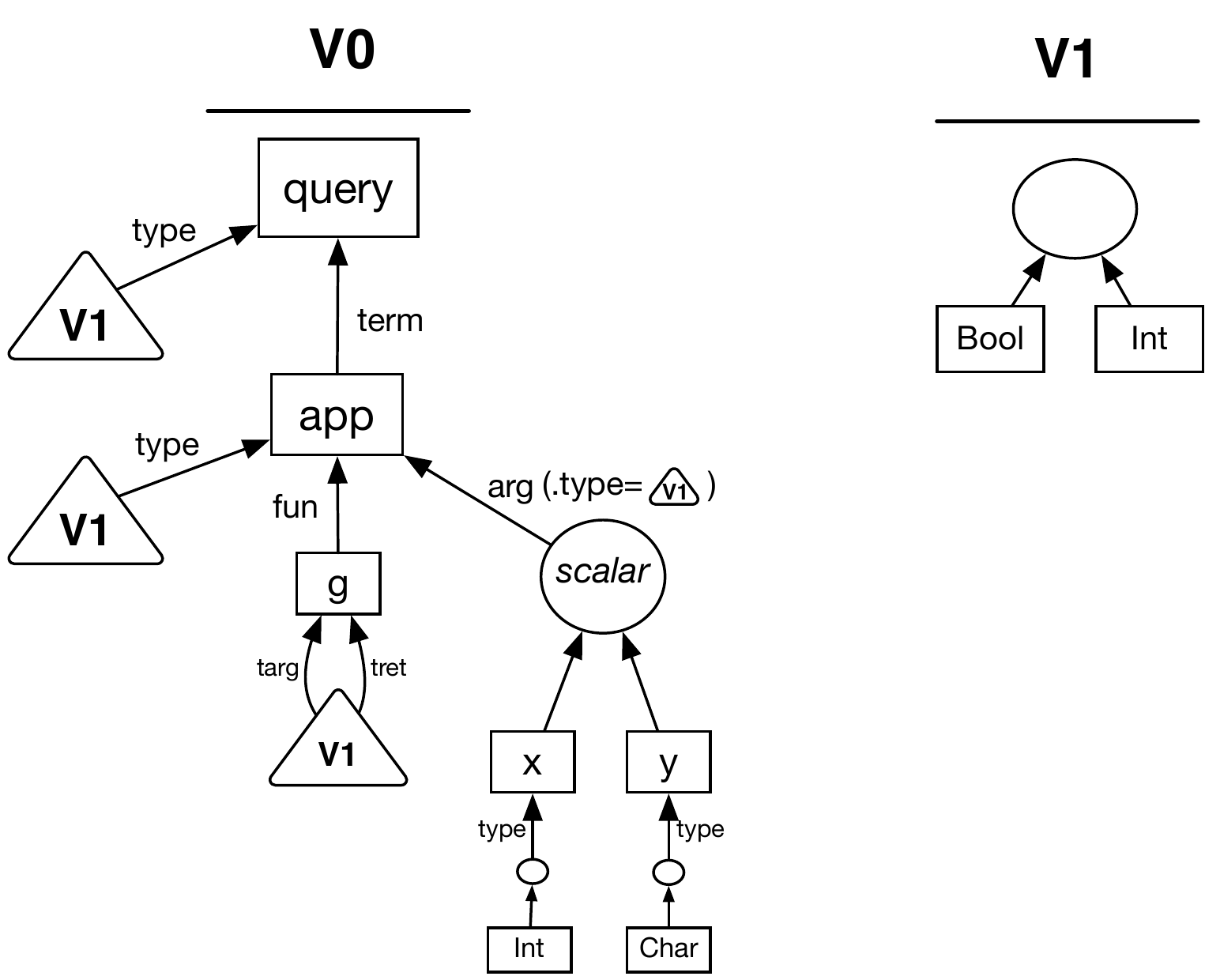}
    \end{minipage}  
  \caption{Type-driven synthesis with polymorphic functions. 
  Left: ECTA representing all terms of types \T{Bool} or \T{Int} in $\Gamma_2$. Grayed-out transitions are eliminated by static reduction.
  Right: Intermediate state during enumeration. The choice of the query type has been suspended into an auxiliary automaton $V_1$.}
  \label{fig:overview-ex3}
\end{figure}

Although static reduction can eliminate some of the transitions in this ECTA (shown in gray),
a fair number of spurious runs remain.
For example, a na\"ive left-to-right enumeration 
would first pick \T{Bool} as the query type and $g$ as the function 
(forcing the selection of $\T{Bool}\to\T{Bool}$ as the type of $g$),
only to discover later that there is no argument of type $\T{Bool}$.
More generally, in the presence of constraints,
the choices made for constrained nodes are not independent,
and making a wrong combination of choices early on (such as \T{Bool} and $g$ in our example)
may lead to expensive backtracking further down the line.

Our \textbf{second key insight} is that such backtracking can be avoided
by \emph{deferring} the enumeration of constrained nodes until more information is available.
\autoref{fig:overview-ex3} (right) illustrates this idea.
It depicts a \emph{partially enumerated} term from the ECTA on the left.
You can think of a partially enumerated term as a tree fragment at the top
with yet-to-be-enumerated ECTAs among the branches.
Importantly, because the node \staten{qtype} is constrained (by \circled{1}),
it is not enumerated eagerly, but instead \emph{suspended} into a named sub-automaton $V_1$.
As the enumeration encounters each other node $N$ 
constrained to be equal to \staten{qtype} (via \circled{1}, \circled{2}, and \circled{4}), 
$n$ is replaced by a reference to $V_1$, 
while $V_1$ is updated to $V_1 \sqcap N$.
Thus, to arrive at \autoref{fig:overview-ex3} (right),
the enumeration has made a single decision---picking $g$ over $h$---%
whereas all the other choices have been deferred.

Finally, the enumeration picks $x$ among the two \staten{scalar}s.
The \T{type} state of $x$---let's call it $N_x$---represent a singleton $\{\T{Int}\}$
\emph{and} is constrained to equal $V_1$ (by \circled{3}).
As a result, $V_1$ gets intersected with $N_x$, 
eliminating its \T{Bool} alternative.
Now when it comes time to ``unsuspend'' $V_1$,
it only contains a single alternative, \T{Int},
which is already guaranteed to be consistent with all constraints.
In other words, we have found the solution $g\ x$%
\footnote{The only other solution to this synthesis problem is $h\ y$,
which is discovered after backtracking and picking $h$ over $g$.}
without having to explicitly search over all possible query types, result types of the application, or instantiations of $g$;
instead all these three choices were made simultaneously and consistency.
%
%
We refer to this mechanism as \emph{dynamic reduction}
because it reduces the number of explored spurious runs \emph{during} enumeration. 
%

\section{Acyclic ECTA}
\label{sec:acyclic}

This section formalizes the ECTA data structure and its core algorithms. 
We begin by presenting the special case of ECTAs without cycles, 
which simplifies both the theory and implementation.
Proofs of all theorems omitted from this and the following sections can be found in \appref{app:proofs}.

\subsection{Preliminaries}

We first present standard definitions of terms, paths, and the prefix-free property
from the term-rewriting literature.

\mypara{Terms}
A \emph{signature} $\Sigma$ is a set of function symbols,
each associated with a natural number by the $\arity$ function.
%
$\closedtermsof{\Sigma}$ denotes the set of \emph{terms} over $\Sigma$,
defined as the smallest set containing all $s(t_0,\dots,t_{k-1})$
where $s \in \Sigma$, $k = \arity(s)$, and $t_0,\dots,t_{k-1}\in\closedtermsof{\Sigma}$.
We abbreviate nullary terms of the form $s()$ as $s$.

\mypara{Paths}
Paths are used to denote locations inside terms.
Formally, a \emph{path} $p$ is a list of natural numbers $i_1.i_2. \dots . i_k \in \nat^*$. 
The empty path is denoted $\epsilon$,
and $p_1.p_2$ denotes the concatenation of paths $p_1$ and $p_2$.
We write $p_1\sqsubseteq p_2$ if $p_1$ is a prefix of $p_2$
(and $p_1\sqsubset p_2$ if it is a proper prefix).
A set $P$ of paths is \emph{prefix-free} if there are no $p_1, p_2 \in P$ 
such that $p_1 \sqsubset p_2$.

Given a term $t\in \closedtermsof{\Sigma}$,
a \emph{subterm of $t$ at path $p$}, written $\at{t}{p}$, is inductively defined as follows:
\begin{enumerate*}[(i)]
    \item $\at{t}{\epsilon} = t$
    \item $\at{s(t_0, \dots, t_{k-1})}{i.p} = \at{t_i}{p}$ if $i < k$ and $\bot$ otherwise.
\end{enumerate*}
%
For example, for $t=+(f(a),f(b))$: $\at{t}{0.0} = a$, $\at{t}{1.0}=b$ and $\at{t}{2.0}=\bot$.

\subsection{Path Constraints and Consistency}

The difference between ECTAs and conventional tree automata is the presence of path equalities, 
such as $0.0 = 1.0$ in \autoref{fig:ex2-ecta}.
We now formalize the semantics of these path equalities over terms,
before using them to define the ECTA data structure.
In the following, we are interested in equalities between an arbitrary number $n > 0$ of paths rather than just two paths;
we refer to such $n$-ary constraints as \emph{path equivalences classes} (PECs).

\begin{definition}[Path Equivalence Classes]
A \emph{path equivalence class} (PEC) $c$, 
is a set of paths. 
We write a PEC $\{p_1,p_2,\dots,p_n\}$ as $\{p_1=p_2=\dots=p_n\}$.
\end{definition}

\noindent
Intuitively, the constraint $0.0=1.0$ is satisfied on a term $t$ if $\at{t}{0.0}=\at{t}{1.0}$; 
this notion generalizes straightforwardly to non-binary PECs:
\begin{definition}[Satisfaction of a PEC, Value at a PEC]
A path equivalence class $c=\{p_1=\dots=p_n\}$ is \emph{satisfied} on a term $t$
if there is some $t'$ such that, $\forall p_i \in c, \at{t}{p_i}=t'$. 
We write $\pecsat{c}{t}$ if this condition holds, and $\at{t}{c}$ to denote this unique $t'$.
\end{definition}

\noindent 
Finally, we discuss sets of PECs, called \textit{path constraint sets} (PCSs):

\begin{definition}[Path Constraint Sets, Satisfaction, Consistency]
A \emph{path constraint set} $C=\{c_1, \dots, c_m\}$ is a set of disjoint path equivalence classes.%
\footnote{Any set of PECs can be \emph{normalized} into a PCS by merging non-disjoint PECs;
for example, the set $\{\{0=1\}, \{1=2\}\}$ can be normalized into $\{\{0=1=2\}\}$.
In the following, we assume that the results of all PCS operations (\eg $C_1 \cup C_2$) are always implicitly normalized.}
A term $t$ satisfies $C$, written $\pecsat{C}{t}$, if $\forall c\in C, \pecsat{c}{t}$. 
If there exists a $t$ such that $\pecsat{C}{t}$, then $C$ is \emph{consistent}; otherwise, it is \emph{inconsistent}.
\end{definition}

\noindent
We are interested in detecting inconsistent PCSs
because ECTA operations can use this property to prune empty subautomata.
For a single PEC $c$, consistency is rather straightforward:
$c$ is consistent iff it is prefix-free.%
\footnote{Technically, we must also ensure that $\forall i\in c . i < \max_{s\in\Sigma}\arity(s)$, but this is trivially maintained by all ECTA operations.}
A non-prefix-free PEC, such as $1.0.0 = 1$, requires a term to be equal to its subterm,
which is impossible since terms are finite trees.
For a PCS, however, the story is more complicated:
in particular, it is not sufficient that each of its member PECs is prefix-free,
because two PECs may reference subterms of each other.
For example, consider the PCS $C = \{c_1,c_2\} = \{\{0=1.0\}, \{0.0=1\}\}$.
Although $c_1$ and $c_2$ are prefix-free,
together they imply an inconsistent constraint $1.0.0=1$,
which can be obtained by substituting $1.0$ for $0$ in $c_2$, as justified by $c_1$.

For more intuition, consider two patterns $f(A, g(A))$ and $f(g(B), B)$;
it is easy to see that the terms matching these patterns satisfy the PECs $c_1$ and $c_2$, respectively.
The conjunction of the two PECs corresponds to the unification of the two patterns,
which produces unification constraints $A=g(B)$ and $g(A)=B$,
and eventually the contradictory constraint $B=g(g(B))$---%
which corresponds exactly to the $1=1.0.0$ PEC above.
In unification parlance, we say that this constraint fails an \emph{occurs check}. 
Checking consistency of a PCS is the name-free analogue of the occurs check.



\mypara{Checking Consistency via Congruence Closure}
These observations suggest an algorithm for checking PCS consistency: 
\begin{enumerate*}
\item saturate the PCS with all implied equalities (such as $1.0.0=1$ above), and
\item check if any of them is non-prefix-free. 
\end{enumerate*}
To formalize the former step, we first declaratively define the \emph{closure} operation on PCSs,
and then discuss how to implement it efficiently.

\begin{definition}[Closure]
  A PCS $C$ is \emph{closed} if the following holds for any $c_1,c_2\in C$:
  for any paths $p, p', p''$, if $p',p''\in c_1$ and $p'.p\in c_2$, then $p''.p \in c_2$.
  %
  In other words, whenever $c_2$ contains an extension of a path in $c_1$,
  it also contains the same extension of \emph{all} paths in $c_1$.
  The \emph{closure} of $C$, denoted $\closure(C)$,
  is the smallest closed PCS that contains $C$.
\end{definition}

\noindent
For example, the PCS $C = \{c_1, c_2\} = \{\{0=1.0\}, \{0.0=1\}\}$ is not closed:
if we set $p'=0, p''=1.0$, and $p=0$,
then $0\in c_1$, $1.0 \in c_1$, and $0.0 \in c_2$, but $1.0.0 \notin c_2$.
%
The closure of this PCS $\closure(C) = \{c'_1, c'_2\}$, where $c'_1$ and $c'_2$ are infinite PECs of the form
$c'_1 = \{0 = 1.0 = 0.0.0 = 1.0.0.0 = \ldots\}$ and 
$c'_2 = \{1 = 0.0 = 1.0.0 = 0.0.0.0 = \ldots\}$.
%


\begin{restatable}[Correctness of Closure]{theorem}{thmgenclosure}
  \label{thm:closure-correctness}
For any term $t \in \closedtermsof{\Sigma}$, $\pecsat{C}{t} \Leftrightarrow \pecsat{\closure(C)}{t}$.
\end{restatable}  


\label{sec:contains-closure-algorithm}
\begin{restatable}[Consistency of a Closed PCS]{theorem}{thmgenconsistencyfree}
\label{thm:consistency-prefix-free}
Let $C$ be a closed PCS. Then $C$ is inconsistent iff one of the $c_i\in C$ is not prefix-free.
\end{restatable}

\begin{wrapfigure}[8]{r}{0.2\textwidth}
  \vspace{-5pt}
  \includegraphics[width=0.2\textwidth]{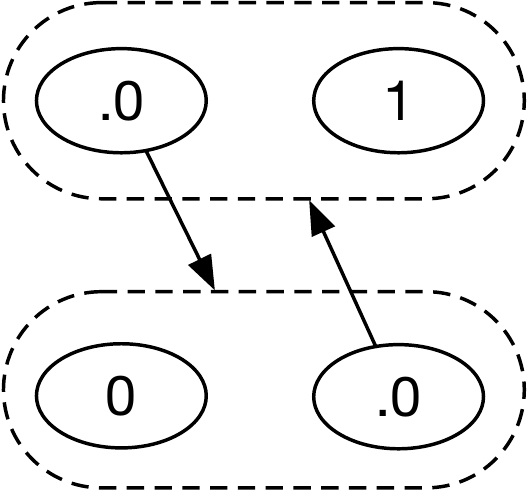}
\end{wrapfigure}
Together \autoref{thm:closure-correctness} and \autoref{thm:consistency-prefix-free}
ensure the correctness of our consistency checking procedure;
what is left is to implement the closure computation efficiently.
It turns out this can be done using the well-known congruence closure algorithm 
for the first-order theory of equality and uninterpreted functions~\cite{DBLP:journals/jacm/NelsonO80}.
This algorithm finitely represents a possibly infinite set of equalities using an e-graph.
%
Hence, to check consistency of a PCS $C$, we can simply
\begin{enumerate*}
  \item add each path of $C$ into an e-graph, interpreting path prefixes as subterms 
  (\ie $1.0$ is $.0$ applied to $1$);
  \item merge all paths from the same PEC into one e-class and run congruence closure on the e-graph;
  \item check if the resulting e-graph has cycles; if so, then $C$ is inconsistent.
\end{enumerate*}
The figure on the right shows the (cyclic) e-graph 
obtained by running this algorithm on our example $\{\{0=1.0\}, \{0.0=1\}\}$.


\subsection{Acyclic ECTAs: Core Definition}

\begin{figure}[t]  
  \begin{minipage}{.5\textwidth}
  \small  
  \centering
  \textbf{Syntax}
  $$
  \begin{array}{rll}
    c \gramdef& \{p_1 = \dots = p_n\} & \text{path equivalence classes} \\
    C \gramdef& \{c_1,\dots,c_m\} & \text {path constraint sets} \\
    n \gramdef& \node{\many{e}} & \text{nodes (states)}\\
    e \gramdef& \edge{s}{\many{n}}{C} \mid \edgebot & \text{transitions}\\
  \end{array}
  $$
  \end{minipage}%
  \begin{minipage}{.5\textwidth}
  \small
  \centering
  \textbf{Denotation}
  $$
  \begin{array}{rl}
    \\\\
    \dennode{\node{\many{e}}} &=  \bigcup_i \denedge{e^i} \\
    \denedge{\edge{s}{\many{n}}{C}} &= \left\{ s(\many{t}) \bigmid t^i \in \dennode{n^i}, \pecsat{C}{s(\many{t})}  \right\}
  \end{array}
  $$
  \end{minipage}
  \caption{Acyclic ECTAs: syntax and semantics. Here $s\in \Sigma$ and $p$ is a path.}
  \label{fig:syntax}
\end{figure}

%

Like string automata, tree automata are usually formalized as graphs, 
defined by a set of states and a transition function. 
For our purposes, it is more convenient to formalize ECTAs using a recursive grammar,
in the same style VSAs are typically presented~\cite{polozov2015flashmeta}.

\mypara{Syntax}
\autoref{fig:syntax} (left) shows the grammar for acyclic ECTAs,
consisting of mutually recursive definitions for nodes (states) $n \in N$ and transitions $e \in E$;
an ECTA then is identified with its root node,
which represents the final state.%
\footnote{Although this representation is restricted to ECTAs with a single final state (the root node),
this is not an important restriction:
any acyclic tree automaton is equivalent to the same automaton with all its final states merged into one.}
In a transition $\edge{s}{\many{n}}{C}$,%
\footnote{Hereafter we write $\many{x}$ to denote a sequence of $x$s,
with $x^i$ referring to the $i$-th element of that sequence.}
the number of child nodes $|\many{n}|$ must equal $\arity(s)$;
both $\many{n}$ and $C$ can be omitted when empty.
As is common for VSAs, we assume implicit sharing of sub-trees:
that is, an acyclic ECTA is a DAG with no duplicate sub-graphs.

The special symbol \edgebot denotes an ``empty transition'',
which is used in intermediate results of ECTA operations.
For symmetry, we also abbreviate the empty node, $\node{}$, as \nodebot.
A \emph{normalized} ECTA contains no occurrences of \edgebot or \nodebot,
unless the root is itself \nodebot. 
Any ECTA can be normalized by iteratively replacing any transition containing a \nodebot child with \edgebot, 
and removing all instances of \edgebot from the children of each node. 
For instance, $\node{\edgeucnull{\s{a}}, \edgeuc{+}{[\node{\edgeucnull{\s{b}}}, \nodebot]}}$ 
normalizes to $\node{\edgeucnull{\s{a}})}$.
We assume henceforth that all ECTAs are implicitly normalized after every operation.

\mypara{Semantics and Spurious Runs}
The \emph{denotation} of an acyclic ECTA, \ie the set of terms it accepts, 
is defined in \autoref{fig:syntax} (right) as a pair of mutually-recursive functions:
$\dennode{\cdot}\colon N \to \P{\closedtermsof{\Sigma}}$ and $\denedge{\cdot}\colon E \to \P{\closedtermsof{\Sigma}}$.
We define a partial order $\prec$ on ECTAs as the subset order on their denotations:
$n_1 \prec n_2$ iff $\dennode{n_1} \subseteq \dennode{n_2}$.
The \emph{skeleton} of an ECTA, \skeleton{n}, is obtained by recursively removing all path constraints from its transitions.
A \emph{spurious run} of $n$ is a term $t$, that is rejected by $n$ but accepted by its skeleton:
$t \notin \dennode{n} \wedge t \in \dennode{\skeleton{n}}$.




\subsection{Basic Operation: Union and Intersection}

We now present algorithms for two basic operations on ECTAs, union and intersection. 
They serve as building blocks for our two core contributions: static and dynamic reduction.

\mypara{Union}
The union of two ECTAs, $n_1 \union n_2$, simply merges the transition of their root nodes:
\begin{definition}[Union]
Let $n_1=\node{\many{e_1}}, n_2=\node{\many{e_2}}$ be two nodes. 
Then $n_1 \union n_2 = \node{\many{e_1}\cup\many{e_2}}$.
\end{definition}

\noindent
\autoref{fig:ecta-ops} gives an example of ECTA union $n_u = n_1 \union n_2$.

\begin{theorem}[Correctness of ECTA Union]
$\dennode{n_1 \union n_2} = \dennode{n_1} \cup \dennode{n_2}$.
\end{theorem}
\begin{trivproof}
Immediately from the definition of $\dennode{\cdot}$.
\end{trivproof}

\mypara{Intersection}
The intersection of two ECTAs is more involved.
Intersecting two nodes, $n_1 \intersect n_2$, involves intersecting all pairs of their transitions;
intersecting two transitions, $e_1 \intersect e_2$, in turn, involves intersecting their child nodes point-wise,
and is only well-defined if the symbols and PCSs of $e_1$ and $e_2$ are compatible:
\begin{definition}[Intersection]
Let $n_1=\node{\many{e_1}}, n_2=\node{\many{e_2}}$ be two nodes, then:
\[
  n_1 \intersect n_2 = \nodesymbol\left(\left\{e_1^i \intersect e_2^j \bigmid e_1^i\in\many{e_1}, e_2^j\in \many{e_2} \right\}\right)
\]
  
Let $e_1=\edge{s_1}{[n_1^0\dots n_1^{k-1}]}{C_1}$, $e_2=\edge{s_2}{[n_2^0\dots n_2^{l-1}]}{C_2}$ be two transitions, then:
\[
  e_1 \intersect e_2 = \begin{cases}
                   \edge{s_1}{[n_1^0 \intersect n_2^0,\dots,n_1^{k-1} \intersect n_2^{k-1}]}{C_1 \cup C_2} & \text{if}\ s_1 = s_2\ \text{and}\ C_1 \cup C_2\ \text{is consistent} \\
                   \edgebot & \text{otherwise}
                 \end{cases}
\]
\end{definition}

Consider the example of ECTA intersection $n_i = n_1 \intersect n_2$ in \autoref{fig:ecta-ops}.
To compute the intersection at the top level,
we intersect all pairs of transitions---$(f,g)$, $(f,h)$, $(g,g)$, and $(g,h)$---%
but the three pairs with incompatible function symbols simpy yield $\edgebot$ and are discarded.
To intersect the two $g$-transitions, we recursively intersect their \T{targ} and \T{tret} nodes;
the resulting $g$-transition also inherits its constraint from $n_2$.

\begin{theorem}[Correctness of ECTA Intersection]
$\dennode{n_1 \intersect n_2} = \dennode{n_1} \cap \dennode{n_2}$.
\end{theorem}
\begin{trivproof}
By induction on the structure of the two ECTAs.
\end{trivproof}

\begin{proposition}
$\nodebot \intersect n = n \intersect \nodebot = \nodebot$
\end{proposition}

\begin{corollary}
Define $n_1 \cong n_2$ if $\dennode{n_1}=\dennode{n_2}$. Then, with respect to $(\cong)$, the $(\intersect)$ and $(\union)$ operations form a distributive lattice, with $\nodebot$ as the bottom element, and $(\prec)$ as the order.
\end{corollary}

\begin{figure}
  \begin{minipage}{.5\textwidth}
    \includegraphics[scale=0.5]{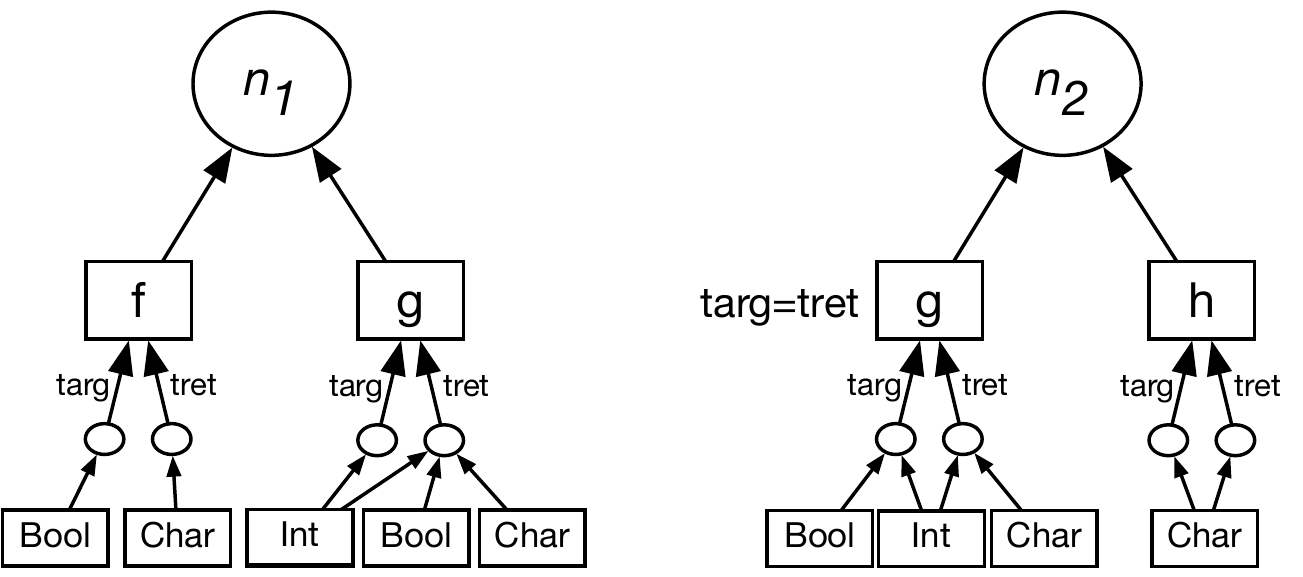}
  \end{minipage}%
  \begin{minipage}{.35\textwidth}
    \includegraphics[scale=0.5]{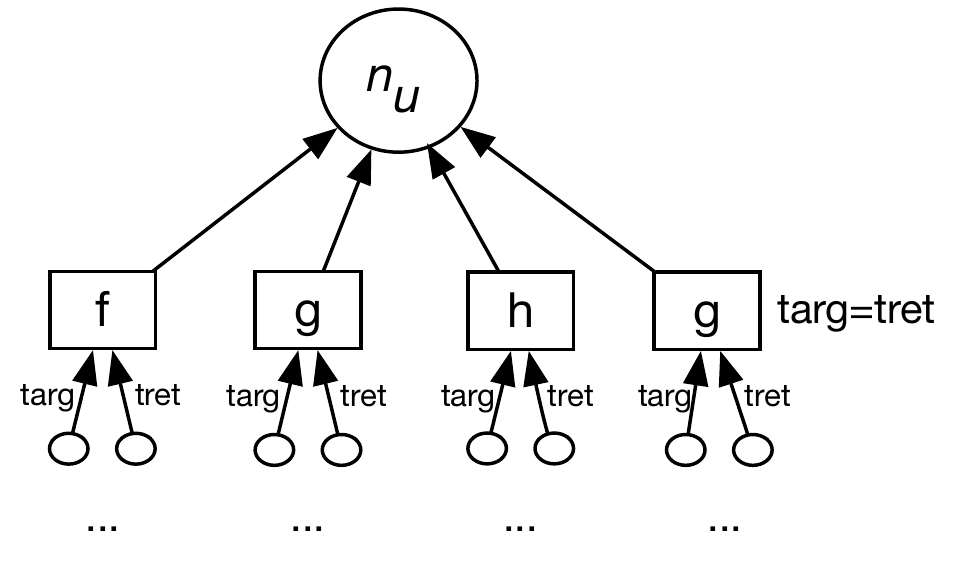}
  \end{minipage}%
  \begin{minipage}{.15\textwidth}
    \includegraphics[scale=0.5]{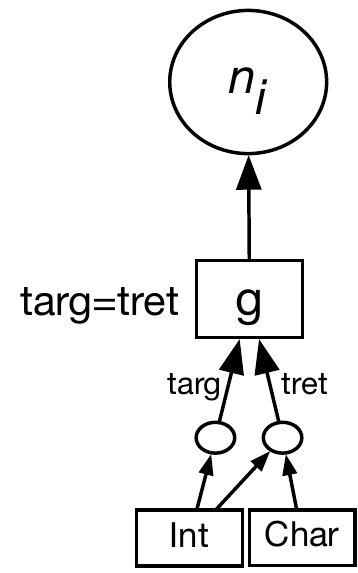}
  \end{minipage}
\caption{Two ECTAs $n_1$ and $n_2$, their union $n_u$ and intersection $n_i$.}\label{fig:ecta-ops}
\end{figure}

\subsection{Static Reduction}
\label{sec:static-reduction}

We are now ready to present static reduction, 
the first of the two core algorithms that enable efficient extraction of terms satisfying ECTA constraints. 
Consider the example in \autoref{fig:overview-ex1}.
Intuitively, the constraint \T{fun.targ = arg.type} has been \emph{reduced} in \autoref{fig:overview-ex1-reduced},
because with $f$ elimitated,
\emph{every} possible formal parameter type at path \T{fun.targ} matches \emph{some} actual parameter type at path \T{arg.type}.
%
%
More generally, a binary constraint $p_1=p_2$ is reduced if everything at path $p_1$ matches something at path $p_2$; 
this definition extends naturally to non-binary constraints. 
We now define the machinery to state this formally, and then provide a simple algorithm for reducing a constraint, 
which builds upon ECTA intersection.

\mypara{Subautomaton at a Path}
First, we generalize the definition of a subterm at a path, $\at{t}{p}$, to ECTAs:

\begin{definition}[Nodes at path, Subautomaton at a path]
The set $\nodesat(n,p)$ of nodes reachable from $n=\node{\many{e}}$ via path $p$ is defined as:
\[
  \nodesat(n, \epsilon) = \{n\}  \quad\quad\quad
  \nodesat(n, j.p) = \bigcup_{i} \nodesat(e^i, j.p)
\]
%
%
The set $\nodesat(e,p)$ of nodes reachable from a transition $e=\edge{s}{\many{n}}{C}$ 
is defined as:
\[
    \nodesat(e, j.p) = \begin{cases}
         \nodesat(n^j, p) & j < \arity(s) \\
         \emptyset & \text{otherwise}
    \end{cases}
\]
Finally, the \emph{subautomaton} of $n$ at path $p$ is defined as $\at{n}{p} = \bigunion \nodesat(n, p)$;
similarly, the subautomaton of $e$ is defined as $\at{e}{p} = \bigunion \nodesat(e, p)$.
\end{definition}

\noindent
In \autoref{fig:overview-ex1-orig}, if $n$ is the root node, 
then $\nodesat(n, \T{arg.type}) = \{ \node{\edgeucnull{\T{Int}}}, \node{\edgeucnull{\T{Char}}} \}$
and $\at{n}{\texttt{arg.type}} = \node{[\edgeucnull{\T{Int}}, \edgeucnull{\T{Char}}]}$.
The reader might be wondering why define $\nodesat(e, j.p)=\emptyset$ for an out-of-bounds index $j$
instead of restricting these and following definitions to ``well-formed'' paths.
The rationale is to enable ECTAs to have ``cousin'' transitions with different arities,
and be able to navigate to nodes and subautomata at higher arities, 
by simply discarding branches with lower arities;
this flexibility is required, for instance, in our full encoding of type-driven synthesis in \secref{sec:applications:hplus}.

Without equality constraints, 
the denotation of $\at{n}{p}$ would simply be the set of subterms $\at{t}{p}$ of all terms $t$ represented by $n$.
With equality constraints, $\at{n}{p}$ is an overapproximation of that set, 
since the equality constraints on the topmost layers get ignored:
\begin{theorem}[Correctness of subautomaton at a path]
\label{thm:nodes-at-path-correctness}
$\denotation{\at{n}{p}}^\text{N} \supseteq \left\{\at{t}{p} \bigmid t \in \denotation{n}^\text{N} \right\}$
\end{theorem}
\begin{trivproof}
By induction on $p$.
\end{trivproof}

\mypara{Reduction Criterion}
We can now formally state what it means for a constraint to be \textit{reduced}:

\begin{definition}[Reduction Criterion]
Let $e=\edge{s}{\many{n}}{C}$ be a transition and let $c=\{p_1=\dots=p_k\} \in C$.
We say that $e$ satisfies the \emph{reduction criterion} for $c$ 
(alternatively, $c$ is \emph{reduced} at $e$) if, 
for each $p_i,p_j\in c$ and each $n \in \nodesat(e, p_i)$, 
$n \intersect \at{e}{p_j} \neq \nodebot$.
\end{definition}

\noindent
The reduction criterion suggests an algorithm for reducing a path constraint: 
given a constraint $p_1=p_2$ on transition $e$, 
replace every node $n$ reachable via $p_1$ with $n \intersect \at{e}{p_2}$. 
As a result, every node in $\nodesat(e,p_1)$ will match some node in $\nodesat(e,p_2)$. 
%
For example, to reduce the constraint \T{fun.targ = arg.type} at the transition \T{app} in \autoref{fig:overview-ex1}, 
the algorithm first computes $\at{\T{app}}{\texttt{arg.type}}$, 
the automaton representing all possible actual parameter types;
the result is $n_a = \node{\edgeucnull{\T{Int}}, \edgeucnull{\T{Char}}}$.
Next, it intersects $n_a$ it with each of the three nodes reachable via \T{fun.targ},
that is, \T{Int}, \T{Char}, and \T{Bool}.
This has no effect on the \T{targ} children of $g$ and $h$, 
but the \T{targ} child of $f$ becomes \nodebot, 
leading to the removal of the $f$ transition upon normalizatoin
and resulting in \autoref{fig:overview-ex1-reduced}.

\mypara{Intersection at a Path}
In order to formalize the reduction algorithm outlined above,
we introduce the notion of \emph{intersection at a path}.
\begin{definition}[Intersection at a Path]
Intersecting node $n$ with node $n'$ at path $p$, denoted $\intersectatpath{n}{p}{n'}$, 
replaces all nodes reachable from $n$ via $p$ with their intersection with $n'$.
More formally, if $n=\node{\many{e}}$:
\[
  \intersectatpath{n}{\epsilon}{n'} = n \intersect n' \quad\quad\quad
  \intersectatpath{n}{j.p}{n'} = \node{\intersectatpath{e^i}{j.p}{n'}}
\]
where intersecting a transition $e=\edge{s}{[n^0,\dots,n^{k-1}]}{C}$ at a non-empty path $p$
is defined as:
\[
    \intersectatpath{e}{j.p}{n'} = \begin{cases}
      \edge{s}{[n^0,\dots,\intersectatpath{n^j}{p}{n'},\dots,n^{k-1}]}{C} & j < \arity(s) \\
         \edgebot & \text{otherwise}
    \end{cases}
\]
\end{definition}

\noindent
For example, in \autoref{fig:overview-ex1-orig}, intersecting the root node $n$ at path \T{fun.targ}
with the node $n_a$ from our previous example ($n_a = \node{\edgeucnull{\T{Int}}, \edgeucnull{\T{Char}}}$)
yields the ECTA in \autoref{fig:overview-ex1-reduced}.


\begin{lemma}
\label{lem:intersecting-at-path-already-present}
If $t\in \dennode{\at{n}{p}}$ and $t \in \dennode{n'}$, then $t\in \dennode{\at{\left(\intersectatpath{n}{p}{n'}\right)}{p}}$.
\end{lemma}
\begin{trivproof}
By induction on $p$. 
\end{trivproof}

\mypara{Reduction Algorithm}
With this new terminology, we can recast our previous explanation 
of how the constraint \T{fun.targ = arg.type} in \autoref{fig:overview-ex1} gets reduced:
once we have obtained the ``actual parameter automaton'' $n_a=\at{\T{app}}{\texttt{arg.type}}$,
we can simply return $\intersectatpath{n}{\texttt{fun.targ}}{n_a}$ (where $n$ is the root node).
This explanation needs one final tweak:
in this example, the information only propagates in one direction---from \T{arg.type} to \T{fun.targ}---%
because the types of the actuals happen to be a subset of the types of the formals;
in general, though, reduction needs to propagate information both ways.
Hence a more accurate recipe for how to perform the reduction in \autoref{fig:overview-ex1} is:
\begin{enumerate*}
\item compute the automaton
$n^* = \left(\at{n}{\texttt{fun.targ}}\right) \intersect \left(\at{n}{\texttt{arg.type}}\right)$,
capturing all \emph{shared} formal and actual parameter types;
\item intersects the root with $n^*$ at \emph{both} paths involved in the constraint:
$\intersectatpath{\left(\intersectatpath{n}{\texttt{fun.targ}}{n^*}\right)}{\texttt{arg.type}}{n^*}$. 
\end{enumerate*}
We extrapolate this description into a general algorithm for static reduction:

\begin{definition}[Static Reduction]
\label{defn:basic-static-reduction}
Let $c=\{p_1=\dots=p_k\}$ be a prefix-free PEC; then
\[
  \reduce(e, c) = \intersectatpath{\intersectatpath{e}{p_1}{n^*}\dots}{p_k}{n^*} \quad\quad\text{where}\quad
  n^* = \bigintersect_{p_i\in c} \at{e}{p_i}
\]
\end{definition}

\begin{restatable}[Completeness of Reduction]{theorem}{thmcompreduction}
  \label{thm:reduce-completeness}
$\reduce(e,c)$ satisfies the reduction criterion for $c$.
\end{restatable}

\begin{restatable}[Soundness of Reduction]{theorem}{thmsoundreduction}
  \label{thm:reduce-soundness}
Let $e=\edge{s}{\many{n}}{C}$ be a transition and $c\in C$; then:
\[
  \denedge{\reduce(e,c)} = \denedge{e}
\]
\end{restatable}

\section{Fast Enumeration with Dynamic Reduction}
\label{sec:enumeration}

We now turn to our second core contribution:
the algorithm for efficiently extracting (or enumerating) terms represented by an ECTA.
As we have outlined in \secref{sec:overview:dynamic},
the main idea behind the algorithm is to avoid eager enumeration of constrained nodes,
instead replacing them with ``unification'' variables---%
the mechanism we dub \emph{dynamic reduction}.

Inspired by presentations of DPLL(T) and Knuth-Bendix completion~\cite{bachmair1994equational,nieuwenhuis2006solving}, 
we formalize the enumeration algorithm as a non-deterministic transition system.
Configurations of this system are called \emph{enumeration states}
and steps are governed by two rules, \textsc{Choose} and \textsc{Suspend}.
Intuitively, \textsc{Choose} handles unconstrained ECTA nodes,
making a non-deterministic choice between their incoming transitions;
\textsc{Suspend} handles constrained nodes,
suspending them into variables.
\autoref{fig:enum}, which serves as the running example for this section,
shows an example sequence of \textsc{Choose} and \textsc{Suspend} steps
applied to a simplified version of the ECTA from \autoref{fig:overview-ex3}
(the simplified ECTA encodes all well-typed size-two terms 
in the environment $\Gamma = \{x\colon \T{Int}, y\colon \T{Char}, g\colon \alpha \to \alpha, h\colon \T{Char}\to\T{Bool}\}$).

\begin{figure}
  \centering
  \includegraphics[width=0.95\textwidth]{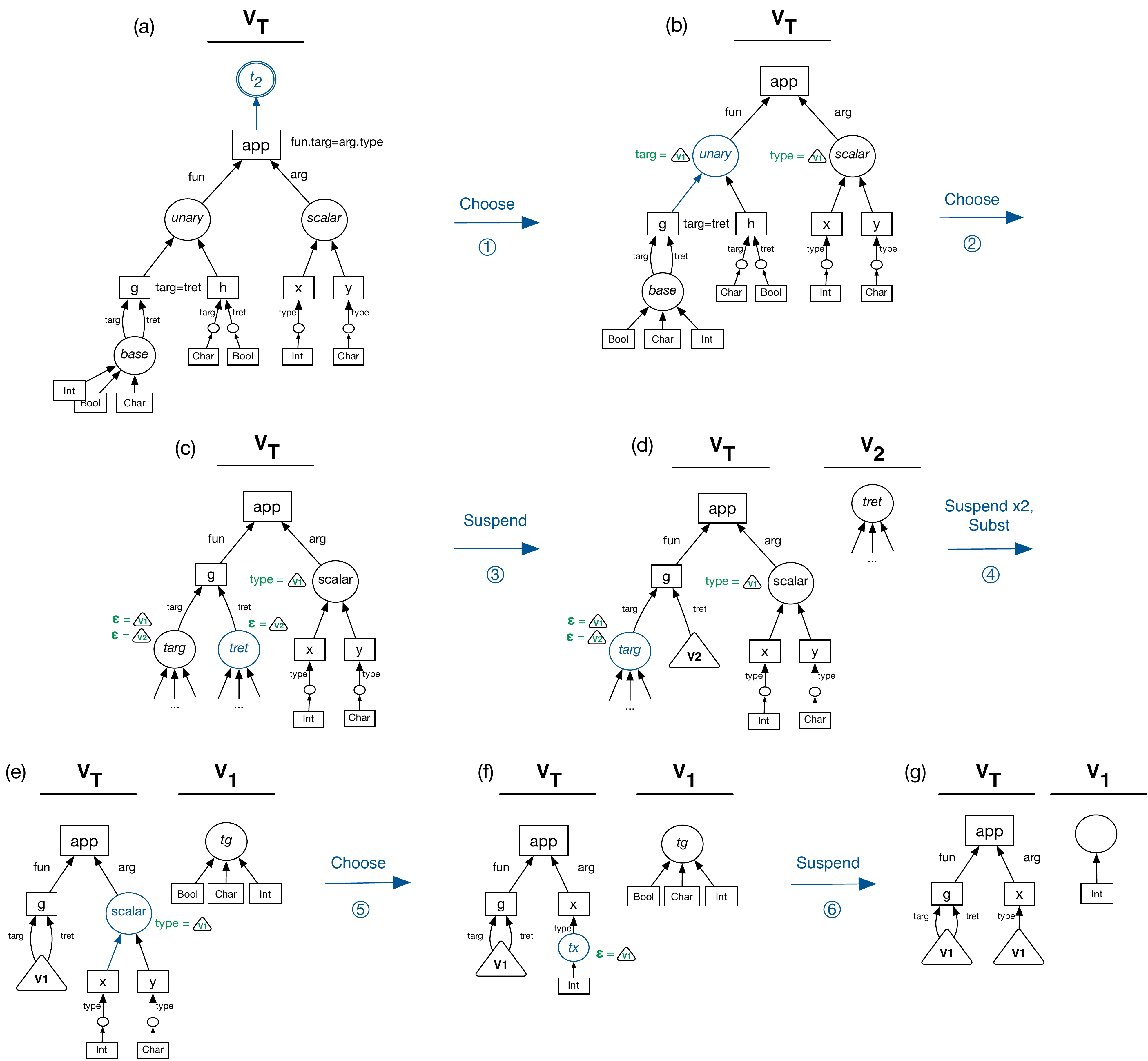}
  \caption{An example sequence of steps through enumeration states. 
  The focus node of each step is highlighted in blue
  and constraint fragments are highlighted in green. 
  The final state is fully enumerated.}\label{fig:enum}
\end{figure}

\subsection{Enumeration State}

The syntax of enumeration states is shown in \autoref{fig:enum-defs} (left).
$\set{Var}$ is a countably infinite set of variables, 
with a dedicated ``root'' variable $v_\top \in \set{Var}$.
An \emph{enumeration state} $\sigma$ is a mapping from variables to \emph{partially-enumerated terms}
(or \emph{p-terms} for short).
With the exception of $v_\top$, which stores the top-level enumeration result,
each variable captures a set of ECTA nodes that are constrained to be equal.
For example, in \autoref{fig:enum}~(g), the variable $v_1$ captures the nodes
\T{g.targ} and \T{x.type}, which are equated by the constraint on \T{app},
and also \T{g.tret}, equated to the former by the constraint on \T{g}.

%
A p-term $\tau$ is a term that might contain variables and \emph{unenumerated nodes}
(\emph{u-nodes} for short).
A u-node $\unenum{n}{\Phi}$ is an ECTA node $n$
annotated with zero or more \emph{constraint fragments} $\phi$,
each consisting of a PEC and a variable.
Intuitively, a constraint fragment is a constraint propagated downward from an ECTA transition.
For example, in \autoref{fig:enum}~(b), 
when the original constraint \T{fun.targ = arg.type} on \T{app} 
is propagated down to \staten{unary} and \staten{scalar},
it is split into two fragments: $\fragment{\T{targ}}{v_1}$ and $\fragment{\T{type}}{v_1}$.
The splitting is necessary because for each of the child u-nodes
one of the sides of this constraint is ``out of scope'';
hence a fresh variable $v_1$ is introduced to refer to the common value at both sides.
A variable $v$ is \emph{solved} in $\sigma$ iff it is not mentioned in any of the constraint fragments;
for example, $v_1$ is unsolved in \autoref{fig:enum}~(b)--(f) and solved in \autoref{fig:enum}~(g).

%
A u-node is \emph{restricted} iff its $\Phi$ is non-empty;
an unrestricted u-node is written $\unenumuc{n}$.
An enumeration state $\sigma$ is called \emph{fully enumerated}
if there are no restricted u-nodes anywhere inside $\sigma$.
The reader might be surprised that a fully enumerated state is allowed to have u-nodes at all;
as we explain in \secref{sec:enum:balanced},
this enables compact representation of enumeration results with ``trivial differences.''

\begin{figure}[t]  
  \begin{minipage}{.43\textwidth}
  \small  
  \centering
  \textbf{Enumeration States}
  $$
  s\in \Sigma, v\in \set{Var}, c \in \PEC
  $$
  $$
  \begin{array}{rll}    
    \phi \gramdef &\fragment{c}{v}   & \text{\emph{Constraint fragments}}\\
    \Phi \gramdef &\many{\phi}       & \text{\emph{Cons. fragment sets}}\\
    \tau \gramdef &                   & \text{\emph{P-terms}}\\
            &\mid v \mid s(\many{\tau})   & \text{variable, application}\\ 
            &\mid \unenum{n}{\Phi}   & \text{unenumerated node}\\
    \sigma \gramdef &[\many{v \mapsto \tau}] &\text{\emph{Enumeration states}}\\
    \ctx[\cdot] \gramdef &            & \text{\emph{Contexts}}\\ 
            &\mid \cdot               & \\
            &\mid s(\many{\tau},\ctx[\cdot],\many{\tau}) &
  \end{array}
  $$
  \end{minipage}%
  \begin{minipage}{.57\textwidth}
  \small
  \centering
  \textbf{Enumeration step}\quad$\boxed{\steps{\tau}{\tau'}, \steps{\sigma}{\sigma'}}$
  \begin{gather*}
  \inference[\textsc{Choose-$\square$}]
  { \edge{s}{[n^0,\dots,n^{k-1}]}{c} \in \many{e} \\     
    \tau^i = \unenum{n^i}{\proj(\fragment{c}{v} \cup \Phi,i)} \\ 
    v\ \text{is fresh} \quad \fragment{\epsilon}{\_} \notin \Phi \\ 
    }
  { \steps{\unenum{\node{\many{e}}}{\Phi}}{s(\tau^0,\dots,\tau^{k-1})} }
  \\
  \inference[\textsc{Choose}]
  { \sigma[v] = \ctx[\tau] \quad \steps{\tau}{\tau'} \quad v\ \text{is solved} }
  { \steps{\sigma}{\sigma[v \mapsto \ctx[\tau']]} }
  \\
  \inference[\textsc{Suspend-1}]
  { \sigma[v] = \ctx[\unenum{n}{\fragment{\epsilon}{v'} \cup \Phi}] \quad v' \notin \dom(\sigma) }
  { \steps{\sigma}{\sigma[v \mapsto \ctx[v'], v' \mapsto \unenum{n}{\Phi}]} }
  \\
  \inference[\textsc{Suspend-2}]
  { \sigma[v] = \ctx[\unenum{n}{\fragment{\epsilon}{v'} \cup \Phi}] \quad \sigma[v'] = \unenum{n'}{\Phi'} }
  { \steps{\sigma}{\sigma[v \mapsto \ctx[v'], v' \mapsto \unenum{n \intersect n'}{\Phi \cup \Phi'}]} }
  \end{gather*}  
  \end{minipage}
  \caption{Enumeration states and rules.}
  \label{fig:enum-defs}
\end{figure}

\mypara{Denotation}
The denotation of an enumeration state $\denstate{\sigma}$ 
is a set of substitutions $\rho\colon \set{Var}\partialfn \closedtermsof{\Sigma}$, 
which is compatible with the constraint fragments and subterm relations imposed by variables inside p-terms. 
Because of the circular dependencies between a p-term and its enclosing $\sigma$, 
the formal definition is somewhat technical and therefore relegated to \appref{app:proofs:enumeration-state}.

\subsection{Enumeration Rules}

\autoref{fig:enum-defs}~(right) formalizes the above-mentioned \textsc{Choose} and \textsc{Suspend} rules
as a step relation $\steps{\sigma}{\sigma}$ over enumeration states
and an auxiliary step relation $\steps{\tau}{\tau'}$ over p-terms.
%

\mypara{\textsc{Choose}}
We first formalize the auxiliary rule \textsc{Choose}-$\square$ for p-terms.
This rule takes a u-node, non-deterministically selects one of its transitions $e$,
and steps to a p-term that has $e$'s function symbol at the root and new u-nodes as children.
Step \circled{5} in \autoref{fig:enum} is an example application of this rule:
here the original u-node \staten{scalar} turns into one of its two incoming transitions, \T{x};
step \circled{1} is also an instance of this rule, albeit with no alternatives.

The tricky aspect of \textsc{Choose}-$\square$ is propagating constraints---%
either from the transition $e$ or from the original u-node---%
to the newly minted u-nodes.
The former scenario is illustrated in step \circled{1}:
here the PEC $c = \{\T{fun.targ} = \T{arg.type}\}$ on the \T{app} transition
is split into two fragments, $\fragment{\T{targ}}{v_1}$ and $\fragment{\T{type}}{v_1}$, 
attached to the new u-nodes \staten{unary} and \staten{scalar}, respectively.
To this end, \textsc{Choose}-$\square$ first creates a fresh variable $v_1$
and forms a constraint fragment $\fragment{c}{v_1}$ using the original PEC $c$;
next it \emph{projects} this fragment down to each $i$-th child,
retaining only those paths of $c$ that start with $i$ and chopping off their heads.
The $\proj$ function is defined formally in \autoref{fig:project}.
Note how the two new fragments together completely capture the semantics of the original constraint.

The latter scenario---propagating existing constraint fragments---is illustrated in step \circled{5}.
Here the u-node \staten{scalar} is restricted by the fragment $\fragment{\T{type}}{v_1}$;
in this case, there is no need to create new variables:
the existing fragment is simply projected down to the child \staten{tx} and becomes $\fragment{\epsilon}{v_1}$.
In the general case, both new and existing constraint fragments should be combined;
this is the case in step \circled{2}, 
where the u-node \staten{targ} inherits the fragment $\fragment{\epsilon}{v_1}$ from \staten{unary},
and also acquires a new fragment $\fragment{\epsilon}{v_2}$ by splitting the constraint on \T{g}.

Finally, consider the rule \textsc{Choose}, 
which lifts \textsc{Choose}-$\square$ to whole enumeration states.
This rule allows making a step inside any component of $\sigma$,
as long as its variable is solved.
For example, in \autoref{fig:enum}~(e) we are not allowed to make a step inside $v_1$
(say, choosing \T{Bool} among the three types),
because $v_1$ still appears in the constraint fragment $\fragment{\T{type}}{v_1}$ inside $v_\top$.
The rationale for this restriction is to avoid making premature choices for constrained nodes:
in our example, picking \T{Bool} would be a mistake,
which is entirely avoidable by simply waiting until all constraints are resolved
(such as the state in \autoref{fig:enum}~(g)).

\mypara{\textsc{Suspend}}
The \textsc{Suspend} rules handle u-nodes with \emph{$\epsilon$-fragments},
\ie constraint fragments of the form $\fragment{\epsilon}{v}$.%
\footnote{\textsc{Choose}-$\square$ does not apply to these nodes thanks to its last premise.
Note also that because all PECs in the original ECTA are prefix-free and this property is maintained by $\proj$,
any fragment that contains $\epsilon$, must \emph{only} contain $\epsilon$.}
Intuitively, an $\epsilon$-fragment indicates that this node
is the target of a constraint captured by $v$.
In response, the \textsc{Suspend} rules simply ``move'' the target u-node to the $v$-component of the state,
replacing it with $v$ in the original p-term.

The two \textsc{Suspend} rules differ in whether the current state $\sigma$ already has a mapping for $v$:
if it does not, \textsc{Suspend-1} initializes this mapping with its target u-node $\unenum{n}{\Phi}$;
if it does, \textsc{Suspend-2} updates this mapping,
combining the old u-node $\unenum{n'}{\Phi'}$ and the new one $\unenum{n}{\Phi}$
by intersecting their ECTAs and merging their constraint fragments.
Note that the old value of $v$ must necessarily be a u-node, 
because \textsc{Choose} is not allowed to operate under unsolved variables.

An example application of \textsc{Suspend-1} is step \circled{3} of \autoref{fig:enum}.
The target node \staten{tret} has an $\epsilon$-fragment $\fragment{\epsilon}{v_2}$;
since $v_2$ is uninitialized,
\textsc{Suspend-1} creates a new mapping $[v_2 \mapsto \text{\staten{tret}}]$.
Step \circled{6}, on the other hand, is an example of \textsc{Suspend-2}:
the target node \staten{tx} is restricted by $\fragment{\epsilon}{v_1}$;
since $v_1$ already maps to \staten{tg}, 
\textsc{Suspend-2} updates it with $\text{\staten{tx}} \intersect \text{\staten{tg}}$.
As a result of this intersection, $v_1$ now contains only those types (in this case, the sole type \T{Int})
that make the term represented by $v_\top$ well-typed.

\mypara{Eliminating Redundant Variables}
Finally, let us demystify the transformation \circled{4} in \autoref{fig:enum},
which consists of three atomic steps.
The first step suspends \staten{targ},
which has \emph{not one but two} $\epsilon$-fragments---%
$\fragment{\epsilon}{v_1}$ and $\fragment{\epsilon}{v_2}$---%
either of which can be targeted by a \textsc{Suspend}.
Suppose that the second one is chosen 
(both choices lead to equivalent results, up to variable renaming).
Since $v_2$ is already initialized, \textsc{Suspend-2} fires,
merging \staten{tret} and \staten{targ} into a single u-node \staten{tg'} under $v_2$;
importantly, \staten{tg'} inherits the other constraint fragment from \staten{targ},
namely $\fragment{\epsilon}{v_1}$.
Because of that, \textsc{Suspend-1} can now fire on \staten{tg'},
creating the state $[v_\top\mapsto \dots, v_2\mapsto v_1, v_1\mapsto \text{\staten{tg}}]$,
where \staten{tg} is \staten{tg'} stripped of its constraint fragment.
This new state is a bit awkward, since it contains a ``redundant'' variable $v_2$,
which simply stores another variable, $v_1$.
To get rid of such redundant variables, we introduce an auxiliary rule \textsc{Subst},
which simply replaces all occurrences of $v_2$ with $v_1$ and removes the unused mapping from $\sigma$
(see \autoref{fig:project}).
After applying \textsc{Subst}, we arrive at the state in \autoref{fig:enum}~(e).

\subsection{Enumeration Algorithm}
\label{sec:enumeration:algorithm}

We are now ready to describe the top-level algorithm \textsc{Enumerate}.
The algorithm takes as input an ECTA $n$ and produces a stream of fully-enumerated states.
To this end, it first creates an initial state $\sigma_0=[v_\top \mapsto \unenumuc{n}]$,
and then enumerates derivations of $\reduces{\sigma_0}{\finalstate}$ where \finalstate is fully enumerated
and $\stepsymb^*$ is the reflexive-transitive closure of $\stepsymb$.
In each step, the algorithm has the freedom to select
\begin{enumerate*}[(i)]
  \item which u-node to target, and
  \item in the case of \textsc{Choose}, which transition to choose.
\end{enumerate*}
The enumeration rules are designed in such a way that
the former selection constitutes ``don't care non-determinism''
(\ie any target node can be selected without loss of completeness);
this is in contrast to the latter selection, 
which constitutes ``don't know non-determinism'' and must be backtracked.
At the same time, different schedules of rule applications might lead to significantly different performance.
The \ectalibrary library provides a default schedule---depth-first, left to right---%
but enables the user to specify a domain-specific schedule in order to optimize performance. 


\begin{restatable}[Termination of Enumeration]{theorem}{themterminationenuremation}
\label{thm:enumeration-termination}
There is no infinite sequence $\sigma_0 \stepsymb \sigma_1 \stepsymb \dots$. 
\end{restatable}

\begin{restatable}[Correctness of Enumeration]{theorem}{thmcorrectnessenumeration}
Let $n$ be an ECTA, $\sigma_0$ be the initial enumeration state, 
and consider all finite sequences $\reduces{\sigma_0}{\finalstate}$, such that \finalstate is fully enumerated; then: 
$$\dennode{n}=\left\{ \rho(v_\top) \bigmid \reduces{\sigma_0}\finalstate, \rho \in \denstate{\finalstate}\right\}$$
\end{restatable}

\subsection{Compact Fully Enumerated States}\label{sec:enum:balanced}

\begin{figure}
  \begin{minipage}{.5\textwidth}
    \centering
    \small
    \textbf{Projecting constraint fragments}
    \[
      \begin{array}{rl}
      \proj(\Phi,i) &= \bigcup_{\phi \in \Phi} \proj(\phi, i) \\\\
      \proj(\fragment{c}{v},i) &= 
      \begin{cases}
        \{\fragment{c'}{v}\} &\text{if}\ c'\neq\emptyset \\
        \emptyset & \text{otherwise}
      \end{cases}\\
      \text{where}\ c' &= \bigcup_{p\in c} \proj(p,i) \\\\
      \proj(p, i) &= 
        \begin{cases}
          \bot   &\text{if}\ p = \epsilon\\         
          \{p'\} &\text{if}\ p = i.p'\\
          \emptyset &\text{otherwise}
        \end{cases}
      \end{array}
    \]
    \textbf{Enumeration step (cont.)}\quad$\boxed{\steps{\sigma}{\sigma'}}$
    \[
    \inference[\textsc{Subst}]
    { \sigma[v_2] = v_1 }
    { \steps{\sigma}{\subst{v_1}{v_2}{\left(\sigma \setminus [v_2\mapsto v_1]\right)}} }
    \]
    \caption{Auxiliary definitions}
    \label{fig:project}
  \end{minipage}%
  \begin{minipage}{.5\textwidth}
    \centering
    \begin{subfigure}[t]{0.3\textwidth}
    \centering
    \includegraphics[scale=0.35]{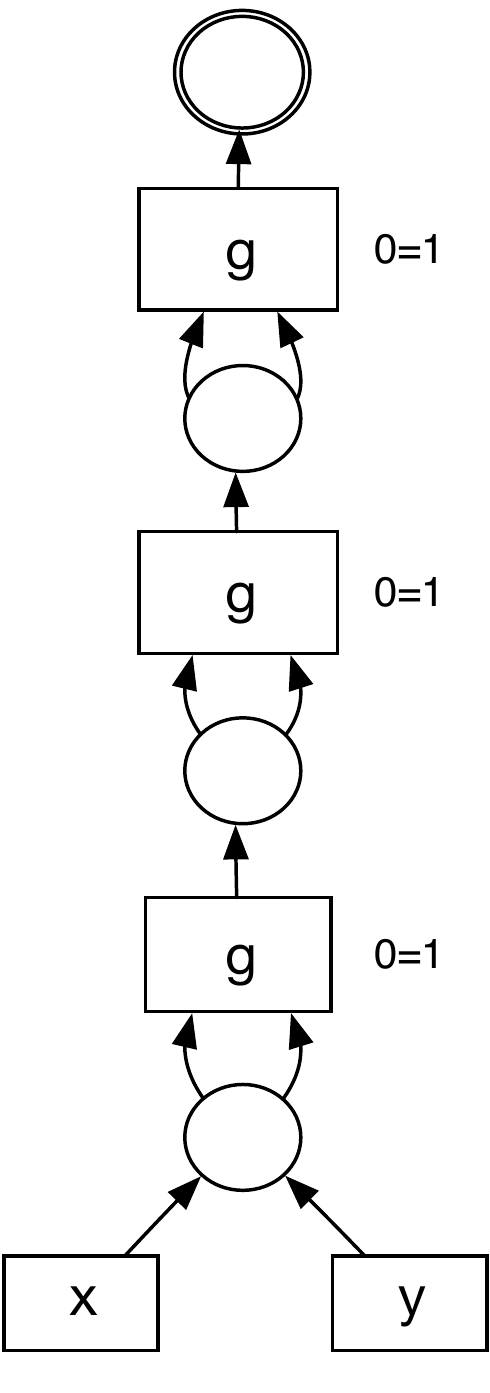}
    \subcaption{}
    \label{fig:balanced-tree:ecta}
    \end{subfigure}
    ~
    \begin{subfigure}[t]{0.7\textwidth}
    \centering
    \includegraphics[scale=0.35]{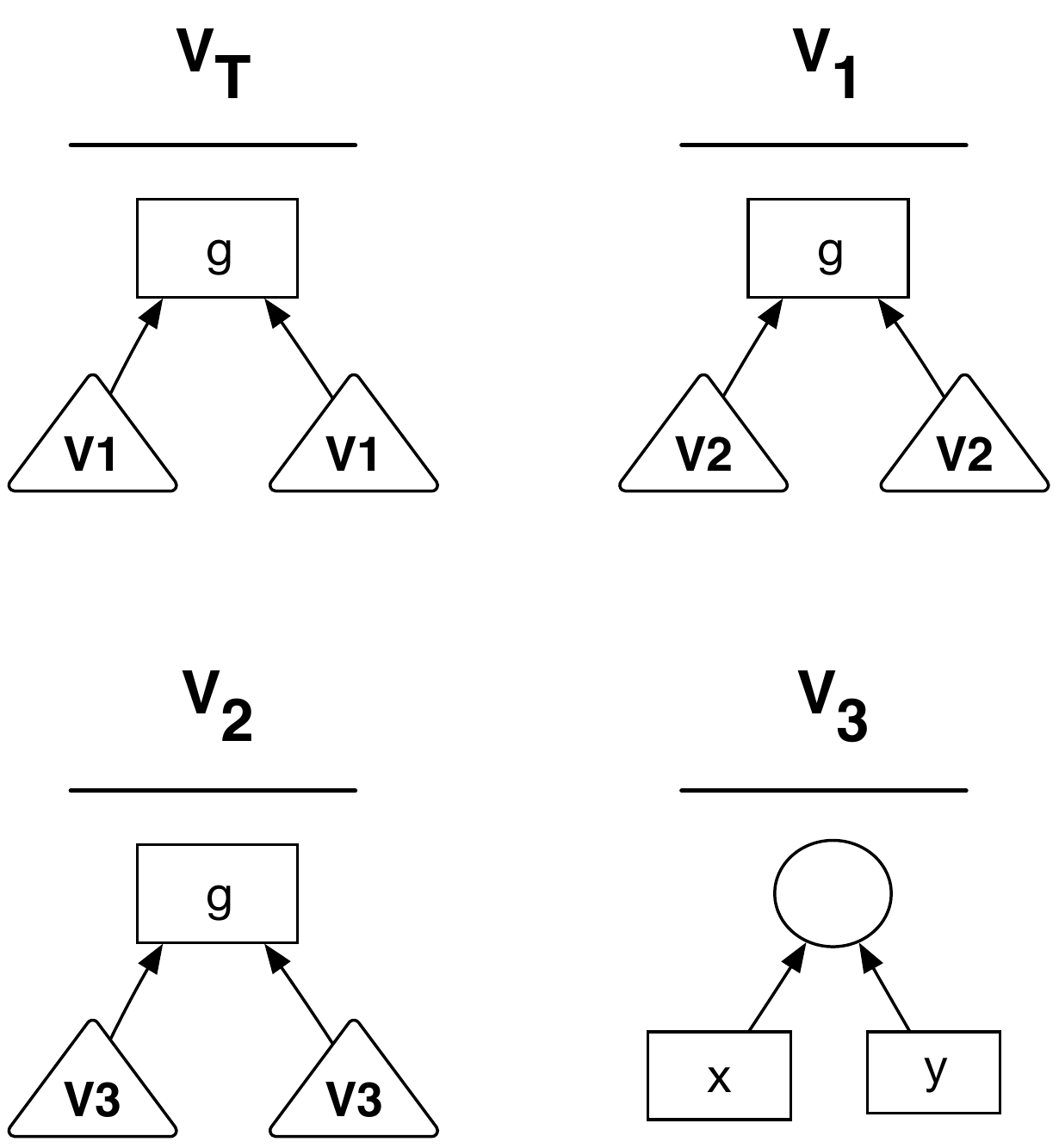}
    \subcaption{}
    \label{fig:balanced-tree:enumerated}
    \end{subfigure}
    \caption{(a) An ECTA representing all perfect trees of depth three, whose leaves are either all \T{x} or all \T{y} 
    (b) The ECTA fully enumerated, in logarithmic space}
    \label{fig:balanced-tree}
  \end{minipage}%
  \end{figure}

We now return to the design decision to allow (unrestricted) u-nodes in fully enumerated states.
Our running example in \autoref{fig:enum} does not motivate this decision very well:
the fully enumerated state in \autoref{fig:enum}~(g) encodes a single term anyway,
so it seems only natural to let \textsc{Choose} loose on the last remaining u-node.
For other ECTAs, however, a single fully-enumerated state might represent exponentially many%
\footnote{Or, with the  cyclic ECTAs of \secref{sec:cyclic}, infinitely many.}
terms, or the terms might be exponentially larger, or both.
For an example, consider \autoref{fig:balanced-tree:ecta}.
This ECTA represents the set of all perfect binary trees of depth three,
whose leaves are either all \T{x} or all \T{y}.
A moment's thought reveals that this set contains two trees, each of size $15$.
Instead of returning these two large trees explicitly,
the fully-enumerated state \finalstate in \autoref{fig:balanced-tree:enumerated} 
represents them as a \emph{hierarchy of unconstrained tree automata}, 
from which the concrete trees may be trivially generated.
It is straightforward to see that the sizes of the two perfect trees grow exponentially with their depth,
while the size of \finalstate grows only linearly.

The main benefit of this design, however,
is that, depending on the problem domain, some nodes \emph{need not be enumerated at all},
as long as we know their denotation is non-empty.
For example, to determine whether a propositional formula is satisfiable (\secref{sec:applications:sat}),
it is often enough to provide a \emph{partial satisfying assignment},
because the values of the unassigned variables are irrelevant;
such a partial assignment can be represented by a \finalstate,
where irrelevant variables are left unenumerated.
Similarly, in type-driven synthesis,
the polymorphic type of a component need not always be fully instantiated,
as long we know that a compatible instantiation exists.
In fact, as we explain in \secref{sec:applications:hplus}, 
cyclic ECTAs can encode infinitely many possible polymorphic instantiations,
and enumerating them all would be simply impossible.

\section{Cyclic ECTA}
\label{sec:cyclic}

\begin{figure}[t]
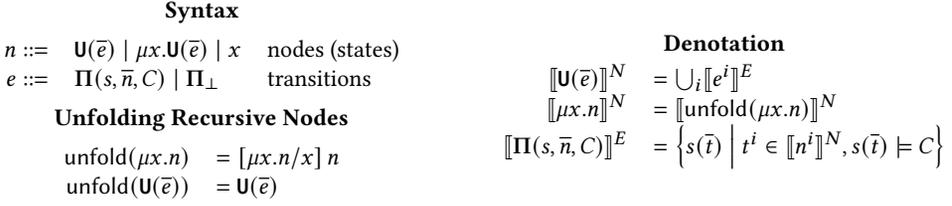
  
  \begin{minipage}{.5\textwidth}
  \small  
  \centering
  \textbf{Syntax}
  $$
  \begin{array}{rll}
    n \gramdef& \node{\many{e}} \mid \recnode{x}{\node{\many{e}}} \mid x & \text{nodes (states)}\\
    e \gramdef& \edge{s}{\many{n}}{C} \mid \edgebot & \text{transitions}\\
  \end{array}
  $$
  \\
  \textbf{Unfolding Recursive Nodes}
  \[
    \begin{array}{rl}
      \unfold{\recnode{x}{n}} &= \bigsubst{\recnode{x}{n}}{x}{n} \\
      \unfold{\node{\many{e}}} &= \node{\many{e}}
    \end{array}
  \]
  \end{minipage}%
  \begin{minipage}{.5\textwidth}
  \small
  \centering
  \textbf{Denotation}
  $$
  \begin{array}{rl}
    \dennode{\node{\many{e}}} &=  \bigcup_i \denedge{e^i} \\
    \dennode{\recnode{x}{n}} &= \dennode{\unfold{\recnode{x}{n}}} \\
    \denedge{\edge{s}{\many{n}}{C}} &= \left\{ s(\many{t}) \bigmid t^i \in \dennode{n^i}, \pecsat{C}{s(\many{t})}  \right\}
  \end{array}
  $$
  \end{minipage}
  \caption{Cyclic ECTAs: syntax and semantics. Here $s\in \Sigma$, $C$ is a PCS, and $x$ is a bound variable.}
  \label{fig:cyclic-syntax}
\end{figure}

We now present the formalism for fully general ECTAs, which may contain cycles. 
With cycles, an ECTA node can now represent an infinite space of terms, 
such as an arbitrary term in some context-free language, 
including (as in \secref{sec:applications:hplus}) the language of arbitrary Haskell types.
While this requires an extension to the syntax of ECTAs to allow recursion, 
shockingly, none of the algorithms require substantial modification.

\subsection{Cyclic ECTAs: Core Definition}

We extend acyclic ECTAs to cyclic by adding ``recursive nodes'' $\recnode{x}{\node{\many{e}}}$. 
Within this node, $x$ is a variable bound to $\node{\many{e}}$. 
In diagrams, we depict any use of $x$ as a back-edge to $\node{\many{e}}$ 
and keep the $\mu$ binding itself implicit.
Semantically, $x$ can be replaced with a copy of the node it is bound to, 
so that an ECTA $n$ is equivalent to  $\subst{\node{\many{e}}}{x}{n}$---%
or rather, to $\subst{\recnode{x}{\node{\many{e}}}}{x}{n}$, since $\node{\many{e}}$ contains further uses of $x$. 
\autoref{fig:cyclic-example:orig} shows an example ECTA, 
with a recursive node \staten{Nat} representing arbitrary natural numbers 
defined by the grammar $\T{Nat} \gramdef S(\T{Nat}) \mid Z$.
\autoref{fig:cyclic-syntax} gives the syntax and semantics of cyclic ECTAs. 
The recursive definition of $\dennode{n}$ should be interpreted with least-fixed-point semantics
(as it may unfold arbitrarily many times). 
Note that this grammar excludes nodes like $\recnode{x}{x}$ or $\recnode{x}{\recnode{y}{x}}$, which would be meaningless. 
We again assume implicit sharing of sub-trees.%
\footnote{However, this pseudo-tree representation precludes sharing of some nodes 
which would be shared in a true graph representation.}

\begin{wrapfigure}[21]{r}{0.38\textwidth}
  \centering
  \begin{subfigure}[b]{0.18\textwidth}
  \centering
  \includegraphics[scale=0.33]{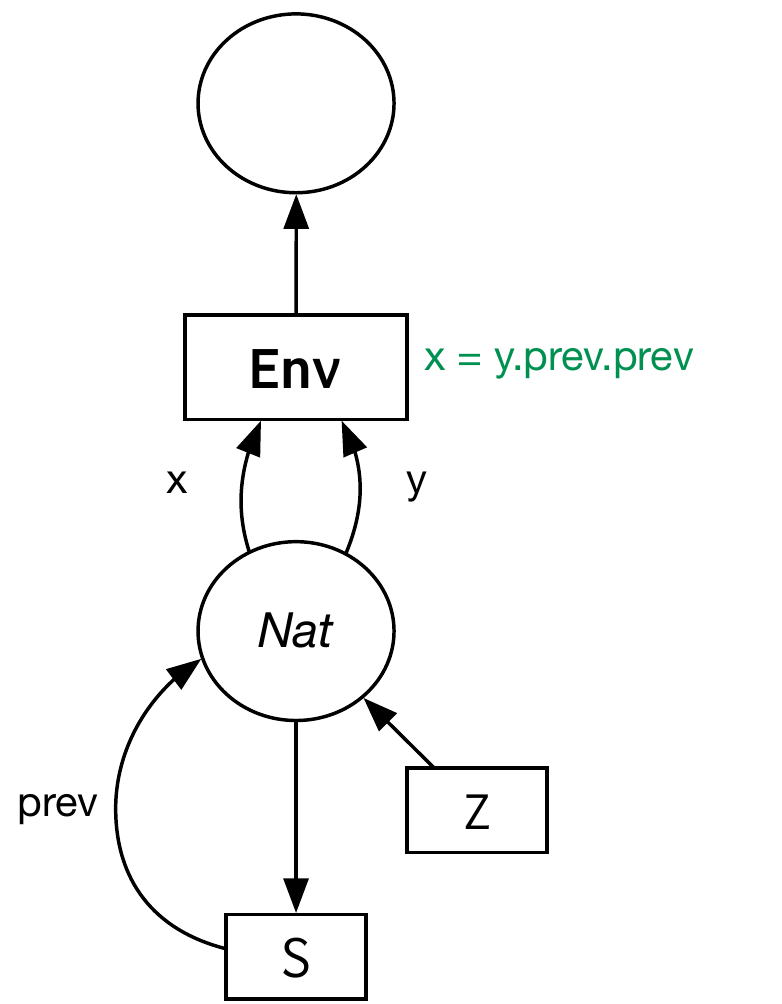}
  \subcaption{}
  \label{fig:cyclic-example:orig}
  \end{subfigure}
  ~
  \begin{subfigure}[b]{0.18\textwidth}
  \centering
  \includegraphics[scale=0.33]{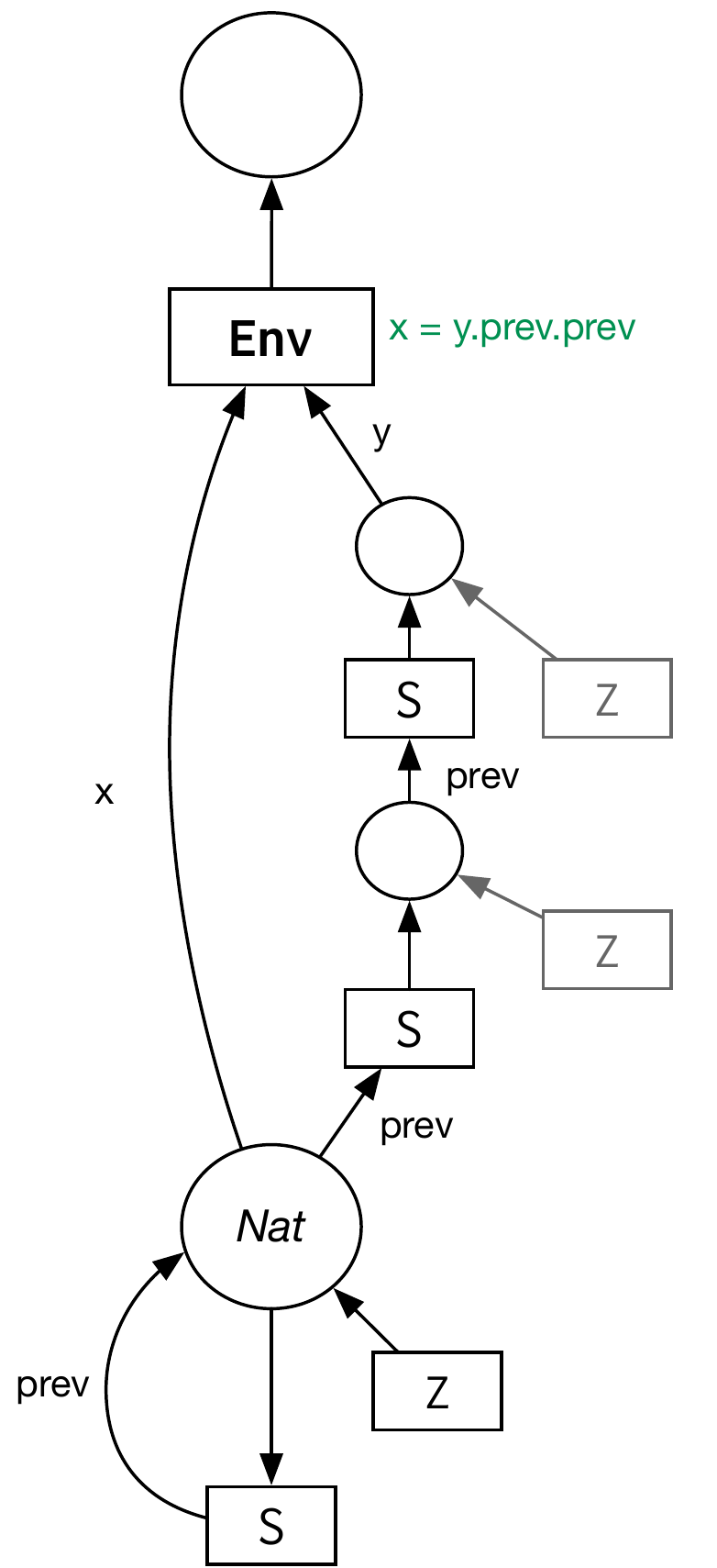}
  \subcaption{}
  \label{fig:cyclic-example:lasso}
  \end{subfigure}
  \caption{(a) ECTA representing an environment with two arbitrary natural numbers $x$ and $y$, where $y=x+2$. The \T{Nat} node is represented $\recnode{x}{\node{\edgeuc{S}{x}, \edgeucnull{Z}}}$.  (b) The ECTA unfolded into lasso form. The grayed-out transitions will be removed by static reduction.}
  \label{fig:cyclic-example}
\end{wrapfigure}
Cyclic ECTAs become unwieldy when constraints are allowed inside cycles.
As a recursive node $\recnode{x}{n}$ is repeatedly unfolded and its constraints duplicated, 
it can yield an arbitrarily large constraint system whose smallest solution may be arbitrarily large. 
In fact, in this general case, ECTA emptiness is undecidable:

\begin{restatable}[Undecidability of \textsc{Emptiness}]{theorem}{thmgenundecidable}
\label{thm:cyclic-undecidability}
Determining whether $\dennode{n}=\emptyset$ is undecidable.
\end{restatable}
\begin{proof}
By reduction from the Post Correspondence Problem. More explanation in \appref{app:proofs:undecidable}.
\end{proof}

\noindent
This motivates a restriction barring constraints on cycles,
which guarantees that the constraint system remains finite 
and efficient algorithms remain possible.
More formally:

\begin{definition}[Finitely-constrained ECTA]
An ECTA $n$ is \emph{finitely-constrained} if, for all recursive nodes $\recnode{x}{m}$ reachable from $n$, $m=\skeleton{m}$.
\end{definition}

\noindent
We assume henceforth that all ECTAs are finitely-constrained.
What makes such ECTAs tractable is that, after sufficient unfolding, 
they enter what we call \emph{lasso form}:

\begin{definition}[Lasso form]
An ECTA $n$ is in \emph{lasso form} if it contains no constrained recursive nodes 
(\ie, no path constraint references a recursive node).
\end{definition}

\noindent
An ECTA in lasso form is split into a top portion, which contains constraints but no cycles, 
and a bottom portion, which contains cycles but no constraints. 
The top portion permits only finitely many choices, 
while the bottom portion can be enumerated and intersected as in classic tree automata theory. 
While they may perform intersection on entire subautomata, 
neither static nor dynamic reduction directly inspect nodes beneath the deepest constraint. 
Hence, with the updated definition of intersection, 
our definitions for both static and dynamic reduction work unmodified on ECTAs in lasso form. 
An example ECTA in lasso form is in \autoref{fig:cyclic-example:lasso}.

\subsection{Algorithms for Cyclic ECTAs}

\begin{figure}[t]  
  \begin{minipage}{.43\textwidth}
  \small  
  \centering
  \textbf{Reachable Nodes}
  \[
    \nodesat(\recnode{x}{n}, p) = \nodesat(\unfold{\recnode{x}{n}}, p)
  \]
  \textbf{Intersection at a Path}
  \[
    \intersectatpath{(\recnode{x}{n})}{n'}{p} = \intersectatpath{\unfold{\recnode{x}{n}}}{n'}{p}
  \]
  \end{minipage}%
  \begin{minipage}{.57\textwidth}
  \small
  \centering
  \textbf{Enumeration step}\quad$\boxed{\steps{\tau}{\tau'}}$
  \begin{gather*}
  \inference[\textsc{Choose-$\mu$}]
  {\Phi \neq \emptyset}
  { \steps{\unenum{\recnode{x}{n}}{\Phi}}{\unenum{\unfold{\recnode{x}{n}}}{\Phi}}}
  \end{gather*}  
  \end{minipage}
  \caption{Extensions to prior algorithms to account for recursive nodes}
  \label{fig:cyclic-algorithms}
\end{figure}

\mypara{Intersection}
One formulation of intersection for classic string automata is a depth-first search
that begins from a pair of initial or final states 
and finds all reachable pairs of states. 
We use this idea to extend our previous definition of intersection to cyclic ECTAs: 
the algorithm tracks all previous visited pairs of nodes, 
and creates a recursive reference upon seeing the same pair twice.

More formally, we define $n_1 \intersect n_2$ in terms of a helper operation, 
$\recintersect{n_1}{n_2}{S}$ (intersection tracking the set of previously-visited pairs); 
which in turn invokes the helper $\recintersectinner{n_1}{n_2}{S}$. 
We assume a function $\varOf(\{n_1, n_2\})$ mapping an unordered pair of nodes 
to a unique named variable for that pair. 
We use the notation $n_1 \in n_2$ to mean that some descendant of $n_1$ is equal to $n_2$.

Let $n_1, n_2 \in N$ be two ECTAs, $S \subseteq \unorderedpairs{N}$ be a set of unordered pairs of ECTAs, 
and define $S' = S \cup \{\{n_1, n_2\}\}$. Then define:
\[
  \recintersect{n_1}{n_2}{S} = \begin{cases}
                                    \varOf(n_1, n_2) & \{n_1, n_2\} \in S \\
                                    \recnode{z}{\recintersectinner{n_1}{n_2}{S'}} & z=\varOf(\{n_1, n_2\}) \wedge z \in (\recintersectinner{n_1}{n_2}{S'}) \\
                                    \recintersectinner{n_1}{n_2}{S'} & \text{otherwise}
                                  \end{cases}
\]
The remainder of the definition is almost identical to the definition for acyclic ECTAs, 
except that recursive nodes are first unfolded. 
Let $\unfold{n_1}=\node{\many{e_1}}$ and $ \unfold{n_2}=\node{\many{e_2}}$. Then:
\[
  \recintersectinner{n_1}{n_2}{S} = \nodesymbol\left(\left\{\recintersect{e_1^i}{e_2^j}{S} \bigmid e_1^i\in\many{e_1}, e_2^j\in \many{e_2} \right\}\right)
\]
  
Let $e_1=\edge{s_1}{[n_1^0\dots n_1^{k-1}]}{C_1}$, $e_2=\edge{s_2}{[n_2^0\dots n_2^{l-1}]}{C_2}$ be two transitions, then:
\[
  \recintersect{e_1}{e_2}{S} = \begin{cases}
                   \edge{s_1}{[\recintersect{n_1^0}{n_2^0}{S},\dots, \recintersect{n_1^{k-1}}{n_2^{k-1}}{S}}{C_1 \cup C_2} & \text{if}\ s_1 = s_2\ \text{and}\ C_1 \cup C_2\ \text{is consistent} \\
                   \edgebot & \text{otherwise}
                 \end{cases}
\]

Now, define $n_1 \intersect n_2 = \recintersect{n_1}{n_2}{\emptyset}$.

\mypara{Static Reduction}
Recall from \secref{sec:static-reduction} that static reduction is defined in terms
intersection at a path, which in turn relies on the definition of reachable nodes. 
\autoref{fig:cyclic-algorithms} extends these operations to unfold recursive nodes until the ECTA enters lasso form, 
at least with regards to the PEC under consideration. 
Then the rest of the static reduction algorithm remains unchanged.

\mypara{Enumeration}
Adapting enumeration to cyclic ECTAs requires a single change: 
the new \textsc{Choose-$\mu$} rule in \autoref{fig:cyclic-algorithms} unfolds recursive nodes referenced by some ancestor's constraint. 
This rule continues unfolding such nodes so long as they are referenced by a parent's constraint, 
at which point it is a fully enumerated node.
Note that a fully-enumerated state will necessarily be in lasso form.
\section{Applications}
\label{sec:applications}

This section gives two examples of problem domains that can be reduced to ECTA enumeration:
boolean satisfiability (SAT; \secref{sec:applications:sat})
and type-driven program synthesis (\secref{sec:applications:hplus}).
The second domain has already been introduced informally in \secref{sec:overview};
here we present its encoding in full generality,
and in \secref{sec:eval} we evaluate our encoding against a state-of-the-art synthesizer \hplus. 
The purpose of presenting the first domain is to demonstrate the versatility of ECTAs,
not to compete with highly-engineered industrial SAT solvers;
hence we leave the SAT domain out of empirical evaluation.

\subsection{Boolean Satisfiability}\label{sec:applications:sat}

\begin{figure}
  \centering
  \begin{subfigure}[t]{0.4\textwidth}
  \centering
  \includegraphics[scale=0.4]{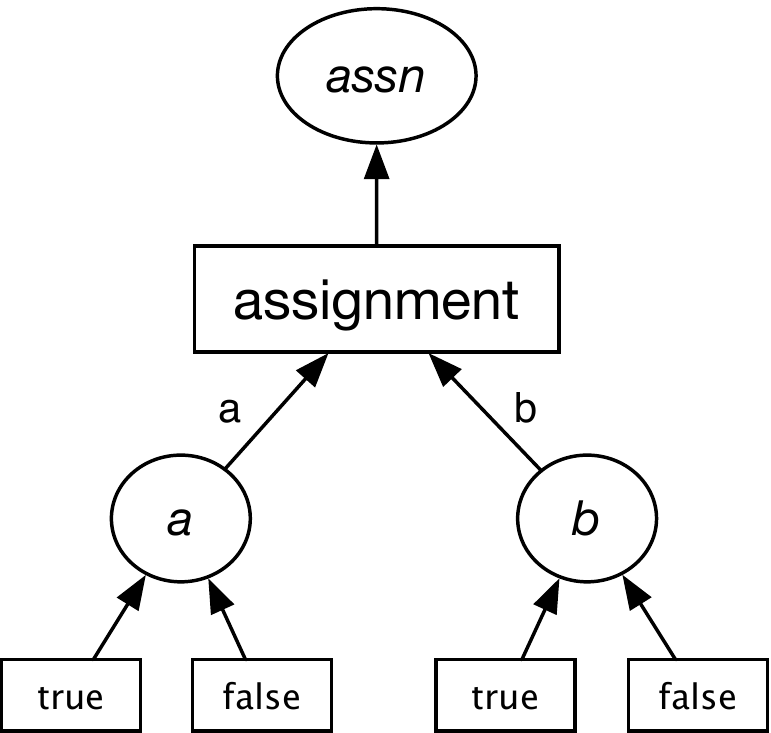}
  \subcaption{Variable assignment}
  \label{fig:sat-assn}
  \end{subfigure}
  ~ 
  \begin{subfigure}[t]{0.6\textwidth}
  \centering
  \includegraphics[scale=0.4]{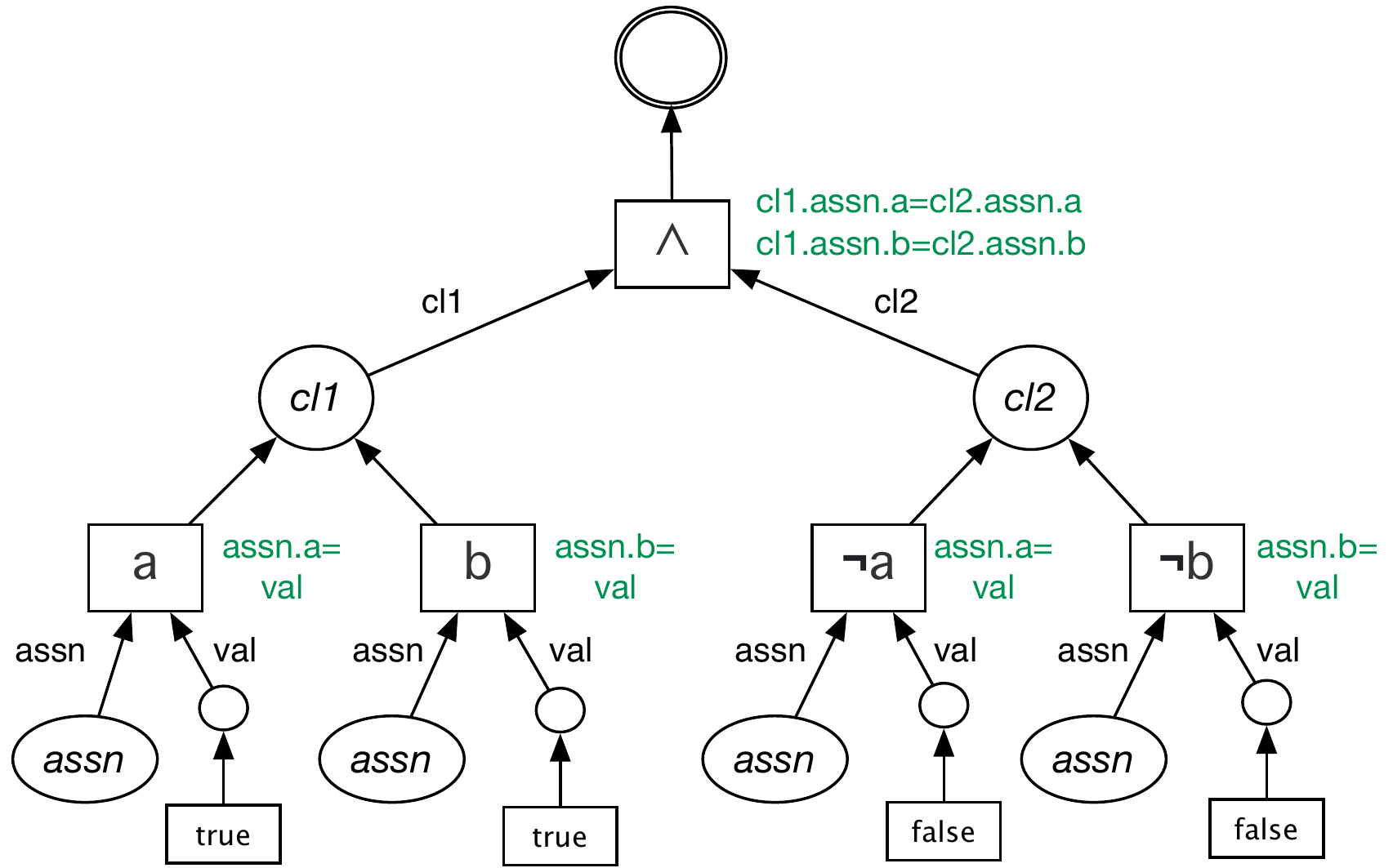}
  \subcaption{CNF formula}
  \label{fig:sat-formula}
  \end{subfigure}
  \caption{ECTA encoding of a CNF formula $(a \vee b) \wedge (\neg a \vee \neg b)$}
  \label{fig:sat}
\end{figure}

\mypara{Problem Statement}
Given a propositional formula in \emph{conjunctive normal form} (CNF),
the SAT problem is to find a satisfying assignment to its variables.
A CNF formula is a conjunction of \emph{clauses},
where each clause is a disjunction of \emph{literals},
and each literal is either a variable or its negation.
For example, the CNF formula $(a \vee b) \wedge (\neg a \vee \neg b)$
has two satisfying assignments: $\{a, \neg b\}$ and $\{\neg a, b\}$.

\mypara{Encoding}
\autoref{fig:sat} illustrates our ECTA encoding for the above formula.
The sub-automaton \staten{assn} in \autoref{fig:sat-assn}
represents the set of all possible variable assignments.
The \T{assignment} transition has one child per variable,
and each variable node has two alternatives: \T{true} and \T{false};
hence, to extract a term from \staten{assn} one must pick a value for each variable.

The ECTA for the entire CNF formula is shown in \autoref{fig:sat-formula};
this ECTA has a single top-level conjunction transition $\wedge$, 
with one child per clause.
Each clause node has one alternative per literal in that clause:
the choice between these alternatives corresponds to picking 
which literal is responsible for making the clause true.
Each literal transition---such as \T{a} or $\neg \T{a}$---has two children:
\T{assn} is its local copy of the assignment sub-automaton
and \T{val} is the Boolean value that this literal assigns to its variable.
The constraint on the literal---such as \T{assn.a = val}---%
restricts its local assignment in such a way that the literal evaluates to true.
%
Finally, the constraints on the $\wedge$ transition force all local assignments to coincide; note that, while the various \T{assn} nodes are shared in memory, each occurrence of this node is an independent choice unless so constrained.
The reader might be wondering why we chose to split these constraints per-variable
instead of simply writing \T{cl1.assn = cl2.assn};
as we explain next, this helps enumeration discover inconsistent assignments quickly.

\mypara{SAT Solving as ECTA Enumeration}
With this encoding, the general-purpose ECTA enumeration algorithm from \secref{sec:enumeration} 
turns into a SAT solver.%
\footnote{A curious reader might be wondering why don't we go the other direction:
encode an ECTA into a SAT formula and use a SAT solver for ECTA enumeration;
this is not possible in general, as we discuss in more detail in \secref{sec:conclusions}.
}
Specifically, once \textsc{Enumerate} has found a fully enumerated state,
the satisfying assignment can be read off the children of any \T{assignment} symbol in that state;
note that if we let \textsc{Enumerate} run past the first result,
it will enumerate all satisfying assignments, modulo irrelevant variables (see \secref{sec:enum:balanced}).

The overall solving procedure amounts to choosing a literal from each clause
and backtracking whenever the assignment becomes inconsistent.
For example, suppose \textsc{Enumerate} has chosen \T{a} from \staten{cl1};
the enumeration state $\sigma$ now contains variables $v_a$ and $v_b$,
which store assignments for $a$ and $b$ consistent with the current choices
(that is, $v_a$ is restricted to \T{true}, while $v_b$ still allows both choices).
If the algorithm now attempts to make an inconsistent choice of $\neg\T{a}$ from \staten{cl2},
this inconsistency is discovered immediately 
when $\neg\T{a}.\T{val}$ is suspended and intersected with $v_a$.

\subsection{Type-Driven Program Synthesis}\label{sec:applications:hplus}

\mypara{Problem Statement}
We are interested in the following \emph{type-driven program synthesis} problem:%
\footnote{This problem is also known as \emph{type inhabitation}~\cite{urzyczyn97} and \emph{composition synthesis}~\cite{rehof16}.}
given a type $T$, called the query type, 
and a components library $\Lambda$, which maps component names to their types,
enumerate terms of type $T$ built out of compositions of components from $\Lambda$.
For example, a Haskell programmer might be interested in a code snippet that,
given a list of optional values, finds the first element that is not \T{Nothing}
(and returns a default value if such an element does not exist).
The programmer might pose this as a type-driven synthesis problem,
where $\Lambda$ is the Haskell standard library, and the query $T$ is $\T{a} \to \T{[Maybe a]} \to \T{a}$.
Given this problem, the state-of-the-art type-driven synthesizer \hplus~\cite{guo2020tygar,james2020digging}
returns a list of candidate programs that includes the desired solution:
$\lambda \T{def}\ \T{mbs} \to \T{fromMaybe def (listToMaybe (catMaybes mbs))}$.

In this section we adopt the setting of \hplus,
where components can be both \emph{polymorphic} and \emph{higher-order},
both of which make the synthesis problem significantly harder.
On the other hand, also following \hplus,
we do not consider synthesis of inner lambda abstractions:
that is, arguments to higher-order functions can be partial applications but not lambdas.

\mypara{\hplus Limitations}
\hplus works by encoding a synthesis problem into a data structure called \emph{type-transition net}:
a Petri net, where places (nodes) correspond to types, and transitions correspond to components;
the synthesis problem then reduces to finding a path from the input types to the output type of the query.
This encoding has two major limitations:
\begin{enumerate}
  \item \emph{No native support for polymorphic components.}
  In the presence of polymorphism, the space of types that can appear in a well-typed program becomes infinite.
  Because types are encoded as places, a finite Petri net cannot represent all candidate programs.
  Instead, \hplus employs a sophisticated abstraction-refinement loop
  to build a series of Petri nets that encode increasingly precise approximations of the set of types of interest.
  \item \emph{No native support for higher-order components.}
  Because components are encoded as transitions with a fixed arity---%
  they transform a fixed number of types into a single type---%
  all components must always be fully applied.
  This precludes the use of higher-order components: 
  for example, in \T{foldr (+) 0 xs}, the binary component \T{(+)} is not fully applied.
  To circumvent this limitation, \hplus must add a separate \emph{nullary copy} of the \T{(+)} component to the library.
  Since these duplicate components bloat the library and slow down synthesis,
  in practice only a few popular components are duplicated,
  thereby limiting the practicality of the synthesizer.
\end{enumerate}
In this section we present \hpp (\hplus: ECTA REvision),
our encoding of type-driven synthesis as ECTA enumeration.
This encoding has native support for both polymorphic and higher-order components,
without the need for an expensive refinement loop or duplicate components.

\begin{figure}[t]
  \centering
  \hskip1em
  \parbox{0.8\textwidth}{
    \parbox{.456\textwidth}{%
      \subcaptionbox{Variables (library components)\label{fig:apps-hplus-vars}}{\includegraphics[width=\hsize]{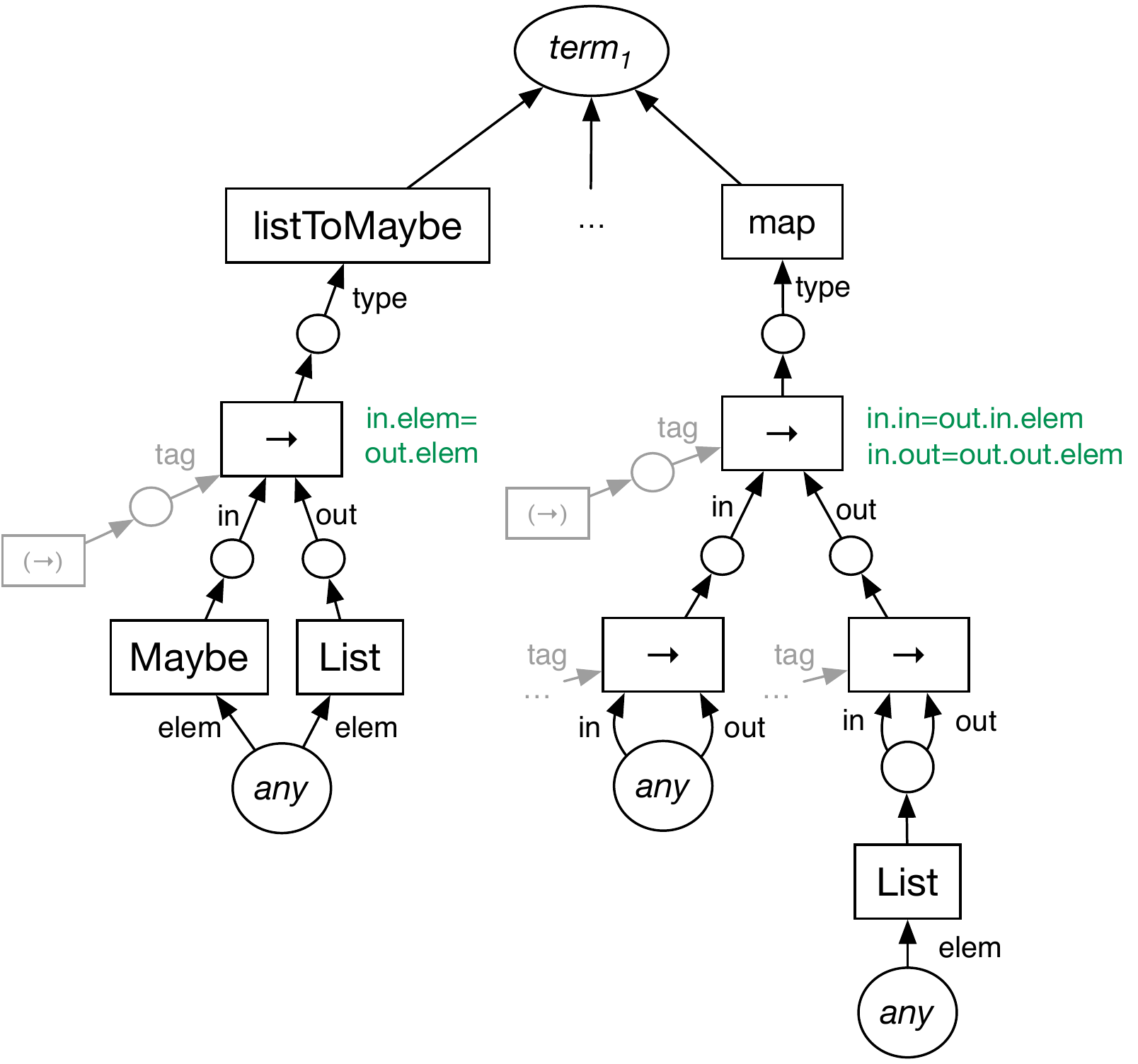}}
    }
    \hskip1em
    \parbox{.304\textwidth}{%
      \subcaptionbox{Types\label{fig:apps-hplus-types}}{\includegraphics[width=\hsize]{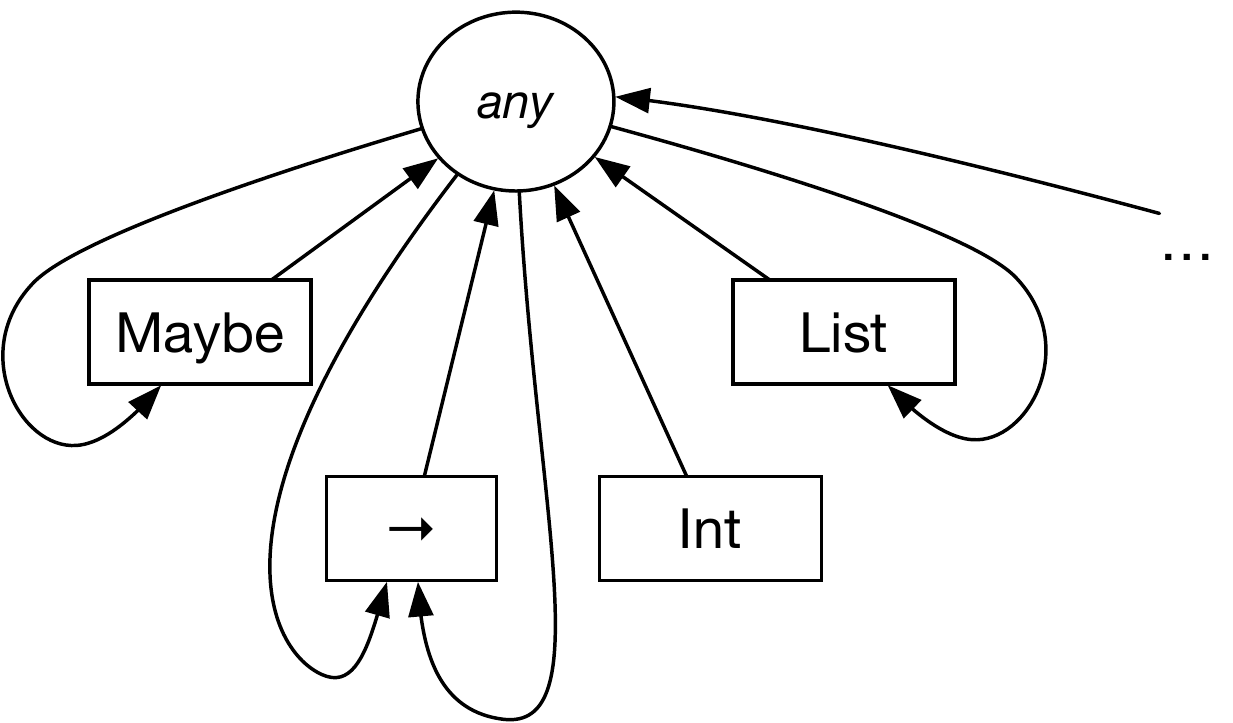}}
      \vskip1em
      \subcaptionbox{Size-two terms\label{fig:apps-hplus-terms}}{\includegraphics[width=\hsize]{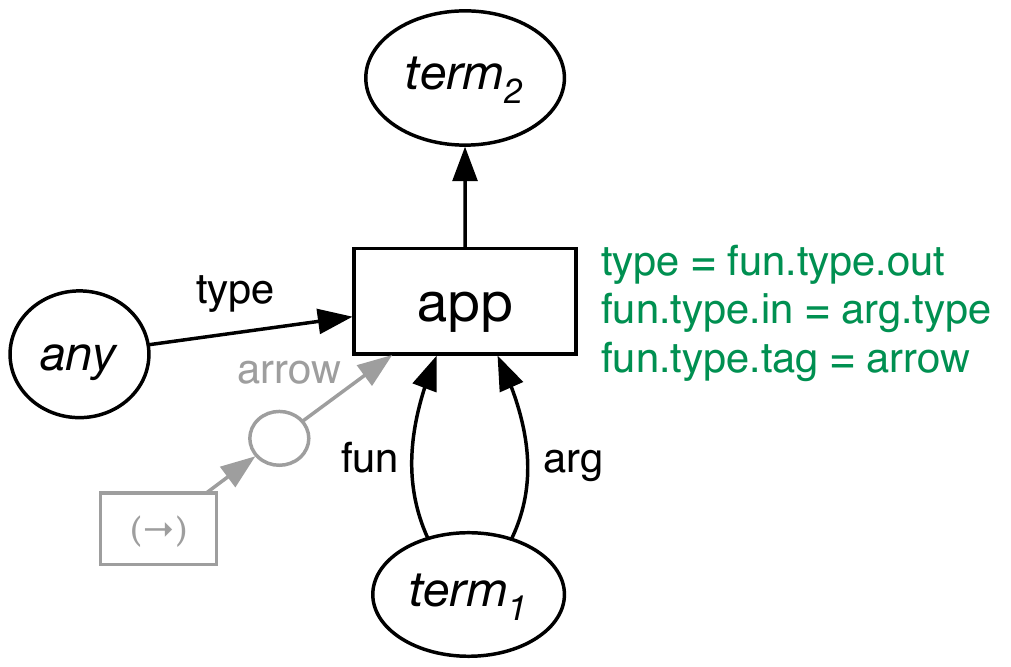}}  
    }
  }
  \caption{Encoding of variables, types, and fixed-size terms in \hpp.}\label{fig:apps-hplus}
\end{figure}

\mypara{Encoding Types}
Recall that \secref{sec:overview} (\autoref{fig:overview-ex3}) introduced an encoding for a limited form of polymorphism, 
where the type variable $\alpha$ in a type like $\alpha \to \alpha$ 
could be instantiated only with base types. 
We now generalize this encoding so that $\alpha$ can be instantiated with \emph{any type},
with arbitrarily nested applications of type constructors.
%
The infinite space of all types can be finitely encoded as a recursive node \staten{any},
as shown in \figref{fig:apps-hplus-types}.
The \staten{any} node has one child per type constructor in $\Lambda$,
with non-nullary type constructors looping back to \staten{any}.
Now the type $\alpha \to \alpha$ can be represented as a $\to$ transition,
whose children are both \staten{any} (and are constrained to equal each other).

\mypara{Encoding Components}
The simplified encoding in \secref{sec:overview} splits components into different nodes by their arity
(\eg the nodes \staten{scalar} and \staten{unary} in \autoref{fig:overview-ex1});
this was necessary given our simplified encoding of function types,
but as we mentioned above, such arity-specific encoding precludes partial applications.
\autoref{fig:apps-hplus-vars} illustrates the generalized encoding of components in \hpp.
Here all components, regardless of arity, are gathered in single node $\text{\staten{term}}_1$.
Each component is annotated with its \T{type};
function types are represented using the $\to$ transition with two child types, \T{in} and \T{out}
(for now, ignore the grayed out edges labeled \T{tag}, we explain those below).
\autoref{fig:apps-hplus-vars} showcases the type encoding for two polymorphic components:
$\T{listToMaybe} :: [\alpha]\to \T{Maybe}\ \alpha$ and
$\T{map} :: (\alpha \to \beta) \to [\alpha] \to [\beta]$;
as before, all occurrences of the same type variable are related by equality constraints, shown in green.

\mypara{Encoding Applications}
\autoref{fig:apps-hplus-terms} illustrates the \hpp encoding of size-two terms.
As before, the application transition \T{app} has two children \T{fun} and \T{arg},
but now they are both represented by the same node $\text{\staten{term}}_1$;
hence this encoding supports partial applications, such as \T{map listToMaybe}.

We now explain the purpose of the grayed-out parts of \autoref{fig:apps-hplus}.
The \T{app} node must ensure that its \T{fun} child has an arrow type.
The mere presence of \T{fun.type.in} and \T{fun.type.out} in its first two constraints does not suffice:
recall that the actual \ectalibrary library refers to children by index instead of by name,
and hence any other binary type constructor (such as \T{Either}) could satisfy those two constraints.
This would lead to accepting ill-typed programs, such as $\T{Left}\ x\ y$.
To circumvent this issue, we introduce a special tag transition $(\to)$,
which occurs nowhere else but as a first child of every $\to$ transition;
by constraining the first child of \T{fun} to be $(\to)$,
\T{app} effectively ensures that it is indeed a function (see the last constraint on \T{app}).

This encoding of application terms generalizes from size-two terms to terms of arbitrary fixed size $n$ as follows:
the node $\text{\staten{term}}_n$ has $n - 1$ incoming \T{app} transitions,
where the $i$-th transition ($i \in 1 \twodots n-1$) has children $\text{\staten{term}}_i$ and $\text{\staten{term}}_{n-i}$.

\mypara{Synthesis Algorithm}
So far we have discussed how to encode the space of all well-typed terms of size $n$.
Let us now proceed to the top-level synthesis algorithm of \hpp.
Given a query type, such as $\T{a} \to [\T{Maybe a}] \to \T{a}$,
\hpp first adds the inputs of the query (here $\T{def} :: \T{a}$ and $\T{mbs} :: [\T{Maybe a}]$)
to the node $\text{\staten{term}}_1$, as if they were components.
The algorithm then iterates over program sizes $n \geq 1$;
for each size $n$, it constructs the ECTA $\text{\staten{term}}_n$
and restricts its top-level type to the return type of the query (here \T{a}),
following the recipe illustrated in \autoref{fig:overview-ex2}.
The algorithm then statically reduces all constraints in the restricted ECTA and enumerates all terms accepted by the resulting reduced ECTA,
before moving on to the next size $n$.
Note that the type variables \emph{of the query} (here \T{a})
are represented as type constructors and not as the \staten{any} node,
since those type variables are universally quantified.

\mypara{Enforcing Relevancy}
Existing type-driven synthesizers~\cite{FengM0DR17,guo2020tygar}
restrict synthesis results to \emph{relevantly typed} terms---%
that is, terms that use all the inputs of the query.
Without such relevancy restriction, any synthesis algorithm gets bogged down
by short but meaningless programs.
\hpp enforces relevancy via a slight modification to the simple synthesis algorithm outlined above:
it splits every $\text{\staten{term}}_n$ node into $2^k$ nodes, 
where $k$ is the number of inputs in the query.
In our example, there are four nodes at each term size:
$\text{\staten{term}}_n^{\{\texttt{def}, \texttt{mbs}\}}, 
\text{\staten{term}}_n^{\{\texttt{def}\}}, 
\text{\staten{term}}_n^{\{\texttt{mbs}\}}, 
\text{\staten{term}}_n^{\emptyset}$,
each representing terms that must mention the corresponding set of inputs.
When constructing a new term node, say $\text{\staten{term}}_2^{\{\texttt{def}\}}$,
\hpp considers all applications of $\text{\staten{term}}_1^P$ to $\text{\staten{term}}_1^Q$ such that $P \cup Q = \{\T{def}\}$.
At the top level, only $\text{\staten{term}}_n^{\{\texttt{def}, \texttt{mbs}\}}$ is connected to the accepting node.
Although the number of term-nodes in this encoding grows exponentially with the number of inputs,
this is not a problem in practice, since the number of inputs is typically small;
note also that due to hash consing in \ectalibrary, the overlapping component sets are not actually duplicated.

\section{Evaluation}\label{sec:eval}

As we explained in \secref{sec:applications:hplus},
we used the \ectalibrary library to implement \tool, a type-driven component-based synthesizer for Haskell.
In this section, we evaluate the performance of \tool 
and compare it with the state-of-the-art synthesizer \hplus,
based on an SMT encoding of Petri-net reachability~\cite{guo2020tygar}.
Both tools are written in Haskell,
but the \hplus implementation (excluding tests and parsing) contains a whopping 4000 LOC, 
while \tool only contains 400. 
Although code size is an imperfect measure of development effort,
these numbers suggest that the \ectalibrary library 
has the potential to significantly simplify the development of program synthesizers.

%
We designed our evaluation to answer the following research questions:
\begin{enumerate}[\textbf{(RQ\arabic*)}]
    \item How does \tool compare against \hplus on existing and new benchmarks?
    \item How significant are the benefits of static and dynamic reduction in program synthesis?
\end{enumerate}

All experiments were conducted on an Intel Core i9-10850K CPU with 32 GB memory.

\begin{table}[]
    \footnotesize
    \centering
    \renewcommand{\arraystretch}{1.25}
    \caption{Three sample queries and corresponding solutions from two benchmark suites.}
    \label{tab:eval:example-query}
    \begin{tabular}{c|c|c|c}
        \toprule
        Suite & Name & Query & Expected solution \\
        \midrule
        \multirow{3}{*}{\rotatebox{90}{\hplus}} & \T{mergeEither} & \T{Either a (Either a b) -> Either a b} & \T{\\e -> either Left id e} \\
         & \T{headLast} & \T{[a] -> (a, a)} & \T{\\xs -> (head xs, last xs)} \\
         & \T{both} & \T{(a -> b) -> (a, a) -> (b, b)} & \T{\\f p -> (f (fst p), f (snd p))} \\
        \midrule
        \multirow{3}{*}{\rotatebox{90}{\parbox{1.3cm}{\textsc{Stack-}\\\textsc{Overflow}}}} & \T{multiIndex} & \T{[a] -> [Int] -> [a]} & \T{\\xs is -> map ((!!) xs) is} \\
         & \T{splitOn} & \T{Eq a => a -> [a] -> [[a]]} & \T{\\x xs -> groupBy (on (&&) (/= x)) xs} \\
         & \T{matchedKeys} & \T{(b -> Bool) -> [(a, b)] -> [a]} & \T{\\p xs -> map fst (filter (p . snd) xs)} \\
        \bottomrule                
    \end{tabular}
\end{table}

\subsection{Comparison on \hplus Benchmarks}
\mypara{Experiment Setup}
For our main experiment, we compare the two synthesizers on the benchmark suite
from the latest \hplus publication~\cite{james2020digging}.
This suite includes \nHplus synthesis queries, 
and a library of \nComps components from 12 popular Haskell modules.
These benchmarks are non-trivial:
the expected solutions range in size from 3 to 9, with the average size of 4.7; 
\polypercent of the components are polymorphic,
and \nHplusHO of the queries require using a higher-order component.
Three samples queries from this suite are listed at the top of \autoref{tab:eval:example-query}.
These queries have sizes 4, 5, and 7 respectively, 
and \T{mergeEither} uses a higher-order component \T{either}. 

Both \hplus and \tool yield candidate programs one at a time,
gradually increasing the size of the programs they consider.
For both tools, our test harness terminates the search once the expected solution has been found
(or the timeout of \nTimeout has been reached).
We report the average time to expected solution over \nRepeat runs.
We configured \tool to perform static reduction on all constructs prior to running the fast enumeration procedure of \secref{sec:enumeration:algorithm}, repeating this operation up to 30 rounds or until the automaton converges.

\begin{figure}
    \centering
    \begin{minipage}{.48\textwidth}      
        \centering
        \includegraphics[width=\textwidth]{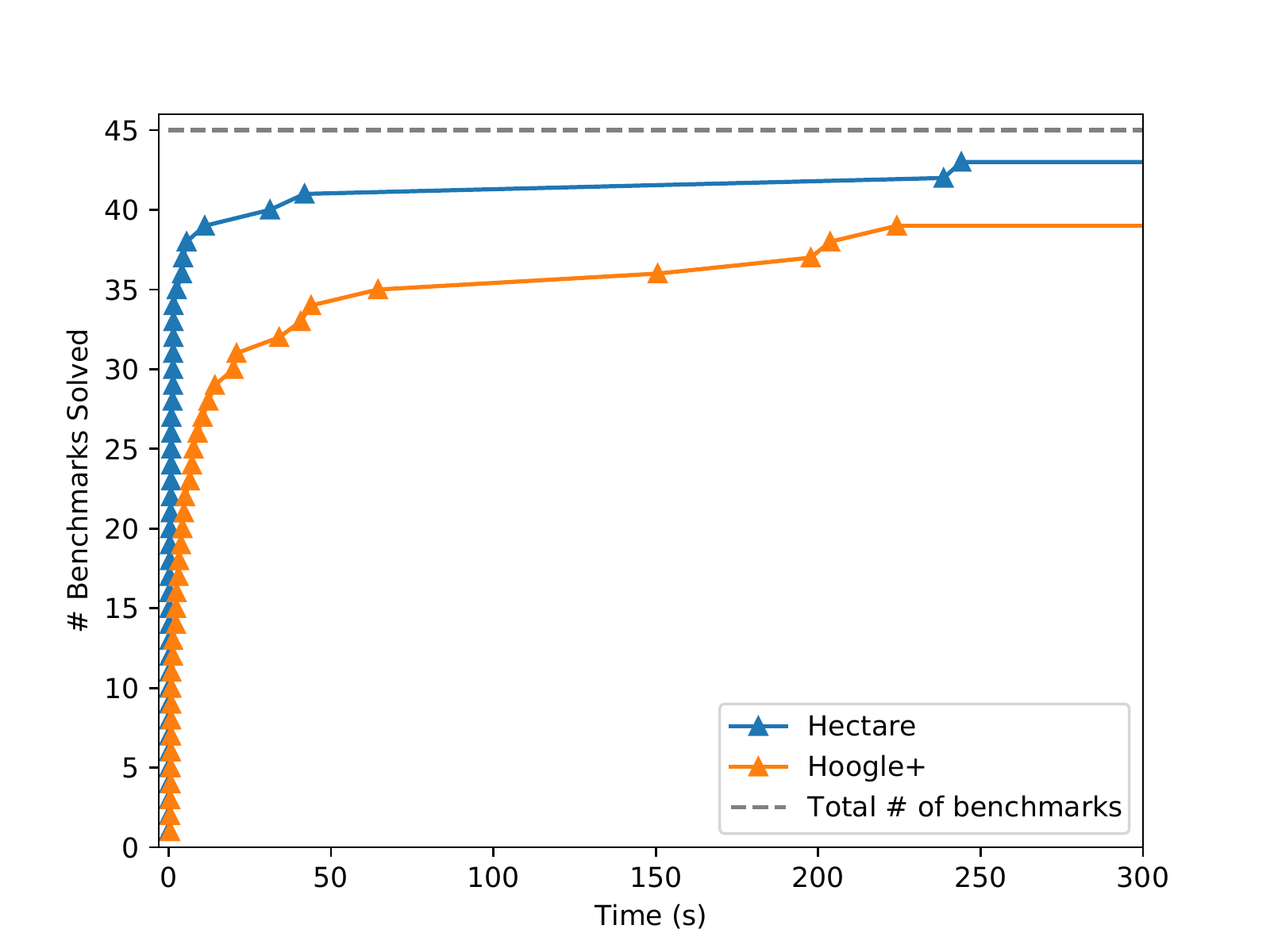}
        \caption{Benchmarks solved vs time for \hpp and \hplus on \hplus benchmarks.}
        \label{fig:hplus-cactus}
    \end{minipage}%
    \hfill%
    \begin{minipage}{.48\textwidth}
        \centering
        \includegraphics[width=\textwidth]{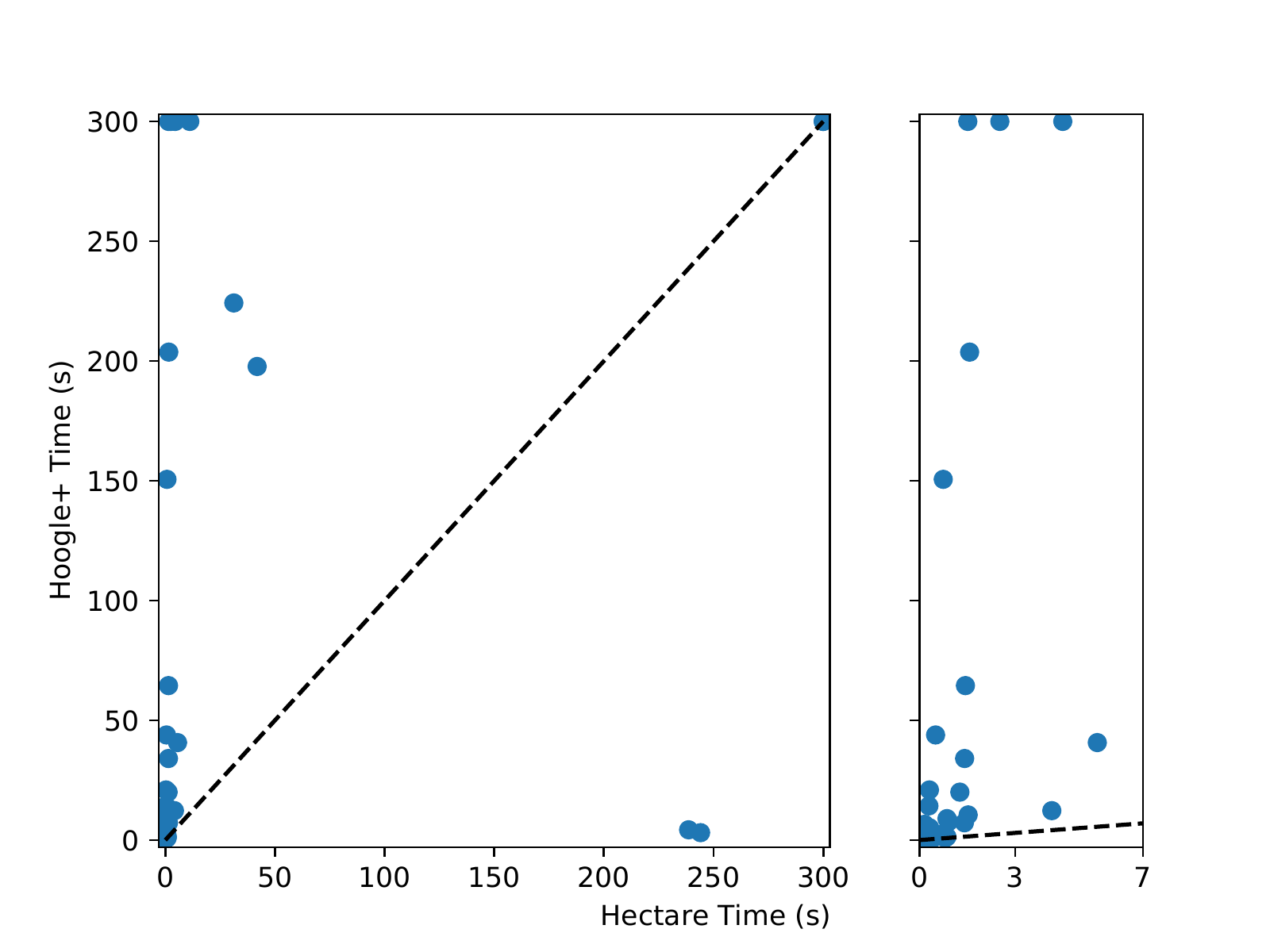}
        \caption{Synthesis times of \tool against \hplus on \hplus benchmarks.}
        \label{fig:hplus-scatter}
    \end{minipage}
\end{figure}

\mypara{Results}
\autoref{fig:hplus-cactus} plots the number of benchmarks solved \vs time for both synthesizers.
Within the timeout, \tool solves \nEctaSolved out of \nHplus benchmarks, 
whereas \hplus only solves \nHplusSolved.
Importantly, as we show in \autoref{fig:hplus-scatter}, 
on commonly solved benchmarks \tool is significantly faster:
it achieves an average speedup (geometric mean) of \hplusSpeedup on this suite,
solving all but two tasks faster than \hplus.
%
%
Fast synthesis times are especially important if a synthesizer is to be used interactively.
As shown in the zoomed-in scatter plot in \autoref{fig:hplus-scatter} (right),
\tool also vastly outperforms \hplus if we consider a shorter timeout of seven seconds,
commonly used for interactive synthesizers~\cite{loopy};
in fact, \tool solves 84\% of the benchmarks within seven seconds.


The poor performance of \hplus can be mainly attributed to the brittleness of the abstraction-refinement loop
it uses to support polymorphic components (\secref{sec:applications:hplus}).
For example, the \T{headLast} benchmark from \autoref{tab:eval:example-query}
is one of the queries where \hplus times out, while \tool only takes 2.5 seconds. 
Upon closer inspection,
\hplus is unable to create an accurate type abstraction for this query
and ends up wasting a lot of time enumerating ill-typed terms.
\tool, in contrast, natively supports polymorphic components via recursive nodes,
which leads to more predictable performance.
On the other hand,
the two benchmarks where \tool is slower than \hplus both involve deconstructing a \T{Pair}
and using both of its fields 
(\T{both} from \autoref{tab:eval:example-query} is one of these benchmarks).
\hplus solves these queries using a special treatment of \T{Pair}s:
it introduces a single component that projects both fields of a pair simultaneously,
which makes the solutions to these queries much shorter;
in \tool, we did not find a straightforward way to add this trick.



In general, we conclude that \emph{\tool is effective in solving type-driven synthesis tasks 
and outperforms a state-of-the-art tool on 89\% of their benchmarks with \hplusSpeedup speedup on average}.

\subsection{Comparison on \stackoverflow Benchmarks}
\mypara{Benchmark Selection}
Recall that another limitation of \hplus we discussed in \secref{sec:applications:hplus} 
is its restricted support for higher-order functions.
In fact, the original \hplus configuration
contains only \emph{nine} components whose nullary versions are added to the Petri net
(and which consequently can appear in arguments to higher-order functions).
In order to push the limits of both tools
and demonstrate the benefits of \tool's native encoding,
we assembled an additional benchmark suite focusing on higher-order functions.
To this end, we searched \stackoverflow for Haskell programming questions;
for each question, we attempted to construct an expected solution
using only applications of library components;
we excluded tasks that can be solved without higher-order functions
or require unsupported features (such as higher-kinded type variables and inner lambda abstractions).
This left us with \nStackoverflow synthesis queries.
The new benchmark suite is generally more complex than the original \hplus suite:
expected solutions range in size from 4 to 9, with the average of 6.2;
all of these programs include partial applications as arguments to higher-order components.
Three sample queries are shown at the bottom of \autoref{tab:eval:example-query}.

\mypara{Experiment Setup}
To run the newly collected benchmarks, 
we augmented the original component set from \hplus with seven components required in these benchmarks.
We also created a variant of \hplus called \hplusAll, 
in which we added nullary copies of all components into the Petri net
(\hplusAll thus has the same expressiveness as \tool).
As before, we record the time to expected solution,
repeat the measurement \nRepeat times, 
and report the average time;
to accommodate the increased benchmark complexity, we use a longer timeout of \nTimeoutL.

\begin{figure}
    \centering
    \begin{minipage}{.48\textwidth}
      \centering
        \includegraphics[width=\textwidth]{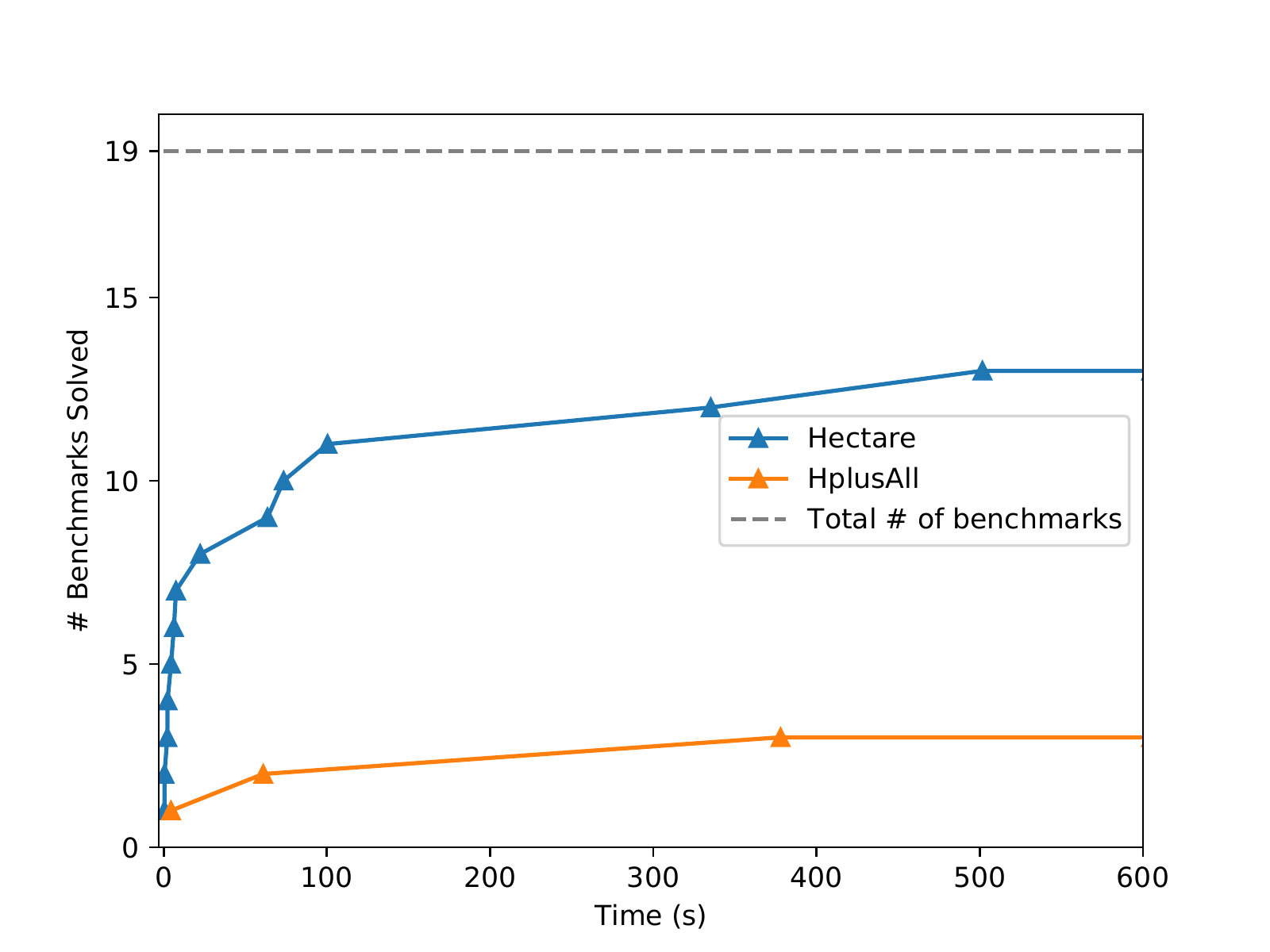}
        \caption{Comparison of synthesis time on higher-order benchmarks between \tool and \hplusAll.}
        \label{fig:compare-ho}
    \end{minipage}%
    \hfill%
    \begin{minipage}{.48\textwidth}
      \centering
        \includegraphics[width=\textwidth]{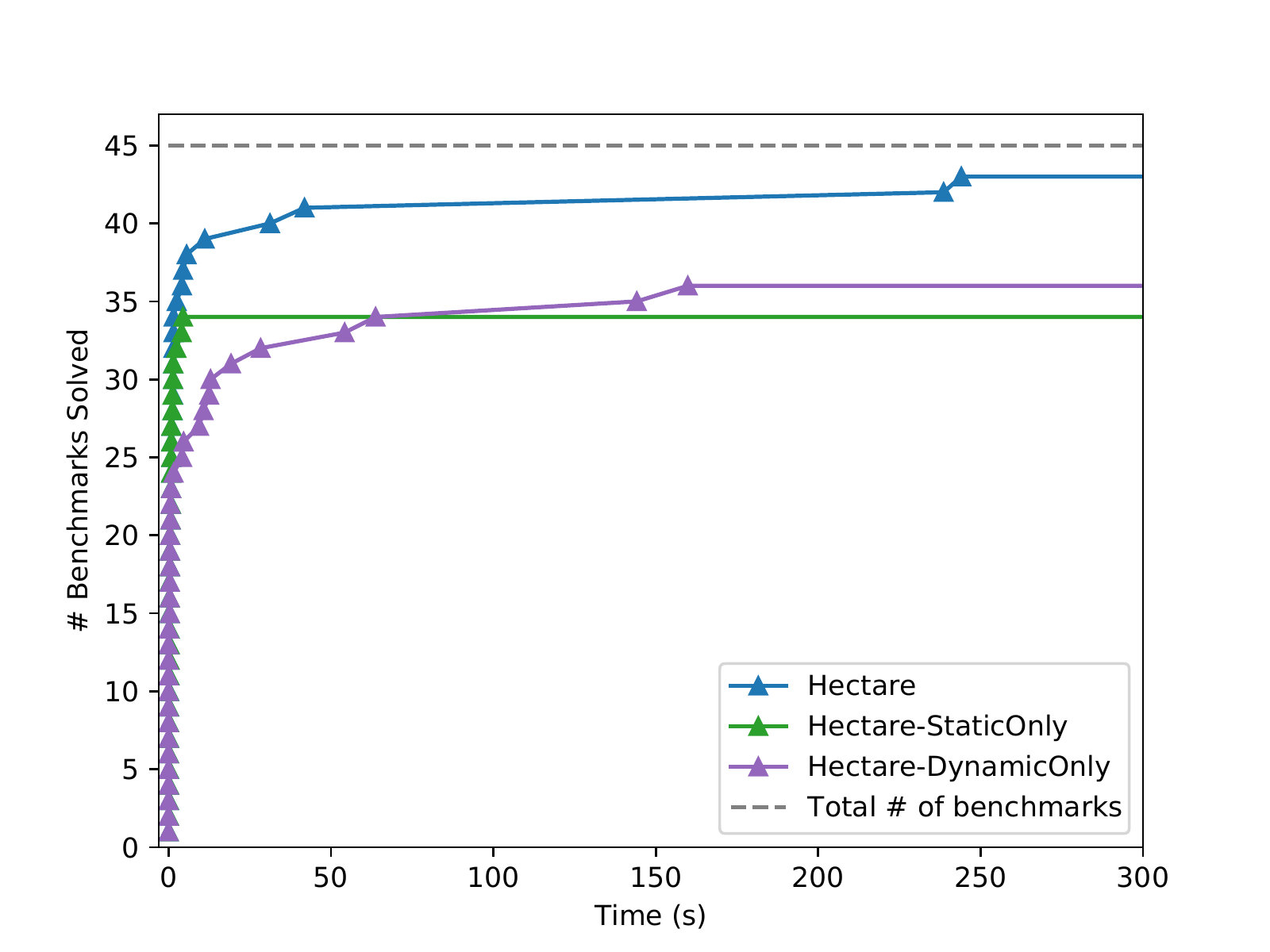}
        \caption{Comparison of synthesis performance between \hpp and its two variants.}
        \label{fig:ablation}
    \end{minipage}
\end{figure}

\mypara{Results}
Unsurprisingly, the original \hplus cannot solve any of the new benchmarks:
most of them require using new components in higher-order arguments
(and the rest are simply too large).
The results for \hplusAll and \tool are shown in \autoref{fig:compare-ho}.
Although the search space of \hplusAll does include all the new benchmarks,
it still fares poorly, solving only \nHplusSolvedSo out of \nStackoverflow.
The reason is that adding nullary versions of all components blows up the Petri net 
and makes the reachability problem intractable.
In contrast, \tool's native support for partial applications
enables it to solve \nEctaSolvedSo out of \nStackoverflow tasks in this challenging suite,
achieving \soSpeedup speedup on the three commonly solved benchmarks.
We therefore conclude that \emph{the benefits of \ectalibrary-based synthesis
are even more pronounced on larger benchmarks focused on higher-order functions}.

\subsection{Benefits of Static and Dynamic Reduction}
\mypara{Experiment Setup}
To isolate the contributions of static and dynamic reduction,
we compare \tool with its three variants:
\hppStatic, \hppDyn, and \hppNone,
which forgo one or both kinds of reduction, respectively.
Specifically, both \hppStatic and \hppNone, use a na\"ive ``rejection-sampling'' enumeration.
Note that in the presence of recursive nodes, such as the \staten{any} type, 
the na\"ive enumeration tends to get ``stuck'',
constructing infinitely many spurious terms and never finding one that satisfies the constraints.
To prevent this behavior, we limit the unfolding depth of recursive nodes to three,
which is sufficient to solve all the benchmarks.
We run the three variants on the \hplus benchmarks with a timeout of \nTimeout
and report the average time to expected solution over \nRepeat runs.

%

\mypara{Results}
\autoref{fig:ablation} plots the number of benchmarks solved \vs time for \hpp and its variants.
\hppNone is omitted from the plot because it cannot solve \emph{any benchmarks}: 
it spends most of its time unfolding the recursive \staten{any} node,
or in other words, blindly going through all possible instantiations of every polymorphic component.
The other two variants fare significantly better:
\hppStatic and \hppDyn are able to solve \nStaticSolved and \nDynSolved tasks, respectively.
That said, as the tasks get harder, \tool still outperforms these variants drastically:
in particular, the variants cannot solve any benchmarks of size six or larger.

A closer look at \autoref{fig:ablation} reveals a curious difference:
\hppStatic is ``all-or-nothing'': it performs as well as \tool on easy benchmarks, but then completely falls flat;
\hppDyn, in contrast, demonstrates a more gradual degradation of performance.
To understand why,
recall that the biggest time sink during enumeration is blindly unfolding the recursive \staten{any} nodes.
Static reduction can sometimes get rid of \staten{any} nodes entirely,
making the resulting ECTA small enough that any enumeration algorithm would do;
when it fails to do so, however, 
na\"ive enumeration spends all its time in \staten{any} nodes.
Dynamic reduction, on the other hand,
provides a more gradual yet robust approach to dealing with \staten{any} nodes,
via \textsc{Suspend}.
In summary, we find that \emph{both static and dynamic reduction individually are critical to the performance of ECTA-based programs synthesis,
and moreover, they complement each other's strengths.}

%

\section{Related Work}
\label{sec:related-work}


\mypara{Constrained Tree Automata}
Tree automata have long been used to represent sets of terms in term rewriting \cite{dauchet1993rewriting,feuillade2004reachability,geser2007tree}, 
and we are not the first to consider adding equality constraints to handle nonlinear rewrites. 
In fact, in 1995, Dauchet introduced a data structure very similar to our ECTAs,
called \emph{reduction automata}~\cite{dauchet1995automata}.
%
In fact, reduction automata are more expressive than ECTAs, 
as they also allow \emph{disequality constraints}, 
including disequalities (but not equalities) on cycles.
Unfortunately, allowing disequalities---or other classes of constraints for that matter---%
precludes efficient static and dynamic reduction based on automata intersection.
For that reason, we consider ECTAs to be a sweet spot:
expressive enough to encode a variety of interesting problems,
yet restricted enough to enable fast enumeration.
%

Other prior work on constrained tree automata 
\cite{bogaert1999recognizability,bogaert1992equality,barguno2010emptiness,barguno2013decidable,reuss2010bottom} 
similarly focuses on theoretical aspects, such as worst-case complexity and decidability results, 
and we have found no reference to these data structures being used in a practical system in the $30$ years since their introduction. 

Attribute grammars \cite{knuth1968semantics, paakki1995attribute, van2010silver} augment context-free grammars with a number of equations of the form $\left<\text{attribute}\right> = \left<\text{expression}\right>$. This notation resembles a constraint system over trees, but those equations are actually unidirectional assignments; attribute grammars compute values over trues, but do not constrain them.
%

\mypara{Unconstrainted FTAs, VSAs, and E-Graphs}
In contrast to the purely theoretical work on constrained tree automata,
their unconstrained counterparts, as well as VSAs and e-graphs,
have enjoyed practical applications
in program synthesis~\cite{wang2017synthesis,DBLP:journals/pacmpl/WangDS18,gulwani2011automating,polozov2015flashmeta,nandi2020synthesizing,ruler,willsey2020egg} 
and related areas, such as theorem proving \cite{detlefs2005simplify}, superoptimization~\cite{tensat}, and semantic code search~\cite{premtoon2020semantic}.
One important feature of these data structures, which ECTAs currently lack,
is the ability extract an optimal term according to a user-defined cost function.
It is not surprising that ECTAs have a slightly different focus,
since in the presence of constraints
extracting or enumerating terms regardless of cost becomes nontrivial---%
at least as hard as SAT solving.
Consequently, extracting optimal terms would be akin to MaxSAT solving~\cite{maxsat};
we leave this non-trivial extension to future work.

Finally note that unlike FTAs and VSAs, 
e-graphs are used to represent a \emph{congruence relation} over terms, 
as opposed to an arbitrary term space;
hence adding equality constraints to an e-graph is less meaningful.
Returning to our introductory example in \autoref{fig:ex1},
an e-graph equivalent to the FTA in \autoref{fig:ex1-fta} 
would actually encode that $\s{a}$, $\s{b}$, and $\s{c}$, are all \emph{equivalent} to each other;
hence it is hard to imagine why one would want to represent only the terms of the form $+(f(X), f(X))$
but not $+(f(X), f(Y))$,
because all these terms are equivalent.









\section{Conclusions and Future Work}\label{sec:conclusions}

This paper has introduced \emph{equality-constrained tree automata} (ECTAs)
and contributed an efficient implementation of this new data structure in the \ectalibrary library.
We think of ECTAs as a general-purpose language for expressing constraints over terms,
and the \ectalibrary library as a solver for these constraints.
Although in this paper we only discussed two concrete examples of properties that can be encoded with ECTAs---%
boolean satisfiability and well-typing---%
in the future we hope to see many fruitful applications in a wide range of domains.


\mypara{ECTA \vs SMT}
Instead of developing a custom solver for ECTAs,
wouldn't it be better to simply translate ECTAs into SAT or SMT constraints,
and use existing, well-engineered solvers?
A natural idea is to introduce a variable per ECTA node,
whose value represents the choice of incoming transition,
and to translate ECTA constraints into equalities between these variables.
This simple idea, however, does not work:
because the choice is made independently every time a node is visited,
this encoding would require unfolding the ECTA \emph{into a tree}.
This is a complete non-starter for cyclic ECTAs (like the \hpp encoding of \secref{sec:applications:hplus}), 
since the corresponding tree is infinite.
For acyclic ECTAs, the tree is finite but might be exponential in the size of the ECTA 
(since we need to ``un-share'' all the shared paths in the DAG).

More generally, the problem of finding an ECTA inhabitant is not in $\mathsf{NP}$, 
because the smallest tree represented by an ECTA can be exponential in the size of the ECTA
(as we illustrated in \autoref{fig:balanced-tree});
hence a general and efficient SAT encoding is not possible.
Although future work might develop a clever SMT encoding using advanced theories,
we believe this problem is far from trivial. 
After all, \hplus uses an SMT encoding that is specifically tailored to the type inhabitation problem (\ie it is less general than ECTA), 
and yet it is less efficient. 
As we discussed in \secref{sec:applications:hplus}, the main source of this inefficiency is polymorphism,
which makes the search space of types \emph{infinite} and precludes a ``one-shot'' SMT encoding,
requiring \hplus to go through a series of finite approximations of the space of types to consider.
Instead, a cyclic ECTA is able to represent the entire infinite space of types at once, 
and ECTA enumeration is able to explore this space efficiently, as our experiments show.
We anticipate that ECTAs will outperform SMT solvers on other similar problems
that require searching an \emph{infinite yet constrained space of terms}.

\mypara{Future Work}
One avenue for extension is to enrich the constraint language supported by ECTAs. The key ingredient for efficiency is that there exists a constraint-propagation mechanism that can be interleaved with \textsc{Choose}. Intersection is this constraint-propagation mechanism for equality, but there may be others. For example, one of our anonymous reviewers suggests that disequality constraints could be processed by creating an alternative rule to \textsc{Suspend} which tracks both sides of a disequality, and modifying \textsc{Choose} to discharge disequality constraints or propagate them into subterms as symbols are selected. 

Another path for extension is to relax the requirement for no constraints on cycles. The careful reader may notice that \autoref{thm:cyclic-undecidability} only impedes emptiness-checking; enumerating all satisfying terms up to a fixed size is trivially decidable. Currently \hpp creates many ECTA nodes for different term sizes, using a meta-program to iterate through successive ECTAs. With constraints on cycles, this meta-program could be internalized, shortening the \hpp implementation even further.

\bibliography{citations.bib}

\clearpage

\appendix

\section{Proofs and Additional Formalization}
\label{app:proofs}

\thmgenclosure*

\begin{proof}
  The reverse direction is trivial, since $\closure(C)$ is a superset of $C$.
  Let us focus on the forward direction: 
  that is, pick a $t$ that satisfies $C$ and show that it satisfies $\closure(C)$. 
  The proof is by induction on the number of steps,
  where each step constructs a new approximation of the closure, $C^{n+1}$, 
  by adding a new path $p''.p$ to a PEC $c_2\in C^n$ if
  $\exists c_1\in C^n$ such that $p',p''\in c_1$ and $p'.p\in c_2$ but $p''.p \notin c_2$.
  We need to show that if $\pecsat{C^n}{t}$ then also $\pecsat{C^{n+1}}{t}$.
  
  Assume $\pecsat{C^n}{t}$.
  Then there exists a unique $t'$ such that $t' = \at{t}{p'} = \at{t}{p''}$.
  But then $\at{t}{p'.p} = \at{t'}{p} = \at{t}{p''.p}$.
  Hence $t$ satisfies the new PEC $c_2' = c_2 \cup\{p''.p\}$, and hence satisfies $C^{n+1}$.
\end{proof}

\thmgenconsistencyfree*

\begin{proof}
It is easy to see that a non-prefix-free PEC is inconsistent.

For the reverse direction, we shall show that if each of the $c_i$ are prefix-free, then $C$ is consistent. Consider $c_1, c_2 \in C$, and define $c_1 \sqsubseteq c_2$ if $\exists p_1\in c_1, p_2\in c_2$ such that $p_1$ is a subpath of $p_2$. We show that $\sqsubseteq$ is a partial order. Reflexivity and transitivity are trivial. For anti-symmetry: suppose $p_1,p_3 \in c_1$ and $p_2, p_4 \in c_2$ with $p_1$ a subpath of $p_2$, $p_4$ a subpath of $p_3$. If $c_1=c_2$, we are done; otherwise, we must have $p_1\neq p_2$ and $p_3 \neq p_4$ by completeness of $C$. We can thus write $p_2=p_1.p'_2$, $p_3=p_4.p'_3$ for some $p'_2, p'_3 \neq \epsilon$. By completeness of $C$, it follows that $c_1$ contains $p_1=p_3=p_4.p'_3=p_2.p'_3=p_1.p'_2.p'_3$, which would mean that $c_1$ is not prefix-free. Thus, $\sqsubseteq$ is a partial order.

Because $C$ can be topologically sorted by $\sqsubseteq$, it is trivially possible to construct a term $t$ such that $C(t)$ by inductively choosing a term for each successive minimal $c_i$. Thus, if each of the $c_i$ are prefix-free, $C$  is consistent.
\end{proof}

\thmcompreduction*

\begin{proof}
  Let $e'=\reduce(e,c)$, let $p_i, p_j\in c$, 
  and suppose there is $n\in \nodesat(e', p_i)$ such that $n \intersect \at{e'}{p_j} = \nodebot$. 
  Let $n^*$ be as defined in \ref{defn:basic-static-reduction}. 
  Then $n \intersect n^* = \nodebot$ by associativity.
  Then $n \notin \nodesat(\intersectatpath{e}{p_i}{n^*}, p_i)$. 
  Because $c$ is prefix-free, no successive intersect-at-path operation in the definition of $\reduce(e, c)$ can alter this property, 
  and so $n \notin \reduce(e, c)$.
  Therefore, there is no such $n$ with $n \intersect \at{e'}{p_j} = \nodebot$, 
  and so $\reduce(e, c)$ satisfies the reduction criterion for $c$.
\end{proof}

\thmsoundreduction*  

\begin{proof}
  Let $c=\{p_1=\dots=p_k\}$. 
  Without loss of generality, it suffices to show that $\denedge{\intersectatpath{e}{p_1}{n^*}}=\denedge{e}$, 
  where $n^*$ is defined as in \autoref{defn:basic-static-reduction}. 
  We know from monotonicity that $\denedge{\intersectatpath{e}{p_i}{n^*}} \subseteq \denedge{e}$, 
  so we need only show $\denedge{\intersectatpath{e}{p_i}{n^*}} \supseteq \denedge{e}$.
  
  Consider a term $t\in \denedge{e}$. 
  Then, by definition, $\pecsat{c}{t}$, 
  and $\at{t}{c} \in \dennode{\at{e}{p_i}}$ for all $i$ by \autoref{thm:nodes-at-path-correctness}, 
  and so $\at{t}{c} \in \dennode{n^*}$ by the properties of intersection and the definition of $n^*$. 
  Since $\at{t}{c} \in \dennode{\at{e}{p_1}}$, so $\at{t}{p_1} \in \dennode{n^*}$.
  Hence $t \in \denedge{\intersectatpath{e}{p_i}{n^*}}$, 
  so $\denedge{\intersectatpath{e}{p_i}{n^*}} \supseteq \denedge{e}$.
\end{proof}

\subsection{Undecidable ECTAs}
\label{app:proofs:undecidable}

\begin{figure}
    \centering
    \includegraphics[scale=0.3]{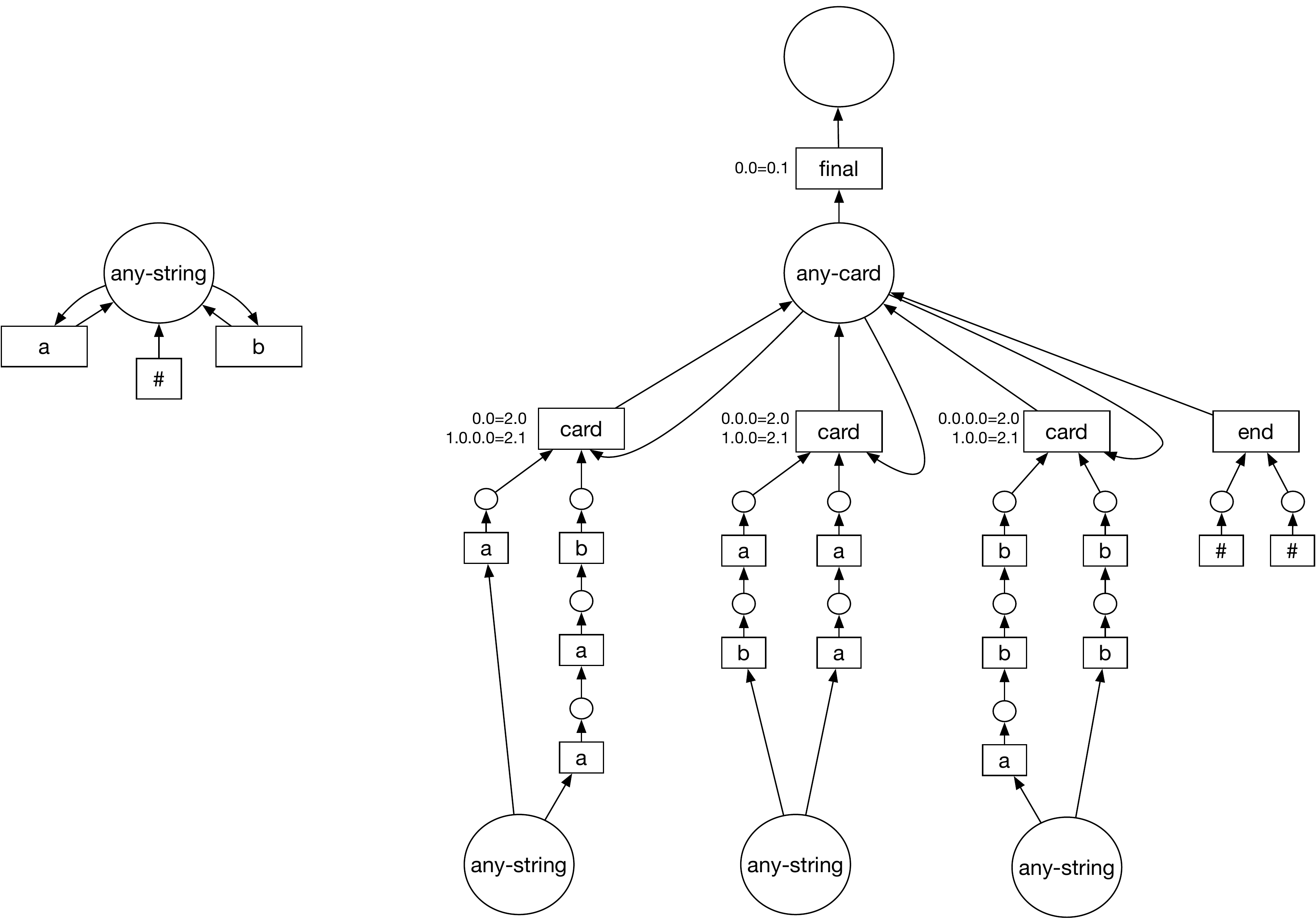}
    \caption{Example encoding of the PCP problem into an ECTA with constraints on cycles}
    \label{fig:undecadibility-proof}
\end{figure}

\thmgenundecidable*

\begin{proof}
By reduction from the Post Correspondence Problem. Consider the instance of the PCP problem with three cards $(a,baa), (ab,aa),$ and $(bba, bb)$. Any term represented by the ECTA in \autoref{fig:undecadibility-proof} encodes either a solution to this instance or is the empty solution $\text{final}(\text{end}(\#,\#))$.  We leave it as an exercise to the reader to generalize this construction to other instances, and to unroll this ECTA so as to exclude the empty solution.
\end{proof}

\subsection{Enumeration State}
\label{app:proofs:enumeration-state}

\mypara{Denotation}
Defining the denotation of a partially-enumerated term is tricky because of a circular reference: enumeration states contain partially enumerated terms, which in turn can reference other parts of the enumeration state.

We choose to break this circularity by giving an \emph{impredicative} definition. Each of the cases in the definition is presented as a \emph{filter} on the space of all possible substitutions. The denotation of a partially-enumerated term is a set of pairs $(t, \rho)$, where the $t$ range over all all terms it may represent, and the $\rho$ range over all possible substitutions $\rho$ compatible with this choice of term. Intersecting these filtered sets of all possible substitutions, and then restricting the substitutions to the set of variables actually used (written $\restrict{\rho}{\text{dom}(\sigma)}$), gives the desired set of substitutions compatible with the constraints of $\sigma$.

\begin{definition}[Denotation of Enumeration State]

\begin{align*}
    \denpt{\cdot}_{\sigma} &\hastype \tau \times (\set{Var} \partialfn \tau) \rightarrow (\P{\closedtermsof{\Sigma}} \times (\set{Var} \rightarrow \tau)) \\
    \denpt{s(\overline{\tau_i})}_{\sigma} &= \left\{ (s(\overline{t_i}), \rho) \bigmid (t_i, \rho_i) \in \denpt{\tau_i}_{\sigma}, \rho\text{ arbitrary}, \forall j\forall v, \rho(v)=\rho_j(v) \right\} \\
    %
    %
    \denpt{\unenum{n}{\many{\fragment{c_i}{v_i}}}}_{\sigma} &= \left\{ (t, \rho) \bigmid \rho\text{ arbitrary}, t\in \dennode{n}, \pecsat{\many{\fragment{c_i}{v_i}}}{t}, \at{t}{c_i}=\rho(v_i)  \right\} \\
    %
    %
    \denotation{v}^{PT} &= \left\{ (t, \rho) \bigmid (t, \rho) \in \denpt{\sigma[v]}, \rho(v)=t \right\}\\
    %
    %
    ~\\
    \denstate{\sigma} &\hastype (\set{Var} \partialfn \tau) \rightarrow (\set{Var} \to \P{\closedtermsof{\Sigma}}) \\
    \denstate{\sigma} &= \left\{ \restrict{\rho}{\text{dom}(\sigma)} \bigmid  (t_v, \rho) \in \denotation{\sigma(v)}^{PT}, \rho(v) = t_v \right\}
\end{align*}

\end{definition}

\begin{example}
Let $v_\top=v_0$ and consider the enumeration state

\begin{align*}
\sigma = [v_0 &\mapsto f(\unenum{\node{\edgeucnull{a}, \edgeucnull{b}}}{\fragment{\epsilon}{v_1}}, v_1), \\
v_1 &\mapsto \unenumuc{\node{\edgeucnull{a}, \edgeucnull{b}}}]
\end{align*}

Then 

\[
  \denotation{\sigma(v_1)}^{PT} = \bigcup_{\rho}\left\{(a, \rho), (b, \rho)\right\}
\]
where the union is over all $\rho \hastype \set{Var} \rightarrow \set{Term}$, and
\[
  \denotation{\sigma(v_0)}^{PT} = \big(\bigcup_{\rho \big| \rho(v_1) = a} \{(f(a, \rho(v_1)), \rho)\} \big) \bigcup \big(\bigcup_{\rho \big| \rho(v_1) = b} (f(b, \rho(v_1)), \rho)\}\big)
\]
which reduces to the terms $f(a, a)$ and $f(b, b)$ with $\rho(v_1)$ being $a$ or $b$ accordingly. This yields the final denotation
\[
  \denstate{\sigma} = \left\{ [v_0 \mapsto f(a, a), v_1 \mapsto a],  [v_0 \mapsto f(b, b), v_1 \mapsto b] \right\}
\]
\end{example}

\subsection{Enumeration Algorithm}
\label{app:proofs:enumeration}

\begin{definition}[Reachable Enumeration State]
Enumeration state $\sigma$ is a \emph{reachable enumeration state} if there is an initial state  $\sigma_0=[v_\top \mapsto \unenumuc{n}]$ such that $\reduces{\sigma_0}{\sigma}$.
\end{definition}

\begin{lemma}[Soundness of \textsc{Choose-$\square$}]\label{lemma:soundness-choose-sq}
Let $\tau$ be a u-node, and define $T=\left\{\tau' \middle| \steps{\tau}{\tau'}\text{ by \textsc{Choose-}}\square \right\}$. Then $\denpt{\tau} =\bigcup_{\tau' \in T} \denpt{\tau'}$
\end{lemma}

\begin{proof}

By rule \textsc{Choose-}$\square$ and the definition of p-term denotation, 
we have $\tau = \unenum{\node{\many{e}}}{\many{\fragment{c_i}{v_i}}}$,
and hence 
$$\termsofunode{\tau}=\{ t \mid t\in \dennode{\node{\many{e}}}, \pecsat{\many{\fragment{c_i}{v_i}}}{t} \}
$$
By the definition of node denotation, we have that
$$
\termsofunode{\tau}=\{ t \mid t \in \bigcup_{\edge{s}{\many{n^j}}{c} \in \many{e}} \denedge{\edge{s}{\many{n^j}}{c}},  \pecsat{\many{\fragment{c_i}{v_i}}}{t} \}
$$
Then we can apply the definition of edge denotation and copy the constraints $\many{\fragment{c_i}{v_i}}$ into each edge to get:
$$
\termsofunode{\tau}=\{ t \mid t\in \bigcup_{\edge{s}{\many{n^j}}{c} \in \many{e}} \{s(\many{t^j}) \mid t^j \in \dennode{n^j}, \pecsat{\fragment{c}{v} \cup \many{\fragment{c_i}{v_i}}}{s(\many{t^j})}\} \}
$$
Meanwhile, from the first and third premise of rule \textsc{Choose-}$\square$, we have that 
$$
T = \{s(\many{\tau_j}) \mid \edge{s}{\many{n^j}}{c} \in \many{e}, \tau_j = \unenum{n^j}{\proj(\fragment{c}{v}\cup\many{\fragment{c_i}{v_i}}, j)}\}
$$
Therefore, for each $\tau' \in T$,
$$
\termsofunode{\tau'} = \termsofunode{s(\many{\tau_j})} = \{s(\many{t^j}) \mid t^j \in \dennode{n^j}, \pecsat{\proj(\fragment{c}{v} \cup \many{\fragment{c_i}{v_i}}, j)}{t^j}\}
$$
It is obvious that, given $s(\many{t^j}), \fragment{c}{v}$, $\many{\fragment{c_i}{v_i}}$ and $\fragment{\epsilon}{\_}\not\in\many{\fragment{c_i}{v_i}}$, the constraints $\pecsat{\proj(\fragment{c}{v} \cup \many{\fragment{c_i}{v_i}}, j)}{t^j}$ and $\pecsat{\fragment{c}{v} \cup \many{\fragment{c_i}{v_i}}}{s(\many{t^j})}$ are equivalent.
Hence, the lemma holds.
\end{proof}

\begin{lemma}[Soundness of \textsc{Choose}]\label{lemma:soundness-choose}
$\sigma$ be a reachable enumeration state, and fix a $v\in\dom(\sigma)$ and a $\ctx,\tau$  such that $\sigma[v] = \ctx[\tau]$. Define $S=\left\{\sigma[v\mapsto\ctx[\tau']] \middle| \steps{\sigma}{\sigma[v\mapsto\ctx[\tau']]}\text{ by \textsc{Choose}}\right\}$. Then $\denstate{\sigma}=\bigcup_{\sigma' \in S} \denstate{\sigma'}$.
\end{lemma}
\begin{proof}
As $\sigma(v) = \ctx[\tau]$ and
$\steps{\tau}{\tau'}$ by the second premise of rule \textsc{Choose}, 
we have that $\denpt{\tau} = \bigcup_{\sigma' \in S, \sigma'(v) = \ctx[\tau']} \denpt{\tau'}$ by \autoref{lemma:soundness-choose-sq}.
Under a fixed context $\ctx$, we also have $\denpt{\ctx[\tau]} = \bigcup_{\sigma' \in S} \denpt{\ctx[\tau']}$,
which is $\denpt{\sigma[v]} = \bigcup_{\sigma' \in S} \denpt{\sigma'[v]}$.
Hence, we can get that 
\[
\begin{aligned}
\denstate{\sigma} &= \left\{{\restrict{\rho}{\dom(\sigma)}} \bigmid (t_v, \rho) \in \denpt{\sigma[v]}, \rho(v) = t_v \right\} \\
&= \left\{{\restrict{\rho}{\dom(\sigma)}} \bigmid (t_v, \rho) \in \bigcup_{\sigma' \in S} \denpt{\sigma'[v]}, \rho(v) = t_v \right\} \\
 &= \bigcup_{\sigma' \in S} \left\{{\restrict{\rho}{\dom(\sigma)}} \bigmid (t_v, \rho) \in \denpt{\sigma'[v]}, \rho(v) = t_v \right\} \\
 &= \bigcup_{\sigma' \in S} \denstate{\sigma'}
\end{aligned}
\]
\end{proof}

\begin{lemma}[Soundness of \textsc{Suspend}]\label{lemma:soundness-suspend}
Let $\sigma$ be a reachable enumeration state, and suppose $\steps{\sigma}{\sigma'}$ by either the \textsc{Suspend-1}, \textsc{Suspend-2}, or \textsc{Subst} rules. Then $\denstate{\sigma} = \left\{\restrict{\rho}{\dom(\sigma)} \bigmid \rho \in \denstate{\sigma’}\right\}$.
\end{lemma}
\begin{proof}
To prove this lemma, we need to show that $\bigcup_{v \in \dom(\sigma)}\denpt{\sigma[v]} = \bigcup_{v \in \dom(\sigma)}\denpt{\sigma'[v]}$.
If $\sigma[v] = \sigma'[v]$, we immediately get $\denpt{\sigma[v]} = \denpt{\sigma'[v]}$ and the lemma holds.
If $\sigma[v] \not= \sigma'[v]$, we split it into three cases:

\mypara{Case I: \textsc{Suspend-1}}
By rule \textsc{Suspend-1}, we have
$\sigma[v] = \ctx[\unenum{n}{\fragment{\epsilon}{v'} \cup \Phi}]$,
and hence $\denpt{\sigma[v]} = \denpt{\ctx[\unenum{n}{\fragment{\epsilon}{v'} \cup \Phi}]} $.
%
Meanwhile,
we have that $\sigma'[v] = \ctx[v'], \sigma'[v'] = \unenum{n}{\Phi}$,
and hence $\denpt{\sigma'[v]} = \denpt{\ctx[v']} = \denpt{\ctx[\unenum{n}{\Phi}]} $ 
by the denotation of p-terms and nodes.
As we know that $\sigma'[v'] = \unenum{n}{\Phi}$, $\fragment{\epsilon}{v'}$ is trivially satisfied.
Hence, we conclude that $\denpt{\sigma[v]} = \denpt{\sigma'[v]}$.

\mypara{Case II: \textsc{Suspend-2}}
By rule \textsc{Suspend-2}, we have
$\denpt{\sigma[v]} = \denpt{\ctx[\unenum{n}{\fragment{\epsilon}{v'} \cup \Phi}]} $
and there exists another $v' \in \dom(\sigma)$ such that $\sigma[v'] = \unenum{n}{\Phi}$.
On the right hand side, 
$\denpt{\sigma'[v]} = \denpt{\ctx[\unenum{n \intersect n'}{\Phi \cup \Phi'}}] $.
By monotonicity, we have $\denpt{\sigma'[v]} \subseteq \denpt{\sigma[v]}, \denpt{\sigma'[v']} \subseteq \denpt{\sigma[v']}$.
So we get $\bigcup_{v \in \dom(\sigma)}\denpt{\sigma'[v]} \subseteq \bigcup_{v \in \dom(\sigma)}\denpt{\sigma[v]}$.

Next, we show $\bigcup_{v \in \dom(\sigma)}\denpt{\sigma[v]} \subseteq \bigcup_{v \in \dom(\sigma)}\denpt{\sigma'[v]}$.
Consider a pair $(\ctx[t], \rho) \in \denpt{\ctx[\unenum{n}{\fragment{\epsilon}{v'} \cup \Phi}]}$ and another pair $(t', \rho') \in \denpt{\unenum{n'}{\Phi'}}$.
By the first two premises of rule \textsc{Suspend-2},
we have $t=v'$

\mypara{Case III: \textsc{Subst}}
Trivial because this rule only propagates equality constraints and merges variables. 

\end{proof}

\begin{lemma}[Completeness of \textsc{Choose} and \textsc{Suspend}]
\label{lemma:rule-completeness}
Let $\sigma$ be a reachable enumeration state, and suppose $\sigma$ is irreducible. Then either $\sigma$ is fully enumerated, or $\denstate{\sigma}=\emptyset$.
\end{lemma}
\begin{proof}
Suppose $\sigma$ is not fully enumerated. Then, by definition, there must be some $v$ such that $\sigma[v]$ contains a subterm of the form $\tau=\unenum{n}{\Phi}$. We proceed by cases:

\begin{enumerate}
    \item There is a $v'$ such that $\fragment{\epsilon}{v'}\in \Phi$: If $v'\notin \dom{\sigma}$, then \textsc{Suspend-1} applies. Otherwise, $v'$ is not solved, and so therefore we must have $\sigma[v']=\unenum{n'}{\Phi'}$ for some $n',\Phi'$ because $\sigma$ is reachable, and \textsc{Suspend-2} applies.
    \item There are $v_1, v_2$ such that $\sigma[v_2]=v_1$. Then the \textsc{Subst} rule applies.
    \item Otherwise, one of the following subcases applies:
    \begin{enumerate}
        \item \textbf{(Choose case)} If $v$ is solved, then either \textsc{Choose} applies, or $n=\nodebot$ and thus $\denstate{\sigma}=\emptyset$.
        \item \textbf{(Occurs-check failure case)} Suppose $v$ is unsolved, and, furthermore, there is no alternate solved $v$ containing a u-node. Then there is a sequence of variables $v_1, \dots, v_k$  and $v_{k+1}=v_1$ such that $\sigma(v_i)$ references $\sigma(v_{i+1})$ for $1 \le i \le k$. Each $v_i$ has a corresponding path $p_i\neq\epsilon$ (because none of the previous cases applied) such that $\rho \in \denstate{\sigma}$ must satisfy $\at{\rho(v_i)}{p_i}=\rho(v_{i+1})$ for $1 \le i \le k$. Since all terms are finite, this is unsatisfiable, and so $\denstate{\sigma}=\emptyset$.
    \end{enumerate}
    
\end{enumerate}
\end{proof}

\themterminationenuremation*
\begin{proof}
For paths $p_1, p_2$ let $p_1 \sqsubseteq p_2$ if $p_1=p_2$ or $\text{length}(p_1) < \text{length}(p_2)$, where $\text{length}(p)$ denotes normal list length. For ECTA nodes $n_1, n_2$, let $n_1 \sqsubseteq n_2$ if $n_1=n_2$ or $\text{depth}(n_1) < \text{depth}(d_2)$, where $\text{depth}(n)$ denotes normal DAG depth. Order integers normally. Order multisets of paths, ECTA nodes, and integers by their respective multiset orderings \cite{dershowitz1987termination}. All counts regard the ECTA as a tree, not as a graph.

We define the height of a node to be the length of the longest path from that node to a leaf transition. Now define $F(\sigma)=(A,B,C)$, where $A$ is the multiset of heights of u-nodes, $B$ is the multiset of paths contained in $\sigma$, and $C$ the number of pairs of variables $v_1, v_2$ such that $\sigma[v_2]=v_1$. We shall show that each of the rules decrease $F(\sigma)$ under the lexicographic ordering. Since this ordering is well-founded by construction, that completes the proof of termination.

\begin{enumerate}
    \item \textsc{Choose}: Decreases $A$
    \item \textsc{Suspend-1}: Leaves $A$ unchanged, decreases $B$.
    \item \textsc{Suspend-2}: Decreases $A$ (by decreasing the number of u-nodes).
    \item \textsc{Subst}: Leaves $A$ and $B$ unchanged, decreases $C$.
\end{enumerate}
\end{proof}

\thmcorrectnessenumeration*
\begin{proof}
Given an ECTA node $n$, the enumeration algorithm first constructs the initial state $\sigma_0=[v_\top \mapsto \unenumuc{n}]$. It is trivial from the definitions that $\denstate{\sigma_0}=\left\{[v_\top \mapsto t] \bigmid t \in \denotation{n}^\text{N} \right\}$, i.e.: that $\denstate{\sigma_0}$ is equivalent to $\denotation{n}^\text{N}$. The theorem then follows immediately from Lemmas \ref{lemma:soundness-choose-sq}, \ref{lemma:soundness-choose},  \ref{lemma:soundness-suspend}, and \ref{lemma:rule-completeness} and Theorem \ref{thm:enumeration-termination}
\end{proof}

\section{Optimizations and Extra Implementation Notes}
\label{app:optimization}

\mypara{Intersection}
The definition of $n_1 \intersect n_2$ in \secref{sec:acyclic} effectively encodes what is known in the databases literature as a \emph{nested-loop join},  performing $|\many{e_1}| * |\many{e_2}|$ transition intersections.
The \ectalibrary library instead performs a more efficient kind of join, known as \emph{simple hash join}.  Further, it replaces $C_1 \cup C_2$ with $\closure(C_1 \cup C_2)$ when intersecting. 

\subsection{Optimized Reduction}
\label{app:optimized-reduction}

For the example of \autoref{fig:overview-ex1}, 
instead of computing a single shared automaton $n^*$,
an alternative is to compute two separate $n^*_1 = \at{n}{\texttt{arg.type}}$, $n^*_2 = \at{n}{\texttt{fun.targ}}$, 
and return $\intersectatpath{\left(\intersectatpath{n}{\texttt{fun.targ}}{n^*_1}\right)}{\texttt{arg.type}}{n^*_2}$. 
At a first glance, this seems less efficient, since it requires computing multiple $n^*_i$.
Surprisingly, our experience shows that this alternative reduction algorithm actually performs better:
the reason is that intersecting nodes with a single shared $n^*$ creates a lot of redundant transitions---%
transitions $e_1$ and $e_2$ into the same node where $\denedge{e_1} \subseteq \denedge{e_2}$.
We now detail this alternative reduction algorithm, on which \ectalibrary's implementation is actually based. It produces a semantically-equivalent ECTA, but with fewer junk edges.

\begin{definition}[Optimized Reduction]

Let PEC $c=\{p_1=\dots=p_k\}$. Let $c_{\overline{i}}=\{p_1, \dots, p_{i-1}, p_{i+1}, \dots, p_k\}$. Let $e$ be a transition, and define

\[
  m_i = \bigintersect_{p\in c_{\overline{i}}} \at{e}{p}
\]

Then define

\begin{align*}
    \text{reduce}(e, c) = \intersectatpath{\intersectatpath{e}{p_1}{m_1}\dots}{p_k}{m_k}
\end{align*}

\end{definition}

Consider running the optimized reduction algorithm to reduce the root $w$ transition of \figref{fig:ex-junk-term}. This involves intersecting each of the two transitions under $q_a$ by $q_b\intersect q_c$, each transition under $q_b$ by $q_a \intersect q_c$, and each transition under $q_c$ by $q_a \intersect q_b$. It would be semantically equivalent to intersect all three of them instead by $q_a \intersect q_b  \intersect q_c$, as is done by the basic reduction algorithm. Counterintuitively, even though this reduces the number of distinct intersections that must be performed, doing so actually degrades performance.

\begin{wrapfigure}{r}{0.40\textwidth}
    \centering
    \includegraphics[scale=0.4]{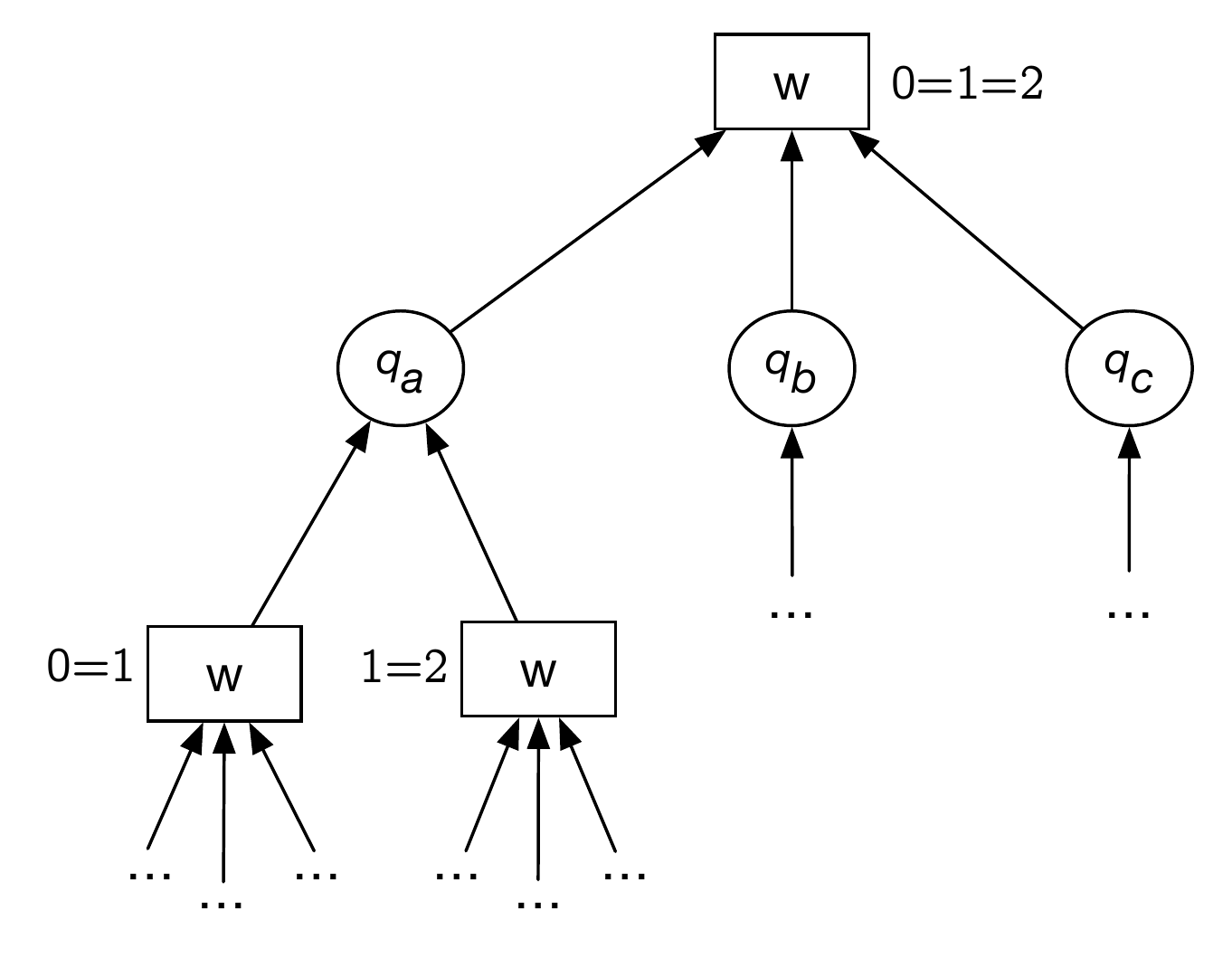}
    \caption{}
    \label{fig:ex-junk-term}
\end{wrapfigure}

The reason is that intersecting $q_a \intersect q_a$ produces three child transitions: one with constraint $\{0=1\}$, one with constraint $\{1=2\}$, and a ``junk transition" with constraint $\{0=1=2\}$. Although this last transition is redundant with the first two, detecting and eliminating these redundant transition is actually rather expensive. Particularly for the type-driven synthesis domain (\secref{sec:applications:hplus}), where a large proportion of transitions have the same symbol (the $\rightarrow$ function type symbol) with different constraints, the creation of a redundant transition can proliferate into many more redundant transitions upon reducing other constraints. Thus, \ectalibrary takes the ``slow" route of computing $q_a \intersect q_b$, $q_a \intersect q_c$, and $q_b \intersect q_c$ separately.


\subsection{Enumeration}
\label{app:optimization:enumeration}

\begin{remark}
The requirement that \textsc{Choose} only runs on solved variables is sometimes too strict. For example, it means that, in the SAT domain, a value cannot be chosen for a variable until all clauses containing that variable have been selected. But the requirement can be relaxed by replacing the intersection of ECTA nodes with a compatibility check between partially-enumerated terms.
\end{remark}

\begin{remark}
The \ectalibrary implementation does all available \textsc{Suspend} and \textsc{Subst} steps on a node simultaneously. It uses a union-find data structure to track variables that have been merged, avoiding actually performing many substitutions.
\end{remark}

\begin{remark}
\label{remark:ordering-performance}
The intersections in \textsc{Subst-$2$} may result in having $\sigma[v]=\unenum{\nodebot}{\Phi}$ for some $v$. Its denotation is hence $\emptyset$, and this entire branch of computation may be pruned. This is the reason that node ordering affects performance.
\end{remark}

\subsection{Cyclic ECTAs}

The \ectalibrary implementation uses a variant of the structured graph representation  \cite{oliveira2012functional}. Recursive nodes are created and accessed as higher-order abstract syntax, but nodes are stored in a first-order fashion. This representation is efficient and allows for compression and fast equality-checks via hash-consing.

\end{document}